\documentclass[11pt,a4paper]{article}

\usepackage{a4wide}
\usepackage{amsmath, amsthm, amssymb, amsfonts, thm-restate,mathtools}
\usepackage{hyperref}
\usepackage{enumerate}
\usepackage{enumitem}
\usepackage{xparse}
\usepackage{todonotes}
\usepackage{tikz}
\usepackage{bbm}
\usepackage{cite}

\usepackage{graphicx}
\usepackage{euscript}
\usepackage{multirow}
\usetikzlibrary{calc}
\usetikzlibrary{decorations.markings}
\usetikzlibrary{arrows.meta, shapes.misc, positioning,patterns}
\usetikzlibrary{snakes}
\usepackage{microtype}
\usepackage{color,colortbl}
\usepackage[Algorithm, section]{algorithm}
\usepackage[noend]{algpseudocode}
 
\usepackage[stable]{footmisc}
\definecolor{darkgrey}{gray}{0.35}
\definecolor{grey}{gray}{0.86}
\definecolor{lightgrey}{gray}{0.91}
\definecolor{ballblue}{rgb}{0, 0.5,0.5}
\definecolor{lightballblue}{rgb}{0, 0.8,0.8}
\definecolor{dbblue}{rgb}{0, 0.4,0.4}
\usepackage{caption}
\usepackage{subcaption}

\newtheorem{theorem}{Theorem}[section]
\newtheorem{lemma}[theorem]{Lemma}

\newtheorem{proposition}[theorem]{Proposition}
\newtheorem{observation}[theorem]{Observation}
\newtheorem{corollary}[theorem]{Corollary}

\newtheorem{definition}[theorem]{Definition}
\theoremstyle{definition}
\newtheorem{remark}[theorem]{Remark}

\newcommand{\N}{\mathbb{N}}

\DeclareMathOperator*{\argmin}{arg\,min}

\newcommand{\pr}{\mathbb{P}}
\renewcommand{\Pr}{\mathbb{P}}
\newcommand{\E}{\mathbb{E}}

\newcommand{\var}{\textnormal{Var}}

\newcommand{\eps}{\varepsilon}

\title{Balanced Bidirectional Breadth-First Search on Scale-Free Networks}
\author{Sacha Cerf\thanks{Ecole Polytechnique, France. Email: \href{mailto:sacha.cerf@polytechnique.org}{\nolinkurl{sacha.cerf@polytechnique.org}}. Research conducted while on a research visit at ETH Zürich.}
		\and
		Benjamin Dayan\thanks{Institut für Theoretische Informatik, ETH Zürich, Zürich, Switzerland. Email: \href{mailto:benjamin.l.dayan@gmail.com}{\nolinkurl{benjamin.l.dayan@gmail.com}}.}
        \and
		Umberto De Ambroggio\thanks{Department of Mathematics, National University of Singapore, S17, 10 Lower Kent Ridge Road Singapore, 119076. Email: \href{mailto:umberto@nus.edu.sg}{\nolinkurl{umberto@nus.edu.sg}}. Research mostly supported by ERC Grant Agreement 772606-PTRCSP.}
        \and
		Marc Kaufmann\thanks{Institut für Theoretische Informatik, ETH Zürich, Zürich, Switzerland. Email: \href{mailto:marc.kaufmann@inf.ethz.ch}{\nolinkurl{marc.kaufmann@inf.ethz.ch}}. The author gratefully acknowledges support by the Swiss National Science Foundation [grant number 200021\_192079].}
		\and
		Johannes Lengler\thanks{Institut für Theoretische Informatik, ETH Zürich, Zürich, Switzerland. Email: \href{mailto:johannes.lengler@inf.ethz.ch}{\nolinkurl{johannes.lengler@inf.ethz.ch}}. The author gratefully acknowledges support by the Swiss National Science Foundation [grant number 200021\_192079].}
		\and
		Ulysse Schaller\thanks{Institut für Theoretische Informatik, ETH Zürich, Zürich, Switzerland. Email: \href{mailto:ulysse.schaller@inf.ethz.ch} {\nolinkurl{ulysse.schaller@inf.ethz.ch}}. The author gratefully acknowledges support by the Swiss National Science Foundation [grant number 200021\_192079].}
	}

\begin{document}

\maketitle

\begin{abstract}
To find a shortest path between two nodes $s_0$ and $s_1$ in a given graph, a classical approach is to start a Breadth-First Search (BFS) from $s_0$ and run it until the search discovers $s_1$. Alternatively, one can start two Breadth-First Searches, one from $s_0$ and one from $s_1$, and alternate their layer expansions until they meet. This bidirectional BFS can be balanced by always expanding a layer on the side that has discovered fewer vertices so far. This usually results in significant speedups in real-world networks, and it has been shown that this indeed yields sublinear running time on scale-free graph models such as Chung-Lu graphs and hyperbolic random graphs.

We improve this layer-balanced bidirectional BFS approach by using a finer balancing technique. Instead of comparing the size of the two BFS trees after each layer expansion, we perform this comparison after each vertex expansion. This gives rise to two algorithms that run faster than the layer-balanced bidirectional BFS on scale-free networks with power-law exponent $\tau\in (2,3)$. The first one is an \emph{approximate} shortest-path algorithm that outputs a path of length at most 1 longer than the shortest path in time $n^{(\tau-2)/(\tau-1)+o(1)}$. The second one is an \emph{exact} shortest-path algorithm running in time $n^{1/2+o(1)}$. These runtime bounds hold with high probability when $s_0$ and $s_1$ are chosen uniformly at random among the $n$ vertices of the graph. We also develop an edge-balanced bidirectional BFS algorithm that works under adversarial conditions. This approximate shortest-path algorithm runs in time $n^{1/2+o(1)}$ with high probability when the adversary is allowed to choose $s_0$ and $s_1$ based on their (expected) degree. We complement our theoretical results with experiments on Chung-Lu graphs, Geometric Inhomogeneous Random Graphs, and real-world networks.
\end{abstract}

\section{Introduction}\label{sec:intro}

The study of Breadth-First Search (BFS) algorithms goes back at least to 1959, when Moore published his paper ``The shortest path through a maze"~\cite{moore1959shortest}. Already ten years later, early works investigated the idea of bidirectional search algorithms~\cite{pohl1969bi, pohl1971bi}. The idea is simple: Instead of exploring the graph (layer-by-layer) just from one side in order to find a shortest path from a starting vertex $s_0$ to a target vertex $s_1$, we run two searches in parallel, a ``forward" search from $s_0$ and a ``backward" search from $s_1$. The algorithm then alternates between the two sides until their search trees intersect. This heuristic does not give a speed-up in adversarially chosen networks such as cycles, but it is known to work well in practice. On many practical instances of networks with homogeneous degrees, such as road networks, the speed-up is modest~\cite{blasius2024external} and seems to be only a constant factor. However, when degree distributions are more heterogeneous the runtime is decreased by much more than a constant factor~\cite{borassi2019kadabra}. This is particularly valuable in instances where BFS is used as a subroutine, such as the computation of betweenness centrality~\cite{borassi2019kadabra}. 

On sparse graphs, how fast (specific variants of) bidirectional BFS run is influenced by both the presence of underlying geometry and the degree distribution. When edges are independent, i.e.\ there is no underlying geometry, and the degree distribution has bounded variance (e.g.\ Erdös-Rényi random graphs), bidirectional search terminates in time $n^{\frac{1}{2}+o(1)}$. When the vertex degrees follow a power law with exponent $2<\tau<3$, the best known runtime bound is $n^{\frac{4-\tau}{2}+o(1)}$ (cf.\ Table 1 in~\cite{blasius2022efficient}). Similarly, we observe a speed-up in the geometric setting when transitioning from graphs with bounded-variance degree distribution, such as Euclidean Random Graphs, to Hyperbolic Random Graphs. In the Euclidean case, it is known that bidirectional BFS takes time $\Theta(n)$, whereas Bläsius et al.\ have shown that in the hyperbolic setting, this takes time at most $n^{\max\{2\frac{\tau-2}{\tau-1}, \frac{1}{\tau-1}\}+o(1)}$ with high probability and in expectation (Theorem 3.1 in~\cite{blasius2022efficient}). Furthermore, they show that their variant of bidirectional BFS, exploring greedily layer-by-layer, is at most by a factor $d$ slower, where $d$ is the graph diameter, than any other layer-alternating bidirectional BFS (Theorem 3.2 in~\cite{blasius2022efficient}).

Bidirectional search is an example where classical worst-case analysis is not sufficient to explain performance in practice. Bringing theory and application closer together has been a key driver of network science. One such approach analyzes deterministic properties amenable to theoretical study that can be checked on any particular (graph) instance, including real-world networks. A recent study investigating bidirectional BFS managed to show sublinear runtimes when for both searches the exploration cost grows exponentially for a logarithmic number of steps, using the notion of \emph{expansion overlap}~\cite{blasius2023deterministic}. 

A second approach studies the runtime in a setting where host graph instances are sampled from a specific probability distribution, such as various kinds of scale-free network models like Chung-Lu graphs and Geometric Inhomogeneous random graphs (GIRGs). Experiments have shown that the performance of graph algorithms, including bidirectional BFS\footnote{The alternation strategy that was studied there greedily applies the
cheaper of the two explorations in every step. The cost of exploring a layer is estimated
via the sum of degrees of vertices in that layer.} , on synthetic network models translates surprisingly well to real-world networks~\cite{blasius2024external}. More concretely, the experiments there showed that on networks with high locality\footnote{The notion of locality is related to the clustering coefficient. For a precise definition, see~\cite{blasius2024external}.} and homogeneous degree distribution, the searches took roughly time $m$. For heterogeneous or less local networks, they observed significantly lower costs of $m^{\frac{1}{2}}$ or even less. The cost was particularly low for very heterogeneous networks. 

We continue this line of research that aims to bring theory closer to application, introducing a set of three bidirectional BFS algorithms with a modified alternation strategy. One algorithm that we present is exact, two of them are approximate, namely the paths they output are at most one hop longer than the graph distance between the vertices they connect. We show that on Chung-Lu graphs and GIRGs, all three of them yield a polynomial speedup over standard bidirectional BFS. One of the approximate algorithms is designed to be robust against adversaries that may pick the two source vertices $s_0, s_1$ by knowing their weights (i.e.\ their expected degrees), before the graph is unveiled. For the graph models we study, finding shortest paths is particularly relevant, since for instance greedy routing occurs along these shortest paths~\cite{bringmann2017greedy}.

Chung-Lu graphs and GIRGs are particularly suited for the runtime analysis of BFS. Both models are sparse, and vertex degrees follow a power law with exponent $\tau\in (2,3)$, i.e.\ the probability of a vertex having degree $k$ decays proportionally to $k^{-\tau}$, a phenomenon observed universally across various real-world network types~\cite{newman2003structure}. Like many real-world networks, both are also (ultra)small worlds, that is, they have polylogarithmically bounded diameter and log-log average distances~\cite{bringmann2016average}. In addition, GIRGs exhibit clustering, a typical feature of social networks~\cite{bringmann2019geometric}. Recent research aiming to fit GIRGs to real networks has shown that they are well-suited to model additional geometric network features such as closeness and betweenness centrality~\cite{dayan2024expressivity}. The increased interest in these models has sparked extensions to non-metric geometries and the investigation of various spreading processes on GIRGs~\cite{kaufmann2024rumour, kaufmann2024sublinear, lengler2017existence, komjathy2023four, komjathy2023four2}. From an analysis perspective, in order to understand the behavior of balanced BFS on a graph, it is crucial to understand the sizes of neighborhoods. In the case of the thoroughly researched degree distributions of Chung Lu graphs and GIRGs, these neighborhood sizes have far-reaching implications for both the structure of the BFS layers and, consequently, the runtime. Owing to the power-law exponent $\tau$ being between 2 and 3, despite the graphs being sparse, the largest-degree neighbor of a vertex $v$ has degree polynomially larger (in the degree of $v$) than $v$ itself. Hence, the maximal degree encountered by the search increases rapidly with the number of expanded vertices. Since the bidirectional search process is balanced, on both sides we quickly reach a neighbor of a vertex whose degree is of maximal order in the graph, namely $n^{\frac{1}{\tau-1}\pm o(1)}$, at which point the two search trees meet.

Before we describe our results in detail, we remark that all our algorithms are deterministically correct regardless of the host graph structure. Only our runtime analyses draw on properties of the graph models. We first devise a pair of vertex-balanced algorithms, where each search expands vertex-by-vertex, exploring one neighborhood at a time. The searches are balanced by comparing, in each iteration, which of the two searches has ``discovered" more vertices, and then proceeding with the one that has discovered fewer vertices so far. The order in which the vertices are processed on one side is still layer-by-layer, and within a layer the order is chosen uniformly at random, see Algorithm~\ref{algo:vertex-approx}. For the runtimes of the different algorithms, it will come in handy throughout the paper, to recall the following ordering of exponents, which holds for $\tau \in (2,3)$:
\begin{align*}
    \frac{\tau-2}{\tau-1}<\frac{1}{2}<\frac{1}{\tau-1}<1.
\end{align*}
One can show that the search trees intersect after time at most $n^{\frac{\tau-2}{\tau-1}+o(1)}$ with high probability when the vertices $s_0$, $s_1$ are chosen uniformly at random, after which we can terminate the algorithm and output an approximate shortest path.

\begin{theorem}\label{thm:vertex-approx-intro}
    Let $\mathcal{G}$ be a Chung-Lu graph or a GIRG on $n$ vertices with power-law parameter $\tau \in (2,3)$ and let $s_0$ and $s_1$ be two vertices of $\mathcal{G}$ chosen uniformly at random. Then $\textnormal{V-BFS}_{approx}(\mathcal{G},s_0,s_1)$ is a correct algorithm outputting a path of length at most $d(s_0,s_1)+1$ and with high probability it terminates in time at most $n^{\frac{\tau-2}{\tau-1}+o(1)}$.
\end{theorem}

It is possible to modify the algorithm to make it exact, in the sense that the output path is indeed of length exactly $d(s_0,s_1)$. In order to do that, we run the approximate algorithm until the two search trees meet and then finish to expand all the remaining vertices in one of the two layers currently being expanded (we choose the side that has the fewest remaining vertices in the current layer), see Algorithm \ref{algo:vertex-exact}. However, this comes at the cost of a polynomial slowdown.

\begin{theorem}\label{thm:vertex-exact-intro}
    Let $\mathcal{G}$ be a Chung-Lu graph or a GIRG on $n$ vertices with power-law parameter $\tau \in (2,3)$ and let $s_0$ and $s_1$ be two vertices of $\mathcal{G}$ chosen uniformly at random. Then $\textnormal{V-BFS}_{exact}(\mathcal{G},s_0,s_1)$ is a correct algorithm outputting a shortest path and with high probability it terminates in time at most $n^{\frac{1}{2}+o(1)}$.
\end{theorem}

The above exact algorithm can be sped up to $n^{\tau-2+o(1)}$ if for the ``layer emptying" phase at the end we query potential edges, i.e.\ we query pairs of vertices, one on each side, and ask for the existence of the edge between them. However, this changes the character of the algorithm, and not all data structures supporting BFS allow to efficiently query for an edge between two vertices. Therefore, we do not think it is still appropriate to consider this a BFS algorithm, and the details of this speed-up are omitted in the current paper.

If, instead of selecting $s_0$ and $s_1$ uniformly at random, an adversary is allowed to pick the starting vertices after the weights are drawn, but before the edges of the graph are unveiled, then we provide a modified algorithm that alternates between the searches edge by edge instead of vertex by vertex, yielding two searches that are perfectly balanced in terms of explored edges, see Algorithm~\ref{algo:edge-approx}. The algorithm outputs a path which is at most one hop longer than the distance between $s_0$ and $s_1$.

\begin{theorem}
\label{thm:edge-introduction}
    Let $\mathcal{G}$ be a Chung-Lu graph or a GIRG on $n$ vertices with power-law parameter $\tau \in (2,3)$ and let $s_0$ and $s_1$ be two vertices of $\mathcal{G}$ that an adversary can choose by looking at the weight sequence $(W_v)_{v\in \mathcal{G}}$. Then $\textnormal{E-BFS}_{approx}(\mathcal{G},s_0,s_1)$ is a correct algorithm outputting a path of length at most $d(s_0,s_1)+1$ and with high probability it terminates in time at most $n^{\frac{1}{2}+o(1)}$.
\end{theorem}

We remark here that aiming for an edge-balanced \emph{exact} algorithm will not lead to a runtime improvement compared to the vertex-balanced approach, as the adversary can pick the highest-weight vertices as source and target, both of degree $n^{\frac{1}{\tau-1}+o(1)}$, which are connected with constant (but in general not 1) probability. If it is not allowed to query directly whether there is an edge between those two vertices, then it is necessary to go through the adjacency list of one of them. This leads to a runtime of $n^{\frac{1}{\tau-1}+o(1)}$, which is much larger than $n^{\frac{1}{2}}$.

\paragraph*{Proof idea}
In Chung-Lu graphs and in GIRGs, every vertex 
$v$ draws a weight $W_v$ from a power-law distribution with parameter $\tau\in(2,3)$, so $\Pr[W_v = w] = \Theta(w^{-\tau})$. The weights correspond to the expected degrees, $\E[\deg(v) \mid W_v] = \Theta(W_v)$. Over the course of a Breadth-First Search on a Chung-Lu graph or a GIRG, the largest weight encountered so far increases rapidly. The scale-free degree distribution of these graphs implies that the degree of a vertex' neighbor follows a power-law distribution with exponent $\tau-1$. A short computation then shows that the largest weight found in the neighborhood of a vertex of weight $w$ is of order $w^{\frac{1}{\tau-2}}$ (note that $\frac{1}{\tau-2} > 1$ since we are considering power-law exponents $\tau<3$), as long as $w^{\frac{1}{\tau-2}}$ is smaller than $w_{\max} \coloneqq n^{\frac{1}{\tau-1}}$, which is the order of the largest weight present in the entire graph. Our vertex-balanced algorithm finds a path between the two sources when the two search trees meet. In particular, this happens if the two search trees discover one vertex (among the constantly many such vertices) of weight roughly $w_{\max}$. 
We expect to find such vertices in the neighborhood of a vertex of weight $n^{\frac{\tau-2}{\tau-1}}$, since then such a node will have a neighbor of weight $(n^{\frac{\tau-2}{\tau-1}})^{\frac{1}{\tau-2}} = w_{\max}$. Therefore, the algorithm $\textnormal{V-BFS}_{approx}$ finds a path once the two search trees have expanded vertices of weight $n^{\frac{\tau-2}{\tau-1}}$. Its runtime is then given by the sum of the degrees of the vertices that have been expanded on both sides, because this corresponds to the total number of edges that have been processed. Since both sides are balanced on the vertex level, the two searches contribute equally to this sum. Since the degrees are power-law distributed with exponent $\tau-1\in(1,2)$, the sum is dominated by its largest term, which yields the $n^{\frac{\tau-2}{\tau-1}}$ runtime for the $\textnormal{V-BFS}_{approx}$ algorithm (c.f.\ Figure~\ref{fig:costfrac} for simulation results). 

The \textit{exact} version of the vertex-balanced algorithm follows the same course until a path between the two sources is found. At this point, $\textnormal{V-BFS}_{approx}$ returns this path and terminates, while $\textnormal{V-BFS}_{exact}$ empties the smallest remaining layer in order to check for the existence of a shorter path. There are two possible outcomes of this additional phase of the algorithm, and which one occurs depends on the presence or absence of vertices of weight $n^{\frac{1}{2}}$ (or larger) in the layer currently being expanded on each side. On the one hand, if such vertices are present in the current layer of both sides, then there will be an edge that connects these layers directly and hence provides a shorter path. This edge connects two vertices of weight $n^{\frac{1}{2}}$, and hence is found when we expand a vertex of weight roughly $n^{\frac{1}{2}}$, which yields a runtime of the same order. On the other hand, if (at least) one of the two layers has no vertex of weight more than $n^{\frac{1}{2}}$, then emptying this layer will take time less than $n^{\frac{1}{2}}$. It turns out that choosing the smallest of the two layers guarantees that we indeed empty that fast-to-process layer in the case where one of them contains vertex of weight $\Omega(n^{\frac{1}{2}})$ and the other one does not. Finally, the edge-balanced algorithm works very similarly to the approximate vertex-balanced algorithm, except that we switch sides after each explored edge. In particular, if the two sources are chosen uniformly at random (and hence have almost constant weight with high probability), then the runtime upper bound of $n^{\frac{\tau-2}{\tau-1}}$ also applies to $\textnormal{E-BFS}_{approx}$. However, if the adversary chooses two vertices of very large weight (e.g.\ of weight $w_{\max}$), then the vertex-balanced algorithm needs time at least $w_{\max} \gg n^{\frac{1}{2}} \gg n^{\frac{\tau-2}{\tau-1}}$ just to expand one vertex. The edge-balanced algorithm avoids this by switching sides after each edge, and in particular it is not necessary that any vertex is fully expanded for the algorithm to terminate. Even if the adversary chooses two vertices of weight much larger than $n^{\frac{1}{2}}$, a path between them is found within time $n^{\frac{1}{2}}$, because two vertices of such high weight have many common neighbors, so we only need to uncover a tiny fraction of their neighborhoods to find a vertex in the intersection of their neighborhoods.

\paragraph{Applicability of runtime bounds to GIRGs} We formulate our theoretical runtime bounds for Chung Lu graphs only, since many intermediate statements rely on a conditional weight density result which is stated only for Chung-Lu graphs. Note however that our results continue to hold for GIRGs, in particular the intuitive proof ideas outlined above also apply to them, albeit at the cost of more technical proofs. We refer the reader to Subsection~\ref{subsubsec:applicability} for a more detailed discussion, and to the Experimental Results section, notably Figure~\ref{fig:log-edge-log-cost-plots}, for empirical evidence.

\paragraph*{Organization of the paper}
The remainder of the paper is organized as follows. In Section~\ref{sec:preliminaries}, we introduce important notation and tools which we use throughout the paper, including concentration inequalities. Section~\ref{sec:algorithms} contains the full pseudocode of all three algorithms which we analyze in this paper, complemented by intuitive descriptions. In Section~\ref{sec:model}, we formally introduce Chung-Lu and Geometric Inhomogeneous Random Graphs and describe their key properties that are required for our proofs, such as information about their degrees, marginal connection probabilities and the size of components. Our proofs follow in Section~\ref{sec:proofs}, where we first demonstrate correctness of our algorithms (Subsection~\ref{sec:correctness}) and then analyze their runtimes (Subsection~\ref{sec:runtime}). Section~\ref{sec:simulations} contains simulations, where we run our algorithms on synthetic graphs as well as real-world networks.

\section{Preliminaries}\label{sec:preliminaries}

\subsection{Notation}

We begin by introducing some useful notation that will be used throughout the paper. We say an event happens $\emph{with high probability (whp)}$ if it happens with probability $1-o(1)$ as $n\rightarrow\infty$. We use standard Landau notation, such as $O(.),\Theta(.),\omega(.)$, and use both $f(n)\le O(g(n))$ as well as $f(n)= O(g(n))$ to indicate that the function $g(n)$ dominates $f(n)$ (and similarly for the other Landau notations). In particular, we will often use inequalities for the sake of readability in cases where the term involving Landau notation bounds other expressions from below or above. We will always work with $n$ tending to $\infty$, and as such some of the inequalities we write are only true for large enough $n$, but for the sake of readability we will not always mention when a large $n$ is needed.

\subsection{Concentration inequalities}
In the proofs we will use the following concentration inequalities.
\begin{theorem}[Chernoff-Hoeffding bound, Theorem 1.1 in \cite{dubhashi2009concentration}] \label{thm:dubhashichernoff}
  Let $X:=\sum_{i\in [n]}X_i$ where for all $i \in [n]$, the random variables $X_i$ are independently
  distributed in $[0,1]$. Then 
\begin{enumerate}[label=(\roman*)]
\item $\Pr[ X > (1+\eps)\mathbb{E}[X]] \leq \exp\left( -\frac{\eps^2}{3}\mathbb{E}[X] \right)$ for all $0<\eps<1$,
\item $\Pr[ X < (1-\eps)\mathbb{E}[X]] \leq \exp\left( -\frac{\eps^2}{2}\mathbb{E}[X] \right)$ for all $0<\eps<1$, and
\item $\Pr[X>t]\leq 2^{-t}$ for all $t>2e\E[X]$.
\end{enumerate}  
\end{theorem}

The following lemma is a simplified application of the second moment method which can be used when the random variable is a sum of pairwise correlated random variables.

\begin{lemma}
\label{lem:chebyshev-negatively-correlated}
Let $(Z_n)_{n \in \mathbb{N}}$ be a sequence of random variables such that each $Z_n$ is the sum of pairwise negatively correlated Bernoulli random variables. If $\E{[Z_n]} \to \infty$, then $\Pr{[Z_n=0]} \to 0$.
\end{lemma}
\begin{proof}
    Let $X_1, \dots, X_n$ be pairwise negatively correlated Bernoulli random variables. We begin by showing that their sum $Z_n\coloneqq\sum_{i=1}^n X_i$ satisfies $\var(Z_n) \le \E{[Z_n]}$. Indeed, using the definition of the variance and a reordering of the terms yields
    \begin{align*}
        \var(Z_n)= \E{[Z_n^2]}- (\E{[Z_n]})^2 = \sum_{1\le i,j \le n} (\E{[X_i X_j]} -\E{[X_i]}\E{[X_j]}).
    \end{align*}
    For $i \neq j$, since $X_i$ and $X_j$ are negatively correlated, the respective summand is at most zero. We can therefore upper-bound the variance by
    \begin{align*}
        \var(Z_n) \le \sum_{1\le i \le n} (\E{[X_i^2]} -\E{[X_i]}^2) \le  \sum_{1\le i \le n} \E{[X_i^2]} = \sum_{1\le i \le n} \E{[X_i]} = \E{[Z_n]}.
    \end{align*}
    Now, using Chebyshev's inequality, we have
    \begin{align*}
        \Pr[Z_n=0] \le \Pr[|Z_n-\E{[Z_n]}| \ge \E{[Z_n]}] \le \frac{\var(Z_n)}{\E{[Z_n]}^2} \le \frac{1}{\E{[Z_n]}}.
    \end{align*}
    As $\E{[Z_n]}\to \infty$, the right-hand side tends to 0, concluding the proof.
\end{proof}

\section{Algorithms}\label{sec:algorithms}

In this section, we present our three algorithms, providing pseudocode and intuitive descriptions. Recall that our algorithms provide correct outputs on any graph $\mathcal{G}$. The common rationale of all three algorithms is that, to find an (approximate) shortest path between two vertices $s_0$ and $s_1$, we can run two breadth-first searches in an alternating manner, one started at $s_0$ and the other at $s_1$, until the two search trees meet ``in the middle". The alternation is done with the goal of ``balancing the amount of work" that is performed between the two searches. How this balancing is achieved depends on the algorithm and will be described in detail below. We will consider two types of balancing, \textit{vertex-balanced} and \textit{edge-balanced}.
We begin by describing the vertex-balanced algorithms.

\subsection{The vertex-balanced algorithms}

We begin by remarking that until the first path connecting vertices $s_0$ and $s_1$ is found, the algorithms $\textnormal{V-BFS}_{approx}(\mathcal{G},s_0,s_1)$ and $\textnormal{V-BFS}_{exact}(\mathcal{G},s_0,s_1)$ behave exactly in the same way. In particular, this means that if $s_0$ and $s_1$ lie in different components of the graph, then the algorithms exhibit the same behavior until termination. We describe this first phase jointly. Consider vertices $s_0$ and $s_1$ in our host graph $\mathcal{G}$. We begin by initializing for each search on side $s_i$, $i \in \{0,1\}$, a queue $Q_{s_i}^-$, containing at the beginning only the respective starting vertex $s_i$ of the search, an empty queue $Q_{s_i}^+$ and a set $S_{s_i}$ (initialized with $s_i$) that records all the vertices that have been ``discovered" by the search on side $s_i$ (lines~1 and 2 in Algorithms \ref{algo:vertex-approx} and \ref{algo:vertex-exact}). The concatenation of $Q_{s_i}^-$ and $Q_{s_i}^+$ constitutes one unified queue containing the vertices $v$ that we expand, i.e.\ whose neighborhoods $\Gamma(v)$ we add to the set of discovered vertices $S_i$. The vertices in $Q_{s_i}^-$ belong to a different layer of the BFS than the vertices in $Q_{s_i}^+$, hence we need to keep the two ``subqueues" separate for the exact algorithm $\textnormal{V-BFS}_{exact}(\mathcal{G},s_0,s_1)$, see details below. The process of expansion occurs as follows. 

We first introduce some simplifying notation. Denote $s_{1-i}$ by $\overline{s}$ whenever $s=s_i$. Now, observe the following: As soon as on a side $s$ there are no more vertices to be processed and we have not yet found a path to $\overline{s}$, this means that we have exhausted the component of $s$ and the algorithm should terminate, returning $\emptyset$ (line~16 in Algorithm \ref{algo:vertex-approx} and line~29 in Algorithm \ref{algo:vertex-exact}), which we use as shorthand for ``there exists no path connecting $s_0$ and $s_1$". Therefore we run the searches only as long as \emph{both} $Q_{s_0}^-$ and $Q_{s_1}^-$ are non-empty (line~3). Note that we only remove vertices from the queue $Q_{s}^-$ (line~5) but only add vertices to $Q_{s}^+$ (line~6). Only when the current layer $Q_{s}^-$ is finished and hence empty, we replace it entirely by the next layer $Q_{s}^+$, at the same time emptying the queue $Q_{s}^+$ (lines~13-15 and lines~26-28).

As long as the condition in line~3 is fulfilled, we first select on which side to continue the search. We choose the side which has discovered fewer vertices so far by determining  $s = \argmin_{s_0,s_1}\{|S_{s_0}|, |S_{s_1}|\}$ (line~4), see Figure \ref{fig:balancing}.\footnote{Note that here we compare the cardinalities of the sets $S_i$ and not of the queues $Q_{s_i}^-$. In principle if, say on side $s_0$, only low-degree vertices were added, while $Q_{s_1}^-$ contained the neighbors of a high-degree vertex, then $Q_{s_0}^-$ would be consistently smaller. This would prevent the alternation between the two searches, and instead of a bidirectional process we would grow ``long, thin" branches on one side only.} If $|S_{s_0}|=|S_{s_1}|$, the algorithm picks a side arbitrarily. We set the side $s$ to be expanded accordingly and begin processing the next vertex $v$ by removing it from the queue ($\mathtt{pop}(Q_s^-)$, line~5) and expanding it by appending those neighbors that the search from side $s$ has not encountered yet, in uniformly random order, to $Q_s^+$ ($\mathtt{append}(Q_s^+, \Gamma(v) \setminus S_s)$, line~6) and adding them to $S_s$ (line~7). For a given run of the algorithm, we will refer to $v$ as the \textit{parent vertex} of those of its neighbors $\Gamma(v) \setminus S_s$ which had not yet been encountered.
Each time a vertex has been expanded, we check if one of its neighbors has already been discovered by the BFS on the other side, i.e.\ if $\Gamma(v)$ intersects $S_{\overline{s}}$. If so, then the algorithm has found a path between $s_0$ and $s_1$. Indeed, take any vertex $u$ in this intersection $\Gamma(v) \cap S_{\overline{s}}$ (line~9 respectively line~22). Then concatenating the path from $s$ to $u$ in $S_s$ found by the BFS on one side, with the (reversed) path from $\overline{s}$ to $u$ in $S_{\overline{s}}$ found by the BFS from the other side, yields such a path (lines~10-12 respectively lines~23-25).

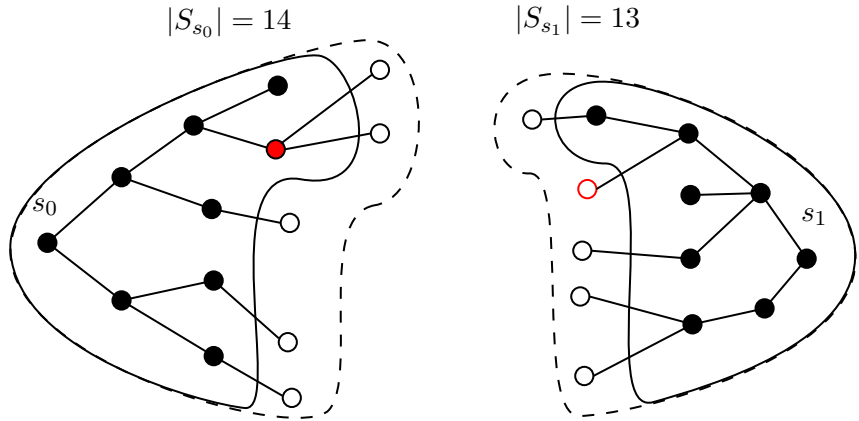
\begin{figure}
    \centering
\tikzset{every picture/.style={line width=0.75pt}} %set default line width to 0.75pt        

\begin{tikzpicture}[x=0.75pt,y=0.75pt,yscale=-1,xscale=1]

%Shape: Polygon Curved [id:ds1306304688815041] 
\draw   (250.77,56.25) .. controls (269.77,61.25) and (285.77,126.25) .. (240.77,118.25) .. controls (195.77,110.25) and (232.77,235.3) .. (212.77,233.3) .. controls (192.77,231.3) and (99.86,206.3) .. (95.81,154.8) .. controls (91.77,103.3) and (231.77,51.25) .. (250.77,56.25) -- cycle ;
%Shape: Polygon Curved [id:ds28260627679069406] 
\draw   (387.77,69.3) .. controls (356.77,72.3) and (364.29,110.3) .. (392.77,110.3) .. controls (421.24,110.3) and (388.77,238.3) .. (416.77,230.3) .. controls (444.77,222.3) and (520.03,200.62) .. (517.4,153.96) .. controls (514.77,107.3) and (418.77,66.3) .. (387.77,69.3) -- cycle ;
%Shape: Polygon Curved [id:ds5314801017275448] 
\draw  [dash pattern={on 4.5pt off 4.5pt}] (277.77,131.3) .. controls (247.77,136.3) and (269.77,217.3) .. (254.77,234.3) .. controls (239.77,251.3) and (95.86,211.3) .. (95.81,154.8) .. controls (95.77,98.3) and (237.77,50.3) .. (272.77,48.3) .. controls (307.77,46.3) and (307.77,126.3) .. (277.77,131.3) -- cycle ;
%Shape: Polygon Curved [id:ds8416488110479159] 
\draw  [dash pattern={on 4.5pt off 4.5pt}] (354.77,115.3) .. controls (375.77,114.3) and (356.77,237.3) .. (379.77,237.3) .. controls (402.77,237.3) and (524.03,209.62) .. (517.4,153.96) .. controls (510.77,98.3) and (405.77,62.3) .. (365.77,65.3) .. controls (325.77,68.3) and (333.77,116.3) .. (354.77,115.3) -- cycle ;
%Shape: Ellipse [id:dp015950529691229676] 
\draw  [fill={rgb, 255:red, 0; green, 0; blue, 0 }  ,fill opacity=1 ] (109.72,150.19) .. controls (109.72,147.65) and (111.75,145.58) .. (114.24,145.58) .. controls (116.74,145.58) and (118.77,147.65) .. (118.77,150.19) .. controls (118.77,152.74) and (116.74,154.81) .. (114.24,154.81) .. controls (111.75,154.81) and (109.72,152.74) .. (109.72,150.19) -- cycle ;
%Shape: Ellipse [id:dp5373473645973064] 
\draw  [fill={rgb, 255:red, 0; green, 0; blue, 0 }  ,fill opacity=1 ] (146.72,179.19) .. controls (146.72,176.65) and (148.75,174.58) .. (151.24,174.58) .. controls (153.74,174.58) and (155.77,176.65) .. (155.77,179.19) .. controls (155.77,181.74) and (153.74,183.81) .. (151.24,183.81) .. controls (148.75,183.81) and (146.72,181.74) .. (146.72,179.19) -- cycle ;
%Shape: Ellipse [id:dp5570434085371879] 
\draw  [fill={rgb, 255:red, 0; green, 0; blue, 0 }  ,fill opacity=1 ] (146.72,117.19) .. controls (146.72,114.65) and (148.75,112.58) .. (151.24,112.58) .. controls (153.74,112.58) and (155.77,114.65) .. (155.77,117.19) .. controls (155.77,119.74) and (153.74,121.81) .. (151.24,121.81) .. controls (148.75,121.81) and (146.72,119.74) .. (146.72,117.19) -- cycle ;
%Shape: Ellipse [id:dp6454845731781416] 
\draw  [fill={rgb, 255:red, 0; green, 0; blue, 0 }  ,fill opacity=1 ] (182.72,91.19) .. controls (182.72,88.65) and (184.75,86.58) .. (187.24,86.58) .. controls (189.74,86.58) and (191.77,88.65) .. (191.77,91.19) .. controls (191.77,93.74) and (189.74,95.81) .. (187.24,95.81) .. controls (184.75,95.81) and (182.72,93.74) .. (182.72,91.19) -- cycle ;
%Shape: Ellipse [id:dp12755434975413638] 
\draw  [fill={rgb, 255:red, 0; green, 0; blue, 0 }  ,fill opacity=1 ] (191.72,133.19) .. controls (191.72,130.65) and (193.75,128.58) .. (196.24,128.58) .. controls (198.74,128.58) and (200.77,130.65) .. (200.77,133.19) .. controls (200.77,135.74) and (198.74,137.81) .. (196.24,137.81) .. controls (193.75,137.81) and (191.72,135.74) .. (191.72,133.19) -- cycle ;
%Shape: Ellipse [id:dp409492206366204] 
\draw  [fill={rgb, 255:red, 255; green, 0; blue, 0 }  ,fill opacity=1 ] (223.72,103.19) .. controls (223.72,100.65) and (225.75,98.58) .. (228.24,98.58) .. controls (230.74,98.58) and (232.77,100.65) .. (232.77,103.19) .. controls (232.77,105.74) and (230.74,107.81) .. (228.24,107.81) .. controls (225.75,107.81) and (223.72,105.74) .. (223.72,103.19) -- cycle ;
%Shape: Ellipse [id:dp5001617043860968] 
\draw   (229.72,200.19) .. controls (229.72,197.65) and (231.75,195.58) .. (234.24,195.58) .. controls (236.74,195.58) and (238.77,197.65) .. (238.77,200.19) .. controls (238.77,202.74) and (236.74,204.81) .. (234.24,204.81) .. controls (231.75,204.81) and (229.72,202.74) .. (229.72,200.19) -- cycle ;
%Shape: Ellipse [id:dp9038647907793198] 
\draw  [fill={rgb, 255:red, 0; green, 0; blue, 0 }  ,fill opacity=1 ] (224.72,71.19) .. controls (224.72,68.65) and (226.75,66.58) .. (229.24,66.58) .. controls (231.74,66.58) and (233.77,68.65) .. (233.77,71.19) .. controls (233.77,73.74) and (231.74,75.81) .. (229.24,75.81) .. controls (226.75,75.81) and (224.72,73.74) .. (224.72,71.19) -- cycle ;
%Shape: Ellipse [id:dp21146954155087927] 
\draw  [fill={rgb, 255:red, 0; green, 0; blue, 0 }  ,fill opacity=1 ] (192.72,169.19) .. controls (192.72,166.65) and (194.75,164.58) .. (197.24,164.58) .. controls (199.74,164.58) and (201.77,166.65) .. (201.77,169.19) .. controls (201.77,171.74) and (199.74,173.81) .. (197.24,173.81) .. controls (194.75,173.81) and (192.72,171.74) .. (192.72,169.19) -- cycle ;
%Shape: Ellipse [id:dp8423773157357138] 
\draw   (275.72,63.19) .. controls (275.72,60.65) and (277.75,58.58) .. (280.24,58.58) .. controls (282.74,58.58) and (284.77,60.65) .. (284.77,63.19) .. controls (284.77,65.74) and (282.74,67.81) .. (280.24,67.81) .. controls (277.75,67.81) and (275.72,65.74) .. (275.72,63.19) -- cycle ;
%Shape: Ellipse [id:dp864639068221696] 
\draw  [fill={rgb, 255:red, 0; green, 0; blue, 0 }  ,fill opacity=1 ] (192.72,207.19) .. controls (192.72,204.65) and (194.75,202.58) .. (197.24,202.58) .. controls (199.74,202.58) and (201.77,204.65) .. (201.77,207.19) .. controls (201.77,209.74) and (199.74,211.81) .. (197.24,211.81) .. controls (194.75,211.81) and (192.72,209.74) .. (192.72,207.19) -- cycle ;
%Shape: Ellipse [id:dp9528987415967012] 
\draw   (230.72,140.19) .. controls (230.72,137.65) and (232.75,135.58) .. (235.24,135.58) .. controls (237.74,135.58) and (239.77,137.65) .. (239.77,140.19) .. controls (239.77,142.74) and (237.74,144.81) .. (235.24,144.81) .. controls (232.75,144.81) and (230.72,142.74) .. (230.72,140.19) -- cycle ;
%Shape: Ellipse [id:dp36993297047462503] 
\draw   (231.72,228.19) .. controls (231.72,225.65) and (233.75,223.58) .. (236.24,223.58) .. controls (238.74,223.58) and (240.77,225.65) .. (240.77,228.19) .. controls (240.77,230.74) and (238.74,232.81) .. (236.24,232.81) .. controls (233.75,232.81) and (231.72,230.74) .. (231.72,228.19) -- cycle ;
%Shape: Ellipse [id:dp2463068842596826] 
\draw   (275.72,95.19) .. controls (275.72,92.65) and (277.75,90.58) .. (280.24,90.58) .. controls (282.74,90.58) and (284.77,92.65) .. (284.77,95.19) .. controls (284.77,97.74) and (282.74,99.81) .. (280.24,99.81) .. controls (277.75,99.81) and (275.72,97.74) .. (275.72,95.19) -- cycle ;
%Straight Lines [id:da9843052444240591] 
\draw    (114.24,150.19) -- (151.24,117.19) ;
%Straight Lines [id:da5139644427286716] 
\draw    (114.24,150.19) -- (151.24,179.19) ;
%Straight Lines [id:da6827155436771416] 
\draw    (151.24,117.19) -- (187.24,91.19) ;
%Straight Lines [id:da35682211340722236] 
\draw    (151.24,117.19) -- (196.24,133.19) ;
%Straight Lines [id:da26029802331270635] 
\draw    (151.24,179.19) -- (197.24,169.19) ;
%Straight Lines [id:da021327533560323797] 
\draw    (151.24,179.19) -- (197.24,207.19) ;
%Straight Lines [id:da639718987329334] 
\draw    (187.24,91.19) -- (229.24,71.19) ;
%Straight Lines [id:da9129929276892039] 
\draw    (187.24,91.19) -- (223.77,102.3) ;
%Straight Lines [id:da6567198914763859] 
\draw    (196.24,133.19) -- (230.72,140.19) ;
%Straight Lines [id:da1248575172321742] 
\draw    (197.24,169.19) -- (229.77,198.25) ;
%Straight Lines [id:da5330497881707156] 
\draw    (197.24,207.19) -- (231.77,227.25) ;
%Straight Lines [id:da36329320530624953] 
\draw    (232.77,103.19) -- (275.72,95.19) ;
%Straight Lines [id:da4765254011380976] 
\draw    (276.77,65.25) -- (230.77,100.3) ;
%Shape: Ellipse [id:dp49595411557456237] 
\draw  [fill={rgb, 255:red, 0; green, 0; blue, 0 }  ,fill opacity=1 ] (488.72,158.19) .. controls (488.72,155.65) and (490.75,153.58) .. (493.24,153.58) .. controls (495.74,153.58) and (497.77,155.65) .. (497.77,158.19) .. controls (497.77,160.74) and (495.74,162.81) .. (493.24,162.81) .. controls (490.75,162.81) and (488.72,160.74) .. (488.72,158.19) -- cycle ;
%Shape: Ellipse [id:dp7858427115800537] 
\draw  [fill={rgb, 255:red, 0; green, 0; blue, 0 }  ,fill opacity=1 ] (465.72,125.19) .. controls (465.72,122.65) and (467.75,120.58) .. (470.24,120.58) .. controls (472.74,120.58) and (474.77,122.65) .. (474.77,125.19) .. controls (474.77,127.74) and (472.74,129.81) .. (470.24,129.81) .. controls (467.75,129.81) and (465.72,127.74) .. (465.72,125.19) -- cycle ;
%Shape: Ellipse [id:dp6026607696226928] 
\draw  [fill={rgb, 255:red, 0; green, 0; blue, 0 }  ,fill opacity=1 ] (467.72,183.19) .. controls (467.72,180.65) and (469.75,178.58) .. (472.24,178.58) .. controls (474.74,178.58) and (476.77,180.65) .. (476.77,183.19) .. controls (476.77,185.74) and (474.74,187.81) .. (472.24,187.81) .. controls (469.75,187.81) and (467.72,185.74) .. (467.72,183.19) -- cycle ;
%Shape: Ellipse [id:dp16167258131158135] 
\draw  [fill={rgb, 255:red, 0; green, 0; blue, 0 }  ,fill opacity=1 ] (429.72,95.19) .. controls (429.72,92.65) and (431.75,90.58) .. (434.24,90.58) .. controls (436.74,90.58) and (438.77,92.65) .. (438.77,95.19) .. controls (438.77,97.74) and (436.74,99.81) .. (434.24,99.81) .. controls (431.75,99.81) and (429.72,97.74) .. (429.72,95.19) -- cycle ;
%Shape: Ellipse [id:dp824132564641341] 
\draw  [fill={rgb, 255:red, 0; green, 0; blue, 0 }  ,fill opacity=1 ] (431.72,191.19) .. controls (431.72,188.65) and (433.75,186.58) .. (436.24,186.58) .. controls (438.74,186.58) and (440.77,188.65) .. (440.77,191.19) .. controls (440.77,193.74) and (438.74,195.81) .. (436.24,195.81) .. controls (433.75,195.81) and (431.72,193.74) .. (431.72,191.19) -- cycle ;
%Shape: Ellipse [id:dp2906046716391274] 
\draw  [fill={rgb, 255:red, 0; green, 0; blue, 0 }  ,fill opacity=1 ] (430.72,126.19) .. controls (430.72,123.65) and (432.75,121.58) .. (435.24,121.58) .. controls (437.74,121.58) and (439.77,123.65) .. (439.77,126.19) .. controls (439.77,128.74) and (437.74,130.81) .. (435.24,130.81) .. controls (432.75,130.81) and (430.72,128.74) .. (430.72,126.19) -- cycle ;
%Shape: Ellipse [id:dp3964992443605656] 
\draw  [color={rgb, 255:red, 255; green, 0; blue, 0 }  ,draw opacity=1 ] (379.24,123.19) .. controls (379.24,120.65) and (381.27,118.58) .. (383.77,118.58) .. controls (386.26,118.58) and (388.29,120.65) .. (388.29,123.19) .. controls (388.29,125.74) and (386.26,127.81) .. (383.77,127.81) .. controls (381.27,127.81) and (379.24,125.74) .. (379.24,123.19) -- cycle ;
%Shape: Ellipse [id:dp7861718870025962] 
\draw  [fill={rgb, 255:red, 0; green, 0; blue, 0 }  ,fill opacity=1 ] (383.72,86.19) .. controls (383.72,83.65) and (385.75,81.58) .. (388.24,81.58) .. controls (390.74,81.58) and (392.77,83.65) .. (392.77,86.19) .. controls (392.77,88.74) and (390.74,90.81) .. (388.24,90.81) .. controls (385.75,90.81) and (383.72,88.74) .. (383.72,86.19) -- cycle ;
%Shape: Ellipse [id:dp7081138443722089] 
\draw   (375.72,177.19) .. controls (375.72,174.65) and (377.75,172.58) .. (380.24,172.58) .. controls (382.74,172.58) and (384.77,174.65) .. (384.77,177.19) .. controls (384.77,179.74) and (382.74,181.81) .. (380.24,181.81) .. controls (377.75,181.81) and (375.72,179.74) .. (375.72,177.19) -- cycle ;
%Straight Lines [id:da4312196593163058] 
\draw    (470.24,125.19) -- (493.24,158.19) ;
%Straight Lines [id:da9411440028209112] 
\draw    (472.24,183.19) -- (493.24,158.19) ;
%Straight Lines [id:da26119207930154253] 
\draw    (470.24,125.19) -- (435.24,126.19) ;
%Straight Lines [id:da21152246030185706] 
\draw    (470.24,125.19) -- (434.24,95.19) ;
%Straight Lines [id:da47872827466933443] 
\draw    (434.24,95.19) -- (388.24,86.19) ;
%Straight Lines [id:da912557305647701] 
\draw    (434.24,95.19) -- (388.29,123.19) ;
%Straight Lines [id:da1280110701913053] 
\draw    (384.77,177.19) -- (436.24,191.19) ;
%Shape: Ellipse [id:dp4510855217258609] 
\draw  [color={rgb, 255:red, 0; green, 0; blue, 0 }  ,draw opacity=1 ] (351.72,88.19) .. controls (351.72,85.65) and (353.75,83.58) .. (356.24,83.58) .. controls (358.74,83.58) and (360.77,85.65) .. (360.77,88.19) .. controls (360.77,90.74) and (358.74,92.81) .. (356.24,92.81) .. controls (353.75,92.81) and (351.72,90.74) .. (351.72,88.19) -- cycle ;
%Straight Lines [id:da6746727502735399] 
\draw    (360.77,88.19) -- (388.24,86.19) ;
%Shape: Ellipse [id:dp3067030883397929] 
\draw   (377.72,217.19) .. controls (377.72,214.65) and (379.75,212.58) .. (382.24,212.58) .. controls (384.74,212.58) and (386.77,214.65) .. (386.77,217.19) .. controls (386.77,219.74) and (384.74,221.81) .. (382.24,221.81) .. controls (379.75,221.81) and (377.72,219.74) .. (377.72,217.19) -- cycle ;
%Straight Lines [id:da7220024269453449] 
\draw    (386.77,217.19) -- (436.24,191.19) ;
%Straight Lines [id:da5883969512285016] 
\draw    (472.24,183.19) -- (436.24,191.19) ;
%Shape: Ellipse [id:dp6807786541071144] 
\draw  [fill={rgb, 255:red, 0; green, 0; blue, 0 }  ,fill opacity=1 ] (430.72,158.19) .. controls (430.72,155.65) and (432.75,153.58) .. (435.24,153.58) .. controls (437.74,153.58) and (439.77,155.65) .. (439.77,158.19) .. controls (439.77,160.74) and (437.74,162.81) .. (435.24,162.81) .. controls (432.75,162.81) and (430.72,160.74) .. (430.72,158.19) -- cycle ;
%Straight Lines [id:da9113998109476636] 
\draw    (470.24,125.19) -- (435.24,158.19) ;
%Shape: Ellipse [id:dp08260241049350625] 
\draw   (376.72,154.19) .. controls (376.72,151.65) and (378.75,149.58) .. (381.24,149.58) .. controls (383.74,149.58) and (385.77,151.65) .. (385.77,154.19) .. controls (385.77,156.74) and (383.74,158.81) .. (381.24,158.81) .. controls (378.75,158.81) and (376.72,156.74) .. (376.72,154.19) -- cycle ;
%Straight Lines [id:da031201490729058978] 
\draw    (385.77,154.19) -- (435.24,158.19) ;

% Text Node
\draw (105,126.4) node [anchor=north west][inner sep=0.75pt]    {$s_{0}$};
% Text Node
\draw (489,131.4) node [anchor=north west][inner sep=0.75pt]    {$s_{1}$};
% Text Node
\draw (172,30.4) node [anchor=north west][inner sep=0.75pt]    {$|S_{s_{0}} |=14$};
% Text Node
\draw (346,29.4) node [anchor=north west][inner sep=0.75pt]    {$|S_{s_{1}} |=13$};

\end{tikzpicture}
\caption{Illustration of the balancing of the BFS trees done by the vertex-balanced algorithms. Filled dots denote expanded vertices, while hollow dots mark non-expanded vertices that have already been discovered. The red filled dot is the most recently expanded vertex, which increased the size of the left BFS tree from 12 to 14. The right BFS tree is now smaller, so the hollow red vertex will be expanded next.}
\label{fig:balancing}
\end{figure}

\begin{algorithm}[H]

\caption{{\sc Vertex-balanced BFS for approximate shortest path $\textnormal{V-BFS}_{approx}(\mathcal{G},s_0,s_1)$}}

\label{algo:vertex-approx}

\begin{algorithmic}[1]
\State $Q_{s_0}^- \gets \mathtt{Queue}(s_0)$, $Q_{s_0}^+ \gets \mathtt{Queue}(\emptyset), S_{s_0} \gets \{s_0\}$
\State $Q_{s_1}^- \gets \mathtt{Queue}(s_1)$, $Q_{s_1}^+ \gets \mathtt{Queue}(\emptyset), S_{s_1} \gets \{s_1\}$ 
\While{$Q_{s_0}^- \neq\emptyset$ and $Q_{s_1}^- \neq\emptyset$}
  \State{$s \gets \argmin_{s_0,s_1}\{|S_{s_0}|, |S_{s_1}|\}$}
  \State{$v \gets \mathtt{pop}(Q_s^-)$}
  \State{$\mathtt{append}(Q_s^+, \Gamma(v) \setminus S_s)$}
  \State{$S_s \gets S_s \cup \Gamma(v)$}
  \If{$\Gamma(v) \cap S_{\overline{s}} \neq \emptyset$}
  \State{let $u\in \Gamma(v) \cap S_{\overline{s}}$}
  \State{let $\pi$ be the path from $s$ to $u$ in $S_s$ found by the BFS}
  \State{let $\overline{\pi}$ be the path from $\overline{s}$ to $u$ in $S_{\overline{s}}$ found by the BFS}
  \State\Return{$\pi \cup \overline{\pi}$}
  \EndIf
  \If{$Q_s^- = \emptyset$}
  \State{$Q_s^- \gets Q_s^+$}
  \State{$Q_s^+ \gets \mathtt{Queue}(\emptyset)$}
  \EndIf
\EndWhile
\State\Return{$\emptyset$}

\end{algorithmic}
\end{algorithm}

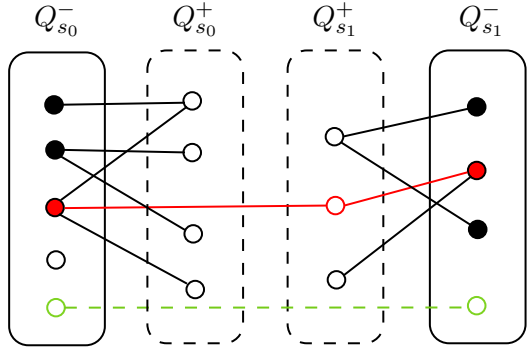
\begin{figure}
    \centering

\tikzset{every picture/.style={line width=0.75pt}} %set default line width to 0.75pt        

\begin{tikzpicture}[x=0.75pt,y=0.75pt,yscale=-1,xscale=1]

%Rounded Rect [id:dp8267894836883589] 
\draw   (187,72.55) .. controls (187,67.28) and (191.28,63) .. (196.55,63) -- (225.21,63) .. controls (230.49,63) and (234.77,67.28) .. (234.77,72.55) -- (234.77,200.75) .. controls (234.77,206.02) and (230.49,210.3) .. (225.21,210.3) -- (196.55,210.3) .. controls (191.28,210.3) and (187,206.02) .. (187,200.75) -- cycle;
\node[above] at (211,60) {$Q_{s_0}^-$};
%Shape: Circle [id:dp4535646022744805] 
\draw  [color={rgb, 255:red, 0; green, 0; blue, 0 }  ,draw opacity=1 ][fill={rgb, 255:red, 255; green, 0; blue, 0 }  ,fill opacity=1 ] (205.5,141.03) .. controls (205.5,138.61) and (207.46,136.65) .. (209.88,136.65) .. controls (212.3,136.65) and (214.27,138.61) .. (214.27,141.03) .. controls (214.27,143.45) and (212.3,145.42) .. (209.88,145.42) .. controls (207.46,145.42) and (205.5,143.45) .. (205.5,141.03) -- cycle ;
%Shape: Circle [id:dp8571565383551499] 
\draw  [color={rgb, 255:red, 0; green, 0; blue, 0 }  ,draw opacity=1 ][fill={rgb, 255:red, 0; green, 0; blue, 0 }  ,fill opacity=1 ] (205,89.38) .. controls (205,86.96) and (206.96,85) .. (209.38,85) .. controls (211.8,85) and (213.77,86.96) .. (213.77,89.38) .. controls (213.77,91.8) and (211.8,93.77) .. (209.38,93.77) .. controls (206.96,93.77) and (205,91.8) .. (205,89.38) -- cycle ;
%Shape: Circle [id:dp9825058548947396] 
\draw  [color={rgb, 255:red, 0; green, 0; blue, 0 }  ,draw opacity=1 ][fill={rgb, 255:red, 0; green, 0; blue, 0 }  ,fill opacity=1 ][line width=0.75]  (205.5,112.03) .. controls (205.5,109.61) and (207.46,107.65) .. (209.88,107.65) .. controls (212.3,107.65) and (214.27,109.61) .. (214.27,112.03) .. controls (214.27,114.45) and (212.3,116.42) .. (209.88,116.42) .. controls (207.46,116.42) and (205.5,114.45) .. (205.5,112.03) -- cycle ;
%Shape: Circle [id:dp2917888990424313] 
\draw  [color={rgb, 255:red, 0; green, 0; blue, 0 }  ,draw opacity=1 ][fill={rgb, 255:red, 255; green, 255; blue, 255 }  ,fill opacity=1 ] (206,167.38) .. controls (206,164.96) and (207.96,163) .. (210.38,163) .. controls (212.8,163) and (214.77,164.96) .. (214.77,167.38) .. controls (214.77,169.8) and (212.8,171.77) .. (210.38,171.77) .. controls (207.96,171.77) and (206,169.8) .. (206,167.38) -- cycle ;
%Shape: Circle [id:dp8978719159795147] 
\draw  [color={rgb, 255:red, 126; green, 211; blue, 33 }  ,draw opacity=1 ][fill={rgb, 255:red, 255; green, 255; blue, 255 }  ,fill opacity=1 ] (206,191.38) .. controls (206,188.96) and (207.96,187) .. (210.38,187) .. controls (212.8,187) and (214.77,188.96) .. (214.77,191.38) .. controls (214.77,193.8) and (212.8,195.77) .. (210.38,195.77) .. controls (207.96,195.77) and (206,193.8) .. (206,191.38) -- cycle ;
%Rounded Rect [id:dp6218374822291778] 
\draw   (397,71.55) .. controls (397,66.28) and (401.28,62) .. (406.55,62) -- (435.21,62) .. controls (440.49,62) and (444.77,66.28) .. (444.77,71.55) -- (444.77,199.75) .. controls (444.77,205.02) and (440.49,209.3) .. (435.21,209.3) -- (406.55,209.3) .. controls (401.28,209.3) and (397,205.02) .. (397,199.75) -- cycle ;
\node[above] at (423,60) {$Q_{s_1}^-$} ;
%Shape: Circle [id:dp5309251233611854] 
\draw  [color={rgb, 255:red, 0; green, 0; blue, 0 }  ,draw opacity=1 ][fill={rgb, 255:red, 0; green, 0; blue, 0 }  ,fill opacity=1 ] (416,90.38) .. controls (416,87.96) and (417.96,86) .. (420.38,86) .. controls (422.8,86) and (424.77,87.96) .. (424.77,90.38) .. controls (424.77,92.8) and (422.8,94.77) .. (420.38,94.77) .. controls (417.96,94.77) and (416,92.8) .. (416,90.38) -- cycle ;
%Shape: Circle [id:dp7884202554494121] 
\draw  [color={rgb, 255:red, 0; green, 0; blue, 0 }  ,draw opacity=1 ][fill={rgb, 255:red, 255; green, 0; blue, 0 }  ,fill opacity=1 ] (416,122.38) .. controls (416,119.96) and (417.96,118) .. (420.38,118) .. controls (422.8,118) and (424.77,119.96) .. (424.77,122.38) .. controls (424.77,124.8) and (422.8,126.77) .. (420.38,126.77) .. controls (417.96,126.77) and (416,124.8) .. (416,122.38) -- cycle ;
%Shape: Circle [id:dp45239019303075834] 
\draw  [color={rgb, 255:red, 0; green, 0; blue, 0 }  ,draw opacity=1 ][fill={rgb, 255:red, 0; green, 0; blue, 0 }  ,fill opacity=1 ] (416.5,152.03) .. controls (416.5,149.61) and (418.46,147.65) .. (420.88,147.65) .. controls (423.3,147.65) and (425.27,149.61) .. (425.27,152.03) .. controls (425.27,154.45) and (423.3,156.42) .. (420.88,156.42) .. controls (418.46,156.42) and (416.5,154.45) .. (416.5,152.03) -- cycle ;
%Shape: Circle [id:dp1317167688176487] 
\draw  [color={rgb, 255:red, 126; green, 211; blue, 33 }  ,draw opacity=1 ][fill={rgb, 255:red, 255; green, 255; blue, 255 }  ,fill opacity=1 ] (416,190.38) .. controls (416,187.96) and (417.96,186) .. (420.38,186) .. controls (422.8,186) and (424.77,187.96) .. (424.77,190.38) .. controls (424.77,192.8) and (422.8,194.77) .. (420.38,194.77) .. controls (417.96,194.77) and (416,192.8) .. (416,190.38) -- cycle ;
%Rounded Rect [id:dp15074991388921066] 
\draw  [dash pattern={on 4.5pt off 4.5pt}] (326,71.55) .. controls (326,66.28) and (330.28,62) .. (335.55,62) -- (364.21,62) .. controls (369.49,62) and (373.77,66.28) .. (373.77,71.55) -- (373.77,199.75) .. controls (373.77,205.02) and (369.49,209.3) .. (364.21,209.3) -- (335.55,209.3) .. controls (330.28,209.3) and (326,205.02) .. (326,199.75) -- cycle ;
\node[above] at (350,60) {$Q_{s_1}^+$};
%Shape: Circle [id:dp21369505973285197] 
\draw  [color={rgb, 255:red, 0; green, 0; blue, 0 }  ,draw opacity=1 ][fill={rgb, 255:red, 255; green, 255; blue, 255 }  ,fill opacity=1 ] (345,105.38) .. controls (345,102.96) and (346.96,101) .. (349.38,101) .. controls (351.8,101) and (353.77,102.96) .. (353.77,105.38) .. controls (353.77,107.8) and (351.8,109.77) .. (349.38,109.77) .. controls (346.96,109.77) and (345,107.8) .. (345,105.38) -- cycle ;
%Shape: Circle [id:dp4939359652119273] 
\draw  [color={rgb, 255:red, 255; green, 0; blue, 0 }  ,draw opacity=1 ][fill={rgb, 255:red, 255; green, 255; blue, 255 }  ,fill opacity=1 ] (345.5,140.03) .. controls (345.5,137.61) and (347.46,135.65) .. (349.88,135.65) .. controls (352.3,135.65) and (354.27,137.61) .. (354.27,140.03) .. controls (354.27,142.45) and (352.3,144.42) .. (349.88,144.42) .. controls (347.46,144.42) and (345.5,142.45) .. (345.5,140.03) -- cycle ;
%Shape: Circle [id:dp3144157813020164] 
\draw  [color={rgb, 255:red, 0; green, 0; blue, 0 }  ,draw opacity=1 ][fill={rgb, 255:red, 255; green, 255; blue, 255 }  ,fill opacity=1 ] (346,177.38) .. controls (346,174.96) and (347.96,173) .. (350.38,173) .. controls (352.8,173) and (354.77,174.96) .. (354.77,177.38) .. controls (354.77,179.8) and (352.8,181.77) .. (350.38,181.77) .. controls (347.96,181.77) and (346,179.8) .. (346,177.38) -- cycle ;
%Straight Lines [id:da26472527639189025] 
\draw    (353.77,105.3) -- (420.38,90.38) ;
%Straight Lines [id:da8630856264406731] 
\draw [color={rgb, 255:red, 255; green, 0; blue, 0 }  ,draw opacity=1 ]   (354.27,140.03) -- (415.77,123.3) ;
%Straight Lines [id:da07250352254564896] 
\draw    (353.77,175.3) -- (417.77,125.3) ;
%Straight Lines [id:da37305019239192116] 
\draw    (352.77,108.3) -- (420.88,152.03) ;
%Straight Lines [id:da8171310884280023] 
\draw [color={rgb, 255:red, 116; green, 203; blue, 24 }  ,draw opacity=1 ] [dash pattern={on 4.5pt off 4.5pt}]  (214.77,191.38) -- (418.77,191.3) ;
%Straight Lines [id:da5659861225855164] 
\draw [color={rgb, 255:red, 255; green, 0; blue, 0 }  ,draw opacity=1 ]   (214.77,141.3) -- (345.5,140.03) ;
%Rounded Rect [id:dp3591593184590457] 
\draw  [dash pattern={on 4.5pt off 4.5pt}] (256,71.55) .. controls (256,66.28) and (260.28,62) .. (265.55,62) -- (294.21,62) .. controls (299.49,62) and (303.77,66.28) .. (303.77,71.55) -- (303.77,199.75) .. controls (303.77,205.02) and (299.49,209.3) .. (294.21,209.3) -- (265.55,209.3) .. controls (260.28,209.3) and (256,205.02) .. (256,199.75) -- cycle ;
\node[above] at (280,60) {$Q_{s_0}^+$};
%Shape: Circle [id:dp2645797555439222] 
\draw  [color={rgb, 255:red, 0; green, 0; blue, 0 }  ,draw opacity=1 ][fill={rgb, 255:red, 255; green, 255; blue, 255 }  ,fill opacity=1 ] (274,87.38) .. controls (274,84.96) and (275.96,83) .. (278.38,83) .. controls (280.8,83) and (282.77,84.96) .. (282.77,87.38) .. controls (282.77,89.8) and (280.8,91.77) .. (278.38,91.77) .. controls (275.96,91.77) and (274,89.8) .. (274,87.38) -- cycle ;
%Shape: Circle [id:dp12506516001871892] 
\draw  [color={rgb, 255:red, 0; green, 0; blue, 0 }  ,draw opacity=1 ][fill={rgb, 255:red, 255; green, 255; blue, 255 }  ,fill opacity=1 ] (274,113.38) .. controls (274,110.96) and (275.96,109) .. (278.38,109) .. controls (280.8,109) and (282.77,110.96) .. (282.77,113.38) .. controls (282.77,115.8) and (280.8,117.77) .. (278.38,117.77) .. controls (275.96,117.77) and (274,115.8) .. (274,113.38) -- cycle ;
%Straight Lines [id:da8945201824764507] 
\draw    (214.27,112.03) -- (274.77,113.3) ;
%Shape: Circle [id:dp7181417729235048] 
\draw  [color={rgb, 255:red, 0; green, 0; blue, 0 }  ,draw opacity=1 ][fill={rgb, 255:red, 255; green, 255; blue, 255 }  ,fill opacity=1 ] (274.5,154.27) .. controls (274.5,151.85) and (276.46,149.88) .. (278.88,149.88) .. controls (281.3,149.88) and (283.27,151.85) .. (283.27,154.27) .. controls (283.27,156.69) and (281.3,158.65) .. (278.88,158.65) .. controls (276.46,158.65) and (274.5,156.69) .. (274.5,154.27) -- cycle ;
%Shape: Circle [id:dp6122116141073073] 
\draw  [color={rgb, 255:red, 0; green, 0; blue, 0 }  ,draw opacity=1 ][fill={rgb, 255:red, 255; green, 255; blue, 255 }  ,fill opacity=1 ] (275.5,182.27) .. controls (275.5,179.85) and (277.46,177.88) .. (279.88,177.88) .. controls (282.3,177.88) and (284.27,179.85) .. (284.27,182.27) .. controls (284.27,184.69) and (282.3,186.65) .. (279.88,186.65) .. controls (277.46,186.65) and (275.5,184.69) .. (275.5,182.27) -- cycle ;
%Straight Lines [id:da8434842198941968] 
\draw    (213.77,89.3) -- (274.77,88.3) ;
%Straight Lines [id:da1271532103936539] 
\draw    (212.77,137.3) -- (275.77,90.3) ;
%Straight Lines [id:da5954888447066304] 
\draw    (212.77,115.3) -- (274.77,152.3) ;
%Straight Lines [id:da535383578138262] 
\draw    (211.77,144.3) -- (275.77,179.3) ;

\end{tikzpicture}
\caption{Illustration of the $\textnormal{V-BFS}_{approx}$ algorithm finding a path of length $d(s_0,s_1)+1$. Expanded vertices are solid disks, non-expanded vertices are hollow. The connection found by the algorithm between the two search trees is shown in red. The connection which yields the actual shortest path is in green. To find the actual shortest path and make the algorithm exact, it suffices to expand all vertices in $Q_{s_0}^-$ or in $Q_{s_1}^-$ (the $\textnormal{V-BFS}_{exact}$ algorithm picks the queue containing fewest vertices, in the above example $Q_{s_1}^-$) after finding the red connection.}
\label{fig:approximate-shortest-path}
\end{figure}

For the approximate-shortest-path algorithm $\textnormal{V-BFS}_{approx}(\mathcal{G},s_0,s_1)$, the detection of such a path marks the end of its execution, as the path is then immediately returned. And indeed, for vertices $s_0, s_1$ at distance $d(s_0,s_1)$, the returned path has length at most $d(s_0,s_1) + 1$ (see Lemma~\ref{lem:vertex-approx-correct}) and hence is at least an ``approximate" shortest path up to an additive increment of 1. This increment stems from the possibility that the vertex $u$ which is in $S_{\overline{s}}$ lies in $Q_{\overline{s}}^+$ instead of $Q_{\overline{s}}^-$, see Figure \ref{fig:approximate-shortest-path}. The exact algorithm hence checks whether $\Gamma(v)$ also intersects $Q^-_{\overline{s}}$, in which case we can return the corresponding path (lines~9-13 of Algorithm \ref{algo:vertex-exact}), which is guaranteed to be shortest. If this intersection is empty, one needs to check whether there is an edge that directly connects a vertex in $Q_{s}^-$ to a vertex in $Q_{\overline{s}}^-$, which would yield a path shorter by one hop compared to the path that was found in the first place. This is exactly what $\textnormal{V-BFS}_{exact}(\mathcal{G},s_0,s_1)$ does in lines~14-21. More precisely, it empties the smaller (line~14) of the two layers $Q_{s_i}^-$ that are currently being expanded. Note that the queue sizes - and not the sizes of the discovered sets $S_{s_i}$ - are relevant here, because the queues $Q_{s_i}^-$ contain the only vertices which might yield such a shorter path. If a shorter path is found, then this path is returned (line~21), otherwise the first found path is returned (line~25).

\begin{algorithm}[H]

\caption{{\sc Vertex-balanced BFS for exact shortest path $\textnormal{V-BFS}_{exact}(\mathcal{G},s_0,s_1)$}}

\label{algo:vertex-exact}

\begin{algorithmic}[1]
\State $Q_{s_0}^- \gets \mathtt{Queue}(s_0)$, $Q_{s_0}^+ \gets \mathtt{Queue}(\emptyset), S_{s_0} \gets \{s_0\}$
\State $Q_{s_1}^- \gets \mathtt{Queue}(s_1)$, $Q_{s_1}^+ \gets \mathtt{Queue}(\emptyset), S_{s_1} \gets \{s_1\}$ 
\While{$Q_{s_0}^- \neq\emptyset$ and $Q_{s_1}^- \neq\emptyset$}
  \State{$s \gets \argmin_{s_0,s_1}\{|S_{s_0}|, |S_{s_1}|\}$}
  \State{$v \gets \mathtt{pop}(Q_s^-)$}
  \State{$\mathtt{append}(Q_s^+, \Gamma(v) \setminus S_s)$}
  \State{$S_s \gets S_s \cup \Gamma(v)$}
  \If{$\Gamma(v) \cap S_{\overline{s}} \neq \emptyset$}
  \If{$\Gamma(v) \cap Q^-_{\overline{s}} \neq \emptyset$}
  \State{let $u'\in \Gamma(v) \cap Q^-_{\overline{s}}$}
  \State{let $\pi'$ be the path from $s$ to $u'$ in $S_s$ found by the BFS}
  \State{let $\overline{\pi}'$ be the path from $\overline{s}$ to $u'$ in $S_{\overline{s}}$ found by the BFS}
  \State\Return{$\pi' \cup \overline{\pi}'$}
  \EndIf
  \State{$p \gets \argmin_{s_0,s_1}\{|Q_{s_0}^-|, |Q_{s_1}^-|\}$}
  \While{$Q_p^- \neq\emptyset$}
  \State{$v' \gets \mathtt{pop}(Q_p^-)$}
  \If{$\Gamma(v') \cap Q_{\overline{p}}^- \neq \emptyset$}
  \State{let $u'\in \Gamma(v') \cap Q_{\overline{p}}^-$}
  \State{let $\pi'$ be the path from $p$ to $u'$ in $S_p$ found by the BFS}
  \State{let $\overline{\pi}'$ be the path from $\overline{p}$ to $u'$ in $S_{\overline{p}}$ found by the BFS}
  \State\Return{$\pi' \cup \overline{\pi}'$}
  \EndIf
  \EndWhile
  \State{let $u\in \Gamma(v) \cap S_{\overline{s}}$}
  \State{let $\pi$ be the path from $s$ to $u$ in $S_s$ found by the BFS}
  \State{let $\overline{\pi}$ be the path from $\overline{s}$ to $u$ in $S_{\overline{s}}$ found by the BFS}
  \State\Return{$\pi \cup \overline{\pi}$}
  \EndIf
  \If{$Q_s^- = \emptyset$}
  \State{$Q_s^- \gets Q_s^+$}
  \State{$Q_s^+ \gets \mathtt{Queue}(\emptyset)$}
  \EndIf
\EndWhile
\State\Return{$\emptyset$}
\end{algorithmic}
\end{algorithm}

\subsection{The edge-balanced algorithm}

\begin{algorithm}[H]

\caption{{\sc Edge-balanced BFS for approximate shortest path $\textnormal{E-BFS}_{approx}(\mathcal{G},s_0,s_1)$}}

\label{algo:edge-approx}

\begin{algorithmic}[1]

\State $Q_{s_0}^- \gets \mathtt{Queue}(s_0)$, $Q_{s_0}^+ \gets \mathtt{Queue}(\emptyset), S_{s_0} \gets \{s_0\}$
\State $Q_{s_1}^- \gets \mathtt{Queue}(s_1)$, $Q_{s_1}^+ \gets \mathtt{Queue}(\emptyset), S_{s_1} \gets \{s_1\}$

\State{$s\gets s_0$}

\While{$Q_{s_0}^- \neq\emptyset$ and $Q_{s_1}^-  \neq\emptyset$}
  \If{$|S_s|=1$}
  \State{$v_s \gets \mathtt{pop}(Q_s^-)$}
  \State{$E_{v_s} \gets \emptyset$} 
  \EndIf
  \If{$E(v_s) \setminus E_{v_s} = \emptyset$}
  \State{$v_s \gets \mathtt{pop}(Q_s^-)$}
  \State{$E_{v_s} \gets \emptyset$}  
  \EndIf
  \State{let $v_s u$ be an edge chosen u.a.r.\ from $E(v_s) \setminus E_{v_s}$}
  \State{$E_{v_s} \gets E_{v_s} \cup \{v_s u\}$}
  \If{$u \notin S_s$}
  \State{$\mathtt{append}(Q_s^+, \{u\})$}
  \State{$S_s \gets S_s \cup \{u\}$}
  \EndIf
  \If{$u\in S_{\overline{s}}$}
  \State{let $\pi$ be the path from $s$ to $u$ in $S_s$ found by the BFS}
  \State{let $\overline{\pi}$ be the path from $\overline{s}$ to $u$ in $S_{\overline{s}}$ found by the BFS}
  \State\Return{$\pi \cup \overline{\pi}$}
  \EndIf
  \If{$Q_s^- = \emptyset$ and $E(v_s) \setminus E_{v_s} = \emptyset$}
  \State{$Q_s^- \gets Q_s^+$}
  \State{$Q_s^+ \gets \mathtt{Queue}(\emptyset)$}
  \EndIf
  \State{$s \gets \overline{s}$}
\EndWhile
\State\Return{$\emptyset$}
\end{algorithmic}
\end{algorithm}

We now describe our edge-balanced algorithm. Consider again a graph $\mathcal{G}$ and the two starting vertices $s_0, s_1$ on which the algorithm is executed. First observe that the initialization and the stopping conditions (lines~1-2, 4 of Algorithm \ref{algo:edge-approx}) are essentially the same as in the vertex-balanced algorithms. Before executing the while-loop, the algorithm checks if $s_0$ or $s_1$ is isolated and returns $\emptyset$ if this is the case.\footnote{Suppressed in the pseudocode for simplicity.} Then the algorithm initializes, for each side $s_i$, the first vertex $v_s$ from which it will begin the search and the, still empty, set $E_{v_s}$ of edges incident to $v_s$ that have already been explored (lines~5-7)\footnote{Note that the if-statement on line~5 will be true exactly once for each side $s$, namely during the first two iterations of the while-loop}. The sets $E_{v_s}$, indexed by the corresponding vertex $v_s$ currently being expanded, will be used to keep track of the already explored edges incident to $v_s$ throughout the course of the algorithm. We denote the set of edges incident to a vertex $v$ by $E(v)$. Now, whenever a vertex $v_s$ has no more \emph{unexplored} incident edges, i.e.\ when $E(v_s) \setminus E_{v_s} = \emptyset$ (line~8), we select the next vertex to be expanded (line~9) and reinitialize the set of already explored edges to the empty set (line~10), since now a new vertex will be expanded. As long as a vertex being expanded still has edges to be explored, we select one of them uniformly at random (line~11) and add it to the set $E_{v_s}$ explored edges (line~12). If the vertex $u$ incident to the edge on the other end has not yet been discovered by the search on the same side, it is appended to the queue $Q_s^+$ and added to the ``discovered" set $S_s$ (lines~13-15). For a given execution of the algorithm, we henceforth refer to $v_s$ as the \textit{parent vertex} of $u$. If it has been discovered by the search on the other side, that is, is already contained in $S_{\overline{s}}$, then the two searches have met and the algorithm has found a path (line~16). The path can again be obtained by concatenating the path from $s_0$ to $u$ with the path from $u$ to $s_1$ (lines~17-19). This is precisely the path that the algorithm returns. If the search trees do not yet intersect, the algorithm checks if the vertex $v_s$ still has unexplored incident edges, and if not, it checks if there are remaining vertices in the queue $Q_s^-$ which can be expanded in the next iteration of the while-loop (line~20). Otherwise, it replenishes $Q_s^-$ by setting $Q_s^- \gets Q_s^+$ and resetting $Q_s^+ \gets \texttt{Queue}(\emptyset)$ (lines~21-22). After each iteration of the while-loop, the algorithm automatically switches between the searches (line~23), ensuring that the searches are balanced edge-by-edge.

\subsection{Notations and definitions}

During a run of the vertex-balanced algorithms (Algorithms \ref{algo:vertex-approx} and \ref{algo:vertex-exact}), we will denote by $t_i$ the number of \emph{iterations} performed (i.e.\ the number of vertices expanded) on the $s_i$-side, $i\in \{0,1\}$, and by $t$ the total number of iterations performed (note that $t=t_0+t_1$ by definition). We use the notations $Q^-_{s_i}(t_i)$, $Q^+_{s_i}(t_i)$, and $S_{s_i}(t_i)$ to denote the sets $Q^-_{s_i}$, $Q^+_{s_i}$, and $S_{s_i}$ after $t_i$ vertices have been expanded on side $s_i$, and we denote by $v^i_{t_i}$ the $t_i$-th vertex that is expanded on the $s_i$-side. Similarly, we write $Q^-_{s_i}(t)$, $Q^+_{s_i}(t)$, and $S_{s_i}(t)$ when we want to refer to those sets after $t$ vertices have been expanded in total. We now give the formal definition of the runtime of the algorithms.

\begin{definition}\label{def:cost}
Consider $\textnormal{V-BFS}_{approx}(\mathcal{G},s_0,s_1)$ or $\textnormal{V-BFS}_{exact}(\mathcal{G},s_0,s_1)$. For $i\in\{0,1\}$, consider the set $Q_{s_i}(t_i)$ of vertices that were expanded on the $s_i$-side up to (and including) iteration $t_i$, namely
\[
    Q_{s_i}(t_i) \coloneqq \Big(\bigcup_{t'=0}^{t_i-1} Q^-_{s_i}(t')\Big) \setminus Q^-_{s_i}(t_i).
\]
Note that by definition $|Q_{s_i}(t_i)| = t_i$. We define the $s_i$-\emph{side cost} of the algorithm after iteration $t_i$ as 
\[
    \mathcal{C}_{s_i}(t_i) \coloneqq \sum_{v\in Q_{s_i}(t_i)} \deg(v),
\]
and the \emph{cost} of the algorithm after iteration $t$ is then simply defined as $\mathcal{C}(t) \coloneqq \mathcal{C}_{s_0}(t_0) + \mathcal{C}_{s_1}(t_1)$.

\end{definition}

Note that the cost corresponds to the number of edges that were expanded up to and including iteration $t$. In particular, the total runtime of the algorithm is its cost once it has terminated.

Since in the $\textnormal{E-BFS}_{approx}(\mathcal{G},s_0,s_1)$ algorithm we alternate between the two sides \emph{while} expanding vertices, the concept of iteration is less relevant in that case. Instead, we simply refer to \emph{round} $k$ to describe the point at which we are exploring the $k$-th edge (in total) in the algorithm. Note that, for every even $k$, the number of edges expanded on each side is equal to exactly $k/2$ after round $k$ (since we are switching sides after each expansion). Similarly as for the vertex-balanced algorithms, we use the notations $Q^-_{s_i}(k)$, $Q^+_{s_i}(k)$, $S_{s_i}(k)$, and $v_{s_i}(k)$ to denote the sets $Q^-_{s_i}$, $Q^+_{s_i}$, $S_{s_i}$, and the vertex $v_{s_i}$ after round $k$.

\section{Graph Models}\label{sec:model}

In this section, we formally introduce Chung-Lu graphs and GIRGs. These will serve as host graphs for our algorithms, on which we prove our runtime bounds in Section~\ref{sec:proofs}. Throughout, we will consider undirected graphs with vertex set $\mathcal{V}=[n]$ and edge set denoted by $\mathcal{E}$. In Chung-Lu graphs and GIRGs, the weights and the degrees are distributed according to a power-law, which is defined as follows.

\begin{definition}\label{def:power-law}
    Let $\tau>1$. A discrete random variable $X$ is said to follow a \emph{power-law with exponent} $\tau$ if $\pr(X = x) = \Theta(x^{-\tau})$ for $x\in \N$. A continuous random variable $X$ is said to follow a \emph{power-law with exponent} $\tau$ if it has a density function $f_X$ satisfying $f_X(x) = \Theta(x^{-\tau})$ for $x\ge 1$.
\end{definition}

We start with the definition of the Chung-Lu model.

\begin{definition}[Chung-Lu graph~\cite{chung2002average}]\label{def:simple-chung-lu}
    Let $\tau>2$ and let $\mathcal{D}$ be a power-law distribution on $[1,\infty)$ with exponent $\tau$. A \emph{Chung-Lu graph} is obtained by the following two-step procedure:
    \begin{enumerate}[label=(\arabic*), leftmargin=1cm]
        \item Every vertex $v\in\mathcal{V}$ draws i.i.d.\ a \emph{weight} $W_v \sim \mathcal{D}$.

        \item For every two distinct vertices $u,v \in\mathcal{V}$, add the edge $uv\in\mathcal{E}$ independently with probability
        \begin{align*}
            \pr(uv\in\mathcal{E} \mid W_u, W_v) = \Theta\Big(\min\Big\{\frac{W_uW_v}{n}, 1\Big\}\Big),
        \end{align*}
        where the hidden constants are uniform over all $u,v$. 
    \end{enumerate}
\end{definition}

Classically, the $\Theta(\cdot)$ simply hides a factor 1, but introducing the $\Theta$ in the model also captures similar random graphs, like the Norros-Reittu model \cite{norros2006conditionally}, while important properties stay asymptotically invariant~\cite{janson2010asymptotic}.

Geometric Inhomogeneous Random Graphs (GIRGs) combine the degree inhomogeneity of Chung-Lu graphs with an underlying geometric space. The vertices are assigned both a weight and a position in the $d$-dimensional unit hypercube $[0,1]^d$ equipped with the torus topology.

\begin{definition}[GIRG~\cite{bringmann2019geometric}]\label{def:simple-girg}
    Let $\tau>2$, $\alpha>1$ and $d\in\mathbb{N}$ and let $\mathcal{D}$ be a power-law distribution on $[1,\infty)$ with exponent $\tau$. A \emph{Geometric Inhomogeneous Random Graph (GIRG)} is obtained by the following three-step procedure:
    \begin{enumerate}[label=(\arabic*), leftmargin=1cm]
        \item Every vertex $v\in\mathcal{V}$ draws i.i.d.\ a \emph{weight} $W_v \sim \mathcal{D}$.

        \item Every vertex $v\in\mathcal{V}$ draws independently and u.a.r.\ a position $x_v$ in the hypercube $[0,1]^d$.

        \item For every two distinct vertices $u,v \in\mathcal{V}$, add the edge $uv\in\mathcal{E}$ independently with probability
        \begin{align*}
            \pr(uv\in\mathcal{E} \mid W_u, W_v, x_u, x_v) = \Theta\Big(\min\Big\{\frac{W_uW_v}{n \|x_u-x_v\|^d}, 1\Big\}^{\alpha}\Big).
        \end{align*}
    \end{enumerate}
\end{definition}

\subsection{Basic Properties of Chung-Lu graphs and GIRGs}\label{sec:basic-properties}

In this subsection we describe some basic properties of GIRGs and Chung-Lu graphs. They hold in fact for a general class of graph models described in~\cite{bringmann2016average}. In particular, we give the expected degree of a vertex and the marginal probability that an edge between two vertices with given weights is present.

\begin{lemma}[Lemma~4.2 and Theorem~7.3 in~\cite{bringmann2016average}]\label{lem:marginal}
    Let $\mathcal{G}=(\mathcal{V},\mathcal{E})$ be a GIRG and $v \in\mathcal{V}$ be a vertex with fixed position $x_v \in [0,1]^d$. Then each edge $uv$ for $u \ne v$ is independently present with probability
\begin{align*}
    \Pr(uv\in\mathcal{E} \mid x_u, W_u = w_u, W_v = w_v) = \Theta\Big(\min\Big\{\frac{w_uw_v}{n}, 1\Big\}\Big). 
\end{align*}
In particular, we can remove the conditioning on $x_u$ (by integrating over $[0,1]^d$) and obtain
\begin{align*}
    \Pr(uv\in\mathcal{E} \mid W_u = w_u, W_v = w_v) = \Theta\Big(\min\Big\{\frac{W_uW_v}{n}, 1\Big\}\Big). 
\end{align*}
\end{lemma}

The next lemma says that the expected degree of a vertex is of the same order as its weight, thus allowing us to treat a given weight sequence $(W_v)_{v\in\mathcal{V}}$ as a sequence of expected degrees.

\begin{lemma}[Lemma~4.3 in~\cite{bringmann2016average}]\label{lem:expected-degree}
    Let $\mathcal{G}=(\mathcal{V},\mathcal{E})$ be a Chung-Lu graph or a GIRG and $v \in\mathcal{V}$ be a vertex with weight $w_v \in [1,\infty)$. Then we have $\E[\deg(v) \mid W_v = w_v]=\Theta(\min\{w_v,n\})$ and $\pr(uv\in\mathcal{E} \mid W_v = w_v) =\Theta(\min\{\tfrac{w_v}{n},1\})$ for a vertex $u$ chosen uniformly at random in $\mathcal{V}$.
\end{lemma}

The next lemma guarantees that GIRGs and Chung-Lu graphs are sparse.

\begin{lemma}\label{lem:constant-average-degree}
    Let $\mathcal{G}=(\mathcal{V},\mathcal{E})$ be a Chung-Lu graph or a GIRG and $u,v \in\mathcal{V}$ be vertices chosen uniformly at random. Then $\pr(uv\in\mathcal{E}) =\Theta(1/n)$.
\end{lemma}
\begin{proof}
By Lemma~\ref{lem:expected-degree} we have $\pr(uv\in\mathcal{E} \mid W_v = w) = \Theta(\min\{\frac{w}{n},1\})$. Since $W_v \ge 1$ deterministically by Definition \ref{def:power-law}, we directly get that  $\pr(uv\in\mathcal{E}) \ge \Omega(\frac{1}{n})$. On the other hand, $\pr(uv\in\mathcal{E} \mid W_v = w) \le O(\frac{w}{n})$, and hence $\pr(uv\in\mathcal{E}) \le O(\frac{\E[W_v]}{n}) = O(\frac{1}{n})$. 

\end{proof}

The next lemma relates the degree and weight of a vertex and shows concentration of a node's degree around its expectation.

\begin{lemma}[Lemma~4.3 in~\cite{bringmann2016average}]\label{lem:degree-concentration}
    Let $\mathcal{G}=(\mathcal{V},\mathcal{E})$ be a Chung-Lu graph or a GIRG. Then the following properties hold with high probability:
    \begin{enumerate}[label=(\roman*)]
    \item $\deg(v) = O(W_v + \log^2 n)$ for all $v \in\mathcal{V}$.
    \item $\deg(v)= (1+o(1))\E[\deg(v)]= \Theta(W_v)$ for all $v \in\mathcal{V}$ with $W_v = \omega(\log^2 n)$.
    \end{enumerate}
\end{lemma}

A corollary of the above lemma is the following results, which will be used many times in the proofs.

\begin{corollary}\label{cor:weight-concentration}
    Let $\mathcal{G}=(\mathcal{V},\mathcal{E})$ be a Chung-Lu graph or a GIRG. Then, with high probability, $W_v = O(\deg(v)+\log^2 n)$ for all $v \in\mathcal{V}$.
\end{corollary}
\begin{proof}
Let $v\in\mathcal{V}$ be an arbitrary vertex. If $W_v = O(\log^2 n)$, then we are trivially done. On the other hand, by Lemma \ref{lem:degree-concentration} whp for all vertices $v$ with $W_v = \omega(\log^2 n)$ we have $\deg(v)=\Theta(W_v)$, and in particular $W_v = O(\deg(v))$.
\end{proof}

The next lemma guarantees that there is at least one vertex with weight roughly $n^{\frac{1}{\tau-1}}$, but no vertices with weights much larger than this.

\begin{lemma}\label{lem:maxweight}
Let $\mathcal{G}=(\mathcal{V},\mathcal{E})$ be a GIRG or a Chung-Lu graph. Then with high probability it contains a vertex of weight at least $n^{\frac{1}{\tau-1}}/\log n$ but no vertex of weight $n^{\frac{1}{\tau-1}}\log n$ or larger.

\end{lemma}
\begin{proof}
   
    Let $X_i$ be the indicator variable that $v_i$ has weight at least $k$, for some positive integer $k$. Then, since the weights are i.i.d., we obtain
\begin{align*}
    \mathbb{P}(\sum_{i=1}^nX_i\geq 1)&=1-\mathbb{P}(\sum_{i=1}^nX_i=0)=1-\mathbb{P}(X_1=0)^n\\
    &=1-(1-\mathbb{P}(X_1\geq 1))^n\geq 1-e^{-n\mathbb{P}(X_1\geq 1)},
\end{align*}
where for the inequality we have used the classical bound $1+x\leq e^x$ (which is valid for every real number $x$). But $\mathbb{P}(X_1\geq 1)=\mathbb{P}(W_1\geq k)\geq ck^{-(\tau-1)}$ for some constant $c>0$ and hence 
\[1-e^{-n\mathbb{P}(X_1\geq 1)}\geq 1-e^{-nck^{-(\tau-1)}}.\]
Setting $k=n^{1/(\tau-1)}/\log n$ we finally obtain
\[\mathbb{P}(\sum_{i=1}^nX_i\geq 1)\geq 1-e^{-c(\log n)^{\tau-1}},\]
whence $\sum_{i=1}^nX_i\geq 1$ whp; that is, whp there is a vertex in the graph with weight at least $n^{1/(\tau-1)}/\log n$. Similarly, if we denote by $Y_i$ the indicator that node $v_i$ has weight larger than $k$, then by Markov's inequality we obtain 
\[
\pr(\sum_{i=1}^n Y_i\geq 1)\leq nk^{-(\tau-1)},
\]
and setting $k=n^{\frac{1}{\tau-1}}\log n$ yields that the above probability is at most $(\log n)^{-(\tau-1)}=o(1)$.
\end{proof}

The next theorem guarantees the existence of a unique giant component in Chung-Lu graphs and GIRGs, and bounds the size of the other components.

\begin{theorem}{(Theorem 5.9 in~\cite{bringmann2016average})}\label{thm:component-sizes}
    Let $\mathcal{G}=(\mathcal{V}, \mathcal{E})$ be a Chung-Lu graph or a GIRG on $n$ vertices. Then with high probability, there is a connected component which contains $\Omega(n)$ vertices. Furthermore, with high probability all other components have at most polylogarithmic size.
\end{theorem}

The next proposition\footnote{The proposition is contained in upcoming work by J.L., M.K., U.S.\ and a fourth co-author. Its proof is available on request.} gives the conditional weight density of vertices that are endpoints of sampled edges.

\begin{proposition}\label{prop:shifted-conditional-weight-distribution}
Let $\mathcal{G} = (\mathcal{V}, \mathcal{E})$ be a Chung-Lu graph or a GIRG. Denote by $f_{W_v}(w \mid W_u = w_u, uv\in\mathcal{E})$ the density of the weight of the endpoint of an edge $uv$ chosen uniformly at random in $\mathcal{E}$, conditioned on the weight of the other endpoint. Then this conditional density satisfies
\begin{align}\label{eq:shifted-conditional-weight-distribution}
    f_{W_v}(w \mid W_u = w_u, uv\in\mathcal{E})=
    \begin{dcases*}
    \Theta(w^{1-\tau}) & if $w_u \le  \frac{n}{w}$, \\
    \Theta(\tfrac{w^{-\tau} n}{w_u}) & if  $\frac{n}{w} < w_u \le n$, \\
    \Theta(w^{-\tau}) & if $w_u > n$.
    \end{dcases*}
\end{align}
\end{proposition}

\section{Proofs}\label{sec:proofs}

This section contain the proofs of our main theorems, demonstrating first correctness of our algorithms - valid for all graphs - followed by the runtime analyses, which hold for specific graph models (see Section \ref{sec:model}).

We begin with a definition which stratifies the vertex set of a graph by distance from a given starting vertex $s$.

\begin{definition}\label{def:layer}
    Let $\mathcal{G}=(\mathcal{V}, \mathcal{E})$ be a graph, let $s\in\mathcal{V}$ and let $j\in\N$. The set of vertices at distance $j$ from $s$ in the graph is called the \emph{$j$-th layer from $s$}, or simply the \emph{$j$-th layer} when $s$ is clear from context, and we denote it by $L_s^j$.
\end{definition}

\subsection{Correctness}\label{sec:correctness}

In this subsection, we will prove the correctness of our three algorithms. Let us first make a straightforward but useful observation, which consists of noticing that if we look at these algorithms only from one side (e.g.\ from the $s_0$-side) they are simply performing a standard Breadth-First Search. 

\begin{observation}
\label{obs:layers}
At any point in the algorithms $\textnormal{V-BFS}_{approx}$, $\textnormal{V-BFS}_{exact}$, and $\textnormal{E-BFS}_{approx}$, and for any source $s_i$, $i\in\{0,1\}$, the nodes in $Q_{s_i}^-$ are equidistant from the source $s_i$, and the same is true for the nodes in $Q_{s_i}^+$, whereas $d(s_i, Q_{s_i}^+) = d(s_i, Q_{s_i}^-)+1$. Hence, for all $t_i$ in the interval $\big(\sum_{k = 0}^{j-1} \vert L_{s_i}^k \vert  \sum_{k = 0}^{j} \vert L_{s_i}^k \vert\big]$, we have $Q_{s_i}^-(t_i) = L_{s_i}^j$, and $Q_{s_i}^+(t_i) \subseteq L_{s_i}^{j+1}.$ In particular, if a vertex in the $j+1$-th layer is expanded, then all vertices in the $j$-th layer have been expanded.

\end{observation}

We now prove a preliminary result that will be used to show the correctness of all three algorithms.
In what follows, we denote by $t^*$ the first iteration (for vertex-balanced algorithms) or round (for the edge-balanced algorithm) in which the stopping condition is met (i.e.\ $\Gamma(v) \cap S_{\Bar{s}} \neq \emptyset$ on line 8 for the vertex-balanced algorithms, and $u \in S_{\bar{s}}$ on line 16 for the edge-balanced algorithm).
\begin{lemma} \label{lem:dist-lb}
     Consider a run of either $\textnormal{V-BFS}_{approx}$, $\textnormal{V-BFS}_{exact}$, or $\textnormal{E-BFS}_{approx}$ on any graph $\mathcal{G}$ and any choice of $s_0$ and $s_1$. Let $d$ be the distance between $Q^-_{s}(t^*)$ and $s$ and let and $\bar{d}$ be the distance between $Q^-_{\bar{s}}(t^*)$ and $\bar{s}$. Then we have :
     $$
     d(s_0,s_1) \ge d + \Bar{d} + 1.
     $$
\end{lemma}

\begin{proof}
    Suppose by contradiction that there is an $s_0 - s_1$ path of length $D \le d + \Bar{d}$, and consider among all such paths the shortest one that we denote by $\Pi \coloneqq (s=u_0, u_1,\dots u_D=\Bar{s})$ (picking an arbitrary one in case there are several of equal minimal length). Because $\Pi$ is a shortest path, any subpath of $\Pi$ is also a shortest path (between its two own endpoints), and hence for $0 \le j \le D$, $u_j \in L_s^j \cap L_{\Bar{s}}^{D-j}$. Now, note that at iteration/round $t^*$, all layers up to $L_s^{d-1}$ on the $s$-side and $L_{\Bar{s}}^{\Bar{d} - 1}$ on the $\Bar{s}$-side have been fully expanded by Observation \ref{obs:layers} (that is, all endpoints of all edges that intersect with these layers have been discovered). If $D \le d$, then we have already expanded $u_{D-1}$, and $u_D = \Bar{s} \in S_{\Bar{s}} \cap \Gamma(u_{D-1})$, so the stopping condition is met and the algorithm should already have stopped at that point. Similarly, if $D > d$, then $u_{d-1}$ has been fully expanded, but $u_d \in \Gamma(u_{d-1}) \cap L_{\Bar{s}}^{D-d}$, and since $D-d \le \Bar{d}$ (because $D\le d+\Bar{d}$), it holds that $L_{\Bar{s}}^{D-d} \subseteq S_{\Bar{s}}$, so in particular $u_{d} \in \Gamma(u_{d-1}) \cap S_{\Bar{s}}$, and again, the algorithm should have stopped here. Hence we have reached a contradiction, which concludes the proof.
\end{proof}
The next lemma states that the algorithm $\textnormal{V-BFS}_{approx}(\mathcal{G},s_0,s_1)$ finds a path between $s_0$ and $s_1$ which is at most one hop longer than their distance.

\begin{lemma}\label{lem:vertex-approx-correct}
    For any graph $\mathcal{G}$ and any choice of $s_0$ and $s_1$, the length of the $s_0 - s_1$-path found by $\textnormal{V-BFS}_{approx}(\mathcal{G},s_0,s_1)$ is at most $d(s_0,s_1)+1$.
\end{lemma}

\begin{proof}
    Define $d$ and $\Bar{d}$ as in Lemma \ref{lem:dist-lb}. Any point in $S_{\Bar{s}}(t^*)$ is - or has been at some point - either in $Q_{\Bar{s}}^-(t^*)$ or in $Q_{\Bar{s}}^+(t^*)$, and consequently is at distance at most $\Bar{d} + 1$ from $\Bar{s}$. Hence, if we consider $u \in \Gamma(v) \cap S_{\Bar{s}(t^*)}$, then the distance from $s$ to $u$ is at most $d+1$ and the distance from $\Bar{s}$ to $u$ is at most $\Bar{d} + 1$, which implies that $d(s_0,s_1) = d(s,\Bar{s}) \le d + \Bar{d} + 2$. 
    By Lemma \ref{lem:dist-lb} we have $d(s_0,s_1) \ge d + \Bar{d} + 1$, which concludes the proof.
\end{proof}

The following lemma guarantees that the lengths of paths found by $\textnormal{V-BFS}_{exact}(\mathcal{G},s_0,s_1)$ indeed match the distance between vertices precisely.

\begin{lemma}\label{lem:vertex-exact-correct}
    For any graph $\mathcal{G}$ and any choice of $s_0$ and $s_1$, the length of the $s_0 - s_1$-path found by $\textnormal{V-BFS}_{exact}(\mathcal{G},s_0,s_1)$ is at most (and hence exactly) $d(s_0,s_1)$.
\end{lemma}
\begin{proof}
    Define $d$ and $\Bar{d}$ as in the two above lemmas. 
    Since the path found by $\textnormal{V-BFS}_{exact}(\mathcal{G},s_0,s_1)$ will be shorter or of the same length as the path found by $\textnormal{V-BFS}_{approx}(\mathcal{G},s_0,s_1)$, by Lemma \ref{lem:vertex-approx-correct} we know the length of the former is also upper bounded by $d + \Bar{d} + 2$.
    Moreover, by Lemma \ref{lem:dist-lb} we have $d(s_0,s_1) \ge d + \Bar{d} + 1$. Hence, we only need to distinguish two cases :
    \begin{itemize}
        \item There exists a shortest $s-\Bar{s}$ path of length exactly $d + \Bar{d} + 1$. Then, it must be of the form $P_{su} \cup uv \cup P_{v\Bar{s}}$, where $u \in L_s^d$, $v \in L_{\Bar{s}}^{\Bar{d}}$, and $P_{u'v'}$ denotes any shortest path between two vertices $u',v'$. Indeed, consider a shortest $s-\Bar{s}$ path $\Pi$, and consider the first node $v$ which is not in $S_s(t^*)$. In particular, its distance from $s$ is larger than or equal to $d + 1$, since $S_s$ contains all nodes at a distance less than $d$. If $v$ is not in $S_{\Bar{s}}$, then its distance to $\Bar{s}$ is larger than or equal to $\Bar{d} + 1$, so the total length of $\Pi$ is at least $d + \Bar{d} + 2$, which contradicts our assumption. If $v$ is not in $L_{\Bar{s}}^{\Bar{d}}$, it is in one of the previous layers of $S_{\Bar{s}}$, so its distance to $\Bar{s}$ is smaller than or equal to $\Bar{d}-1$, and the total length of $\Pi$ is at most $d+\Bar{d}$, which we already showed is not possible. This means that $\Pi$ must have the claimed form. Now, note that between the lines~9 and 21, the algorithm simply finishes to expand one of the sets $L_s^d$ or $L_{\Bar{s}}^{\Bar{d}}$, and stops if it finds an edge between the two. Hence, a path of length $d + \Bar{d} + 1$ will be found.
        \item There exists no $s_0 - s_1$-path of length $d + \Bar{d} + 1$. Then, $d(s_0,s_1) \ge d + \Bar{d} + 2$, and neither the condition on line~9 nor the condition on line~17 is ever met, so the algorithm returns the path $\pi \cup \Bar{\pi}$, which has length $d + \Bar{d} + 2 = d(s_0,s_1)$.
    \end{itemize}
\end{proof}

The final lemma in this subsection gives the analogue of Lemma~\ref{lem:vertex-approx-correct} for the algorithm $\textnormal{E-BFS}_{approx}(\mathcal{G},s_0,s_1)$, demonstrating that found paths are at most by an additive 1 longer than the distance between vertices.

\begin{lemma}\label{lem:edge-approx-correct}
    For any graph $\mathcal{G}$ and any choice of $s_0$ and $s_1$, the length of the $s_0 - s_1$-path found by $\textnormal{E-BFS}_{approx}(\mathcal{G},s_0,s_1)$ is at most $d(s_0,s_1)+1$.
\end{lemma}
\begin{proof}
   The proof is the exact same as for Lemma \ref{lem:vertex-approx-correct}, replacing the definition of the stopping condition by $u \in S_{\bar{s}}$.
\end{proof}

\subsection{Runtime}\label{sec:runtime}

In the following, we prove the claimed runtime bounds for our algorithms. We will first treat the case where both starting vertices lie in the giant component. We begin with a key lemma on the density of weights of vertices discovered by the algorithms, which will be used ubiquitously throughout the article.

\begin{lemma}\label{lem:bfs-new-weight}
Let $\mathcal{G} = (\mathcal{V}, \mathcal{E})$ be a Chung-Lu graph and let $s\in\mathcal{V}$ be any vertex. Let $\textnormal{BFS}(\mathcal{G},s)$ denote a (unidirectional) Breadth-First-Search in $\mathcal{G}$ started at $s$. We additionally denote the sum of degrees of the first $t$ expanded vertices by $C(t)\coloneqq \sum_{i=1}^t \deg(v_i)$, where $v_i$ is the $i$-th vertex expanded by $\textnormal{BFS}(\mathcal{G},s)$. Consider some iteration $t+1$ of the process during which we are expanding vertex $u\coloneqq v_{t+1}\in\mathcal{V}$ of weight $W_u = w_u$, uncovering the neighbors of $u$ in some random order (where by \emph{uncovering} we mean adding them to the $Q^+_s$ queue, line 6 of the $\textnormal{V-BFS}_{approx}(\mathcal{G},s_0,s_1)$ and $\textnormal{V-BFS}_{exact}(\mathcal{G},s_0,s_1)$ algorithms and line 14 of the $\textnormal{E-BFS}_{approx}(\mathcal{G},s_0,s_1)$ algorithm). Whenever we are uncovering some $v\in\Gamma(u)$, denote the density of the weight $W_v$ of $v$, conditioned on this vertex $v$ being discovered for the first time in the BFS\footnote{In other words, conditioned on $u$ being the parent vertex of $v$.}, by $f_{W_v}(w \mid W_u = w_u, \textnormal{BFS}^v(\mathcal{G},s))$. Let $\eps>0$ be an arbitrarily small constant. Then, as long as $C(t)=O(n^{1-\eps})$, we have
\begin{align}\label{eq:bfs-new-weight-upper-bound}
    f_{W_v}(w \mid W_u = w_u, \textnormal{BFS}^v(\mathcal{G},s)) = O(w^{1-\tau}).
\end{align}
Moreover, for all $w$ satisfying $w\cdot C(t) = O(n^{1-\eps})$, we additionally have
\begin{align}\label{eq:bfs-new-weight}
    f_{W_v}(w \mid W_u = w_u, \textnormal{BFS}^v(\mathcal{G},s)) =
    \begin{dcases*}
    \Theta(w^{1-\tau}) & if $w_u \le  \frac{n}{w}$, \\
    \Theta(\tfrac{w^{-\tau} n}{w_u}) & if  $\frac{n}{w} < w_u \le n$, \\
    \Theta(w^{-\tau}) & if $w_u > n$.
    \end{dcases*}
\end{align}
\end{lemma}
\begin{proof}
Let $T^u \coloneqq \{v_1, \ldots, v_t\}$ denote the set of vertices expanded by $\textnormal{BFS}(\mathcal{G},s)$ up to (but excluding) $u = v_{t+1}$. The condition that $v$ is discovered for the first time from vertex $u$ by the BFS is equivalent to the event $\{uv \in \mathcal{E}, u'v \notin \mathcal{E} \, \forall u'\in T^u\}$. Hence, using Bayes' Theorem, we have
\begin{align}
\begin{split}\label{eq:bfs-new-weight-bayes}
    &f_{W_v}(w \mid W_u = w_u, \textnormal{BFS}^v(\mathcal{G},s))
    = f_{W_v}(w \mid W_u = w_u, uv \in \mathcal{E}, \forall u'\in T^u: u'v \notin \mathcal{E}) \\
    &\qquad= \frac{\pr(\forall u'\in T^u: u'v \notin \mathcal{E} \mid W_u=w_u, W_v=w, uv\in\mathcal{E}) \cdot f_{W_v}(w \mid W_u = w_u, uv\in\mathcal{E})}{\pr(\forall u'\in T^u: u'v \notin \mathcal{E}  \mid  W_u = w_u, uv\in\mathcal{E})}.
\end{split}
\end{align}

Note that by Markov's inequality
\begin{align}
\begin{split}\label{eq:prob-not-connected-tree}
    &\pr(\forall u'\in T^u: u'v \notin \mathcal{E} \mid  W_u = w_u, uv\in\mathcal{E})
    = 1 - \pr\Big(\bigcup_{u'\in T^u} \{u'v\in\mathcal{E}\} \mid  W_u = w_u, uv\in\mathcal{E}\Big) \\
    &\qquad=1 - \pr\Big(\sum_{u'\in T^u}^{}\mathbbm{1}_{\{u'v\in\mathcal{E}\}} \geq 1\mid  W_u = w_u, uv\in\mathcal{E}\Big)\\
    &\qquad\ge 1 - \E[|\Gamma(v) \cap T^u| \mid W_u = w_u, uv\in\mathcal{E}] 
    =1 - \E[|\Gamma(v) \cap T^u|],
\end{split}
\end{align}
where the last equality is due to the fact that the number of neighbors of $v$ connecting to vertices in $T^u$ does not depend neither on the weight of $u\notin T^u$, nor on the fact that $uv$ is present in the graph.
Similarly,
\begin{align}\label{eq:cond-prob-not-connected-tree}
    &\pr(\forall u'\in T^u : u'v \notin \mathcal{E} \mid  W_u = w_u, W_v=w, uv\in\mathcal{E}) \ge 1 - \E[|\Gamma(v) \cap T^u| \mid W_v=w].
\end{align}
Using Lemma \ref{lem:expected-degree} and Corollary \ref{cor:weight-concentration}, whp we have
\begin{align*}
    \E[|\Gamma(v) \cap T^u|] &= \sum_{i=1}^t\pr(v_iv\in\mathcal{E})
    \le \sum_{i=1}^t O\Big(\tfrac{W_{v_i}}{n}\Big) \\
    &\le \frac{1}{n}\sum_{i=1}^t O(\deg(v_i)+\log^2 n)
    \le O\Big(\frac{\log^2 n}{n}\sum_{i=1}^t\deg(v_i)\Big).
\end{align*}
By assumption we know that $\sum_{i=1}^t\deg(v_i)=C(t)=O(n^{1-\eps})$, and hence
\[
\E[|\Gamma(v) \cap T^u|] \le O(n^{-\eps}\log^2 n) = o(1).
\]
Thus, using equation \eqref{eq:prob-not-connected-tree} we get 
\begin{align}\label{eq:prob-first-encountered-constant}
    \pr(\forall u'\in T^u: u'v \notin \mathcal{E}  \mid  W_u = w_u, uv\in\mathcal{E}) =\Theta(1),
\end{align}
and in particular we can bound the density using equation \eqref{eq:bfs-new-weight-bayes} by
\begin{align*}
    &f_{W_v}(w \mid W_u = w_u, \textnormal{BFS}^v(\mathcal{G},s)) \\
    &\qquad= \Theta\big(\pr(\forall u'\in T^u: u'v \notin \mathcal{E} \mid W_u=w_u, W_v=w, uv\in\mathcal{E}) \cdot f_{W_v}(w \mid W_u = w_u, uv\in\mathcal{E})\big) \\
    &\qquad\le O\big(f_{W_v}(w \mid W_u = w_u, uv\in\mathcal{E})\big)
    \le O(w^{1-\tau}),
\end{align*}
where the last inequality follows from Proposition \ref{prop:shifted-conditional-weight-distribution} since the expressions on the right-hand side of equation \eqref{eq:shifted-conditional-weight-distribution} are all upper bounded by $O(w^{1-\tau})$. This shows that \eqref{eq:bfs-new-weight-upper-bound} holds, and we now turn to proving \eqref{eq:bfs-new-weight}, assuming that $w\cdot C(t) = O(n^{1-\eps})$. Since this is a stronger assumption than $C(t)=O(n^{1-\eps})$ (as $w\ge 1$), equation \eqref{eq:prob-first-encountered-constant} still holds. Moreover, using Lemma \ref{lem:marginal} and Corollary \ref{cor:weight-concentration}, whp we have 
\begin{align*}
    \E[|\Gamma(v) \cap T^u| \mid W_v = w] &= \sum_{i=1}^t\pr(v_iv\in\mathcal{E} \mid W_v = w)
    \le \sum_{i=1}^t O\Big(\tfrac{wW_{v_i}}{n}\Big) \\
    &\le \frac{w}{n}\sum_{i=1}^t O(\deg(v_i)+\log^2 n)
    \le O\Big(\frac{w\log^2 n}{n}\sum_{i=1}^t\deg(v_i)\Big).
\end{align*}
Since in this case we have $w \cdot \sum_{i=1}^t\deg(v_i)=w\cdot C(t)=O(n^{1-\eps})$, we deduce that
\[
\E[|\Gamma(v) \cap T^u| \mid W_v = w] \le O(n^{-\eps}\log^2 n) = o(1),
\]
and hence by equation \eqref{eq:cond-prob-not-connected-tree} we get 
\begin{align*}
    \pr(\forall u'\in T^u: u'v \notin \mathcal{E} \mid  W_u = w_u, W_v = w, uv\in\mathcal{E}) =\Theta(1).
\end{align*}
Plugging equation \eqref{eq:prob-not-connected-tree} together with the above equation in \eqref{eq:bfs-new-weight-bayes} then yields
\begin{align*}
    f_{W_v}(w \mid W_u = w_u, \textnormal{BFS}^v(\mathcal{G},s)) = \Theta(f_{W_v}(w \mid W_u = w_u, uv\in\mathcal{E})),
\end{align*}
and Proposition \ref{prop:shifted-conditional-weight-distribution} then concludes the proof of \eqref{eq:bfs-new-weight}.
\end{proof}

We now state a lemma that gives us concentration of the largest weight in a sequence of weights associated to the vertices we encounter in a BFS. Both parts of the lemma are consequences of Lemma~\ref{lem:bfs-new-weight}. The upper bound follows immediately by integrating the upper bound on the conditional density given in equation~\eqref{eq:bfs-new-weight-upper-bound}. The lower bound is again derived by integration but requires additional arguments bounding the vertex weights in the graph.

\begin{lemma}\label{lem:max-weight-shifted-distribution}
    Let $\mathcal{G}$ be a Chung-Lu graph on $n$ vertices. Consider a random set $Q$ of vertices whose weights follow the conditional distribution given by equation \eqref{eq:bfs-new-weight-upper-bound} and such that $|Q|\rightarrow \infty$ with high probability. Then we have
     \begin{align}\label{eq:max-weight-shifted-upper-bound}
        \max_{v\in Q} W_{v} \le |Q|^{\frac{1}{\tau-2}}\log|Q|
    \end{align}
    with high probability. Let $\eps>0$ be arbitrarily small and suppose additionally that the cost of the BFS, up to the iteration where $Q$ is sampled, is at most $n^{\frac{\tau-2}{\tau-1}+\eps}$, and that $\max_{v\in Q} W_v \le n^{\frac{1}{2}+\eps}$. Then, with high probability we have
    \begin{align}\label{eq:max-weight-shifted-lower-bound}
        \max_{v\in Q} W_{v}  \ge  |Q|^{\frac{1}{\tau-2}}/\log|Q|.
    \end{align}
\end{lemma}

\begin{proof}
For any $v\in Q$, using equation \eqref{eq:bfs-new-weight-upper-bound} we have that
\begin{align*}
    \pr(W_v > |Q|^{\frac{1}{\tau-2}}\log|Q|\mid Q)
    = \int_{|Q|^{\frac{1}{\tau-2}}\log|Q|}^{\infty} O(w^{1-\tau}) dw
    = O(|Q|^{-1}(\log|Q|)^{-(\tau-2)}),
\end{align*}
where for the second step we used that $\tau>2$. A simple union bound then yields
\begin{align*}
    \pr\big(\max_{v\in Q} W_{v}& > |Q|^{\frac{1}{\tau-2}}\log|Q|\big)\\
    &\leq  \mathbb{E}\big[\mathbbm{1}_{\{|Q|=\omega(1)\}}\pr(\exists v\in Q: W_v > |Q|^{\frac{1}{\tau-2}}\log|Q|\mid Q)\big]+o(1)\\
    &\leq \mathbb{E}\big[\mathbbm{1}_{\{|Q|=\omega(1)\}}\sum_{v\in Q}\pr(W_v > |Q|^{\frac{1}{\tau-2}}\log|Q|\mid Q)\big]+o(1)\\
    &\leq \mathbb{E}\big[\mathbbm{1}_{\{|Q|=\omega(1)\}}O((\log|Q|)^{-(\tau-2)})\big]+o(1)=o(1).
\end{align*}
Now assume additionally that the cost of the BFS is upper bounded by $n^{\frac{\tau-2}{\tau-1}+\eps}$, and that $\max_{v\in Q} W_v \le n^{\frac{1}{2}+\eps}$. Since $\frac{\tau-2}{\tau-1}+\frac{1}{2} < \frac{1}{2}+\frac{1}{2}=1$ (as $\tau<3$) and $\eps$ is small, the additional assumptions of Lemma \ref{lem:bfs-new-weight} are satisfied, and in particular the weight of the vertices in $Q$ follow the conditional distribution given by \eqref{eq:bfs-new-weight}. Thanks to Lemma \ref{lem:maxweight} we know that whp (as $n\rightarrow\infty$) no vertex in the graph has weight larger than $n^{1/(\tau-1)+o(1)}=o(n)$. For the rest of the proof, we condition on this event, which implies in particular that the last case of equation \eqref{eq:bfs-new-weight} never occurs, and thus for all $v\in Q$ the conditional distribution of $W_v$ is given by $\Theta(w^{-\tau}\cdot\min\{w,\frac{n}{W_u}\})$. Using Corollary \ref{cor:weight-concentration} and our upper bound on the cost of the BFS, we get that whp $W_u \le O(\deg(u) + \log^2 n) \le O(n^{\frac{\tau-2}{\tau-1}+\eps})$. Hence, for all $v\in Q$, we have $W_v W_u \le O(n^{\frac{1}{2}+\frac{\tau-2}{\tau-1}+2\eps}) = o(n)$, so that the conditional distribution of $W_v$ is given by $\Theta(w^{1-\tau})$.

Therefore,
\begin{align*}
    \pr\big(\max_{v\in Q} W_{v}  \le |Q|^{\frac{1}{\tau-2}}/\log|Q|\big)
    &\le \mathbb{E}\Big[\mathbbm{1}_{\{|Q|=\omega(1)\}}\Big(1- \int_{|Q|^{\frac{1}{\tau-2}}/\log|Q|}^{\infty} \Theta(w^{1-\tau}) \Big)^{|Q|}\Big] +o(1)\\
    &= \mathbb{E}\Big[\mathbbm{1}_{\{|Q|=\omega(1)\}}\big(1-\Omega(|Q|^{-1}(\log|Q|)^{\tau-2})\big)^{|Q|}\Big]+o(1)\\
    & \le \mathbb{E}\Big[\mathbbm{1}_{\{|Q|=\omega(1)\}}\exp(-\Omega((\log|Q|)^{\tau-2}))\Big]+o(1),
\end{align*}
which shows that \eqref{eq:max-weight-shifted-lower-bound} holds whp.
\end{proof}

\subsubsection{Approximate Vertex-Balanced BBFS}
In this subsection, we delve into the runtime analysis of the vertex-balanced approximate algorithm. Note that we will reuse some of these results when treating the vertex-balanced \textit{exact} algorithm in the following subsection. We begin by observing that the cost incurred on side $s_i$ at iteration $t_i$ can be lower-bounded by $|S_{s_i}(t_i)|-1$.

\begin{observation}\label{obs:cost>=seen}
     Let $\mathcal{G}$ be a graph and let $s_0$ and $s_1$ be two vertices of $\mathcal{G}$. Let $i\in\{0, 1\}$, let $t_i\in\N$, and consider the algorithm $\textnormal{V-BFS}_{approx}(\mathcal{G},s_0,s_1)$ after expanding $t_i$ vertices on the $s_i$-side. Then deterministically 
     \begin{align*}
         \mathcal{C}_{s_i}(t_i) \ge  |S_{s_i}(t_i)|-1.
     \end{align*}
     Indeed, the set $S_{s_i}$ is initialized with one vertex, and then grown by adding the (undiscovered) neighbors of vertices being expanded, and each such neighbor contributes a term of 1 in the sum $\sum_{v\in Q_{s_i}(t_i)} \deg(v) \eqqcolon \mathcal{C}_{s_i}(t_i)$. 
   
\end{observation}

We now prove a result that bounds the cost $\mathcal{C}_{s_i}(t_i)$ in terms of the number of iterations $t_i$. This is done by relating $\mathcal{C}_{s_i}(t_i)$ to the maximum weight among expanded vertices, and using Lemma \ref{lem:max-weight-shifted-distribution} to bound this weight.

\begin{lemma}\label{lem:cost-expanded-relation}
    Let $\mathcal{G}$ be a Chung-Lu graph on $n$ vertices, let $s_0$ and $s_1$ be two vertices of $\mathcal{G}$ chosen uniformly at random and fix $i\in\{0,1\}$. Let $t_i = \omega(1)$ be such that $\mathcal{C}_{s_i}(t_i) \le n^{(\tau-2)/(\tau-1)+o(1)}$. Then, with high probability, we have 
    \begin{align*}
        \mathcal{C}_{s_i}(t_i) \le \Big(t_i^{\tfrac{1}{\tau-2}} + t_i\log^2 n\Big)\cdot t_i^{o(1)}.
    \end{align*}
    Moreover, if $s_i$ is not an isolated vertex, then we also have that $\mathcal{C}_{s_i}(t_i) \ge t_i$ deterministically.
\end{lemma}
\begin{proof}
For the second half of the lemma, notice that if $s_i$ is not an isolated vertex, then any $v\in Q_{s_i}(t_i)$ also has at least one neighbor, and therefore
\begin{align*}
    \mathcal{C}_{s_i}(t_i) = \sum_{v\in Q_{s_i}(t_i)}\deg(v) \ge \sum_{v\in Q_{s_i}(t_i)} 1 = |Q_{s_i}(t_i)| = t_i,
\end{align*}
where the last equality is a direct consequence of the definition of $Q_{s_i}(t_i)$, see Definition \ref{def:cost}.

We now turn to the main part of the lemma, namely the upper bound of $\mathcal{C}_{s_i}(t_i)$. Since $\mathcal{C}_{s_i}(t_i) \le n^{(\tau-2)/(\tau-1)+o(1)}$, we have in particular that $\deg(v) \le n^{(\tau-2)/(\tau-1)+o(1)}$ for all $v\in Q_{s_i}(t_i)$. Hence, by Corollary \ref{cor:weight-concentration} we know that whp $W_v \le n^{(\tau-2)/(\tau-1)+o(1)}$ for all $v\in Q_{s_i}(t_i)$, and we will assume this for the rest of the proof. In particular by Lemma \ref{lem:bfs-new-weight}, except for the very first vertex that is put in $Q^-_{s_i}$ (namely $s_i$ itself), the weight of each vertex in $Q_{s_i}(t_i)$ follows the (conditional) distribution given by equation \eqref{eq:bfs-new-weight}, and moreover the additional assumptions of Lemma \ref{lem:max-weight-shifted-distribution} are satisfied by the set $Q_{s_i}(t_i)\setminus \{s_i\}$.
The additional assumptions of Lemma \ref{lem:bfs-new-weight} are satisfied since $\big(n^{(\tau-2)/(\tau-1)+o(1)}\big)^2 = O(n^{1-\eps})$ (as $\tau<3$), and this also implies that we are always in the first case of equation \eqref{eq:bfs-new-weight}, i.e., conditioned on $W_v \le n^{(\tau-2)/(\tau-1)+o(1)}$ for all $v\in Q_{s_i}(t_i)$, the weights $(W_{t'})_{t'=2}^{t_i}$ of the vertices in $Q_{s_i}(t_i) \setminus \{s_i\}$ are distributed independently with density $f_{W_{t'}}(w) = \Theta(w^{1-\tau})$.

Therefore, the sum $\sum_{t'=2}^{t_i}W_{t'}$ of the weights of all vertices in $Q_{s_i}(t_i) \setminus \{s_i\}$ is a sum of $t_i-1$ independent random variables following a power-law with exponent $\tau-1$. On the one hand, by Theorem 1 in \cite{omelchenko2019concentration}, we have for any $\eps>0$ that 
\begin{align*}
    \pr\Big(\sum_{t'=2}^{t_i}W_{t'} \ge t_i^{\tfrac{1}{\tau-2}+\eps}\Big) = O(t_i^{-(\tau-2)\eps}),
\end{align*}
consequently $\sum_{t'=2}^{t_i}W_{t'} \le t_i^{1/(\tau-2)+o(1)}$ whp (as $t_i \rightarrow \infty$). On the other hand, by Lemma \ref{lem:max-weight-shifted-distribution}, we know that whp (as $t_i \rightarrow \infty$)
\begin{align}\label{eq:max-weight-t}
    \max_{2 \le t' \le t_i}W_{t'} = t_i^{\tfrac{1}{\tau-2} \pm o(1)}.
\end{align}
Hence whp we have
\begin{align}\label{eq:sum-weight<=max-weight}
    \sum_{t'=2}^{t_i}W_{t'} \le t_i^{1/(\tau-2)+o(1)}
    \leq  t_i^{o(1)}\cdot \max_{2 \le t' \le t_i}W_{t'}.
\end{align}
Using Lemma \ref{lem:degree-concentration} together with equations \eqref{eq:max-weight-t}-\eqref{eq:sum-weight<=max-weight}, we deduce that whp (as $t_i \rightarrow \infty$)
\begin{align*}
     \mathcal{C}_{s_i}(t_i) &= \sum_{v\in Q_{s_i}(t_i)}\deg(v)
     \le \sum_{v\in Q_{s_i}(t_i)}O(W_v + \log^2 n) \\
     &\le t_i^{o(1)}\cdot \max_{2 \le t' \le t_i}W_{t'} + O(t_i \log^2 n)
     \le \Big(t_i^{\tfrac{1}{\tau-2}} + t_i\log^2 n\Big)\cdot t_i^{o(1)},
\end{align*}
which concludes the proof.
\end{proof}

The subsequent lemma relates the cost on side $s_i$ at time $t_i$, the maximum degree in $Q_{s_i}(t_i)$ and the size of $S_{s_i}(t_i)$ after a large enough polylogarithmic number of iterations $t_i$, showing that these three quantities are basically of the same order.

\begin{lemma}\label{lem:cost=seen}
    Let $\mathcal{G}$ be a Chung-Lu graph on $n$ vertices and let $s_0$ and $s_1$ be two vertices of $\mathcal{G}$ chosen uniformly at random. Let $i\in\{0, 1\}$, let $t_i\ge (\log n)^{2(\tau-2)/(3-\tau)}$, and consider the algorithm $\textnormal{V-BFS}_{approx}(\mathcal{G},s_0,s_1)$ after expanding $t_i$ vertices on the $s_i$-side such that $\mathcal{C}_{s_i}(t_i) \le n^{(\tau-2)/(\tau-1)+o(1)}$ holds. Then with high probability we have 
    \begin{align*}
        |S_{s_i}(t_i)|-1 \le \mathcal{C}_{s_i}(t_i) \le \max_{v\in Q_{s_i}(t_i)}\deg(v) \cdot t_i^{o(1)} \le |S_{s_i}(t_i)| \cdot t_i^{o(1)}.
    \end{align*} 
    
\end{lemma}

\begin{proof}
The lower bound holds deterministically by Observation \ref{obs:cost>=seen}, and hence we only need to show that $\mathcal{C}_{s_i}(t_i) \le \max_{v\in Q_{s_i}(t_i)}\deg(v) \cdot t_i^{o(1)} \le |S_{s_i}(t_i)| \cdot t_i^{o(1)}$ whp. By Lemma \ref{lem:cost-expanded-relation}, whp we have
\begin{align}\label{eq:cost-expanded-relation}
    \mathcal{C}_{s_i}(t_i) \le \Big(t_i^{\tfrac{1}{\tau-2}} + t_i\log^2 n\Big)\cdot t_i^{o(1)}.
\end{align}
Moreover, equation \eqref{eq:max-weight-t} also holds whp (as $t_i\rightarrow\infty$), namely we have 
\begin{align}\label{eq:max-weight-t-2}
\max_{2 \le t' \le t_i}W_{t'} = t_i^{1/(\tau-2) \pm o(1)}.
\end{align}
Together with the assumption $t_i \ge (\log n)^{2(\tau-2)/(3-\tau)}$, this guarantees that $\max_{2 \le t' \le t_i}W_{t'} = \omega(\log^2 n)$ since $\tau\in (2,3)$. Therefore we can use Lemma \ref{lem:degree-concentration} to deduce that whp
\begin{align}\label{eq:max-weight<max-degree}
    \max_{v\in Q_{s_i}(t_i)}\deg(v) = \Omega\Big(\max_{2 \le t' \le t_i}W_{t'}\Big).
\end{align}
Moreover, using the assumption $t_i \ge (\log n)^{2(\tau-2)/(3-\tau)}$ together with equation \eqref{eq:cost-expanded-relation} yields (as $t_i\log^2(n)\leq t_i^{\tfrac{1}{\tau-2}}$ iff $t_i\geq (\log n)^{2(\tau-2)/(3-\tau)}$) 
\begin{align*}
    \mathcal{C}_{s_i}(t_i) \le t_i^{\tfrac{1}{\tau-2}+o(1)},
\end{align*}
and combining this with equations \eqref{eq:max-weight-t-2}-\eqref{eq:max-weight<max-degree} gives
\[
\mathcal{C}_{s_i}(t_i) \le t_i^{\tfrac{1}{\tau-2}+o(1)}\leq t_i^{o(1)} \cdot \max_{2 \le t' \le t_i}W_{t'}\leq t_i^{o(1)} \cdot \max_{v\in Q_{s_i}(t_i)}\deg(v).
\]
Finally, we clearly have $\Gamma(v) \subseteq S_{s_i}(t_i)$ for all $v\in Q_{s_i}(t_i)$ by definition, which implies that $\max_{v\in Q_{s_i}(t_i)}\deg(v) \le |S_{s_i}(t_i)|$. Therefore, we see that whp
\begin{align*}
    \mathcal{C}_{s_i}(t_i) \le t_i^{o(1)} \cdot \max_{v\in Q_{s_i}(t_i)}\deg(v) \le |S_{s_i}(t_i)| \cdot t_i^{o(1)},
\end{align*}
which concludes the proof.
\end{proof}

Next, we state and prove a crucial balancing lemma. It guarantees that, after a polylogarithmic ``burn-in" phase and until late in the execution of the algorithm, the cost on both sides of the algorithm differ at most by a subpolynomial factor. The reason is that the balancing is done according to $|S_{s_i}(t_i)|$, and we have seen in the previous lemma that this size is of the same order as the cost $\mathcal{C}_{s_i}(t_i)$.

\begin{lemma}\label{lem:balanced-cost}
Let $\mathcal{G}$ be a Chung-Lu graph on $n$ vertices and let $s_0$ and $s_1$ be two vertices of $\mathcal{G}$ chosen uniformly at random. Consider the algorithm $\textnormal{V-BFS}_{approx}(\mathcal{G},s_0,s_1)$. Let $\eps>0$ be arbitrarily small, let $2(\log n)^{2(1+\eps/(\tau-2-\eps))/(3-\tau)} \le k \le n^{(\tau-2)/(\tau-1)+O(\eps)}$ and let $t\in\N$ be such that the cost of the algorithm after iteration $t=t_0+t_1$ satisfies $\mathcal{C}(t)=k$. Then, with high probability, we have that $\mathcal{C}_{s_i}(t_i) \le t^{\eps} \cdot \mathcal{C}_{s_{1-i}}(t_{1-i})$, $i\in\{0,1\}$.
\end{lemma}

\begin{proof}
Since $\eps>0$ is arbitrarily small, we can and do assume in particular that $\eps<\tau-2$. By symmetry and a simple union bound, it is enough to show that $\mathcal{C}_{s_0}(t_0)\le t^{\eps}\cdot\mathcal{C}_{s_{1}}(t_{1})$ whp. Suppose that $\mathcal{C}_{s_0}(t_0)> t^{\eps}\cdot\mathcal{C}_{s_{1}}(t_{1})$. In particular this implies that $\mathcal{C}_{s_0}(t_0)\ge \mathcal{C}_{s_{1}}(t_{1})$ and hence 
$\mathcal{C}_{s_0}(t_0)\ge k/2$ since $\mathcal{C}_{s_0}(t_0) + \mathcal{C}_{s_{1}}(t_{1}) = \mathcal{C}(t)=k$.

We claim that $t_0 \ge \mathcal{C}_{s_0}(t_0)^{\tau-2-\eps}$ whp. Indeed, if $t_0 < \mathcal{C}_{s_0}(t_0)^{\tau-2-\eps}$, then 
\begin{align*}
\Big(t_0^{\tfrac{1}{\tau-2}} + t_0\log^2 n\Big)\cdot t_0^{\eps} 
< \mathcal{C}_{s_0}(t_0)^{1-\eps/(\tau-2)+\eps(\tau-2-\eps)} + \mathcal{C}_{s_0}(t_0)^{(1+\eps)(\tau-2-\eps)}\log^2 n.  
\end{align*}
Using the inequality $\tau-2<1$, together with the fact that $(1+\eps)(\tau-2-\eps)$ is decreasing in $\varepsilon$ (whence it can be made smaller that $\tau-2$ by taking a small enough $\varepsilon$), this yields
\begin{align}\label{eq:iteration-cost-mismatch}
\Big(t_0^{\tfrac{1}{\tau-2}} + t_0\log^2 n\Big)\cdot t_0^{\eps} 
< \mathcal{C}_{s_0}(t_0)^{1-\Omega(\eps)} + \mathcal{C}_{s_0}(t_0)^{\tau-2}\log^2n
< \mathcal{C}_{s_0}(t_0),
\end{align}
where for the second inequality we have used that $\mathcal{C}_{s_0}(t_0) \ge k/2 \ge (\log n)^{2/(3-\tau)+\eps}$ by assumption. However, since 
\[
\mathcal{C}_{s_0}(t_0) \le \mathcal{C}(t) = k \le n^{(\tau-2)/(\tau-1)+O(\eps)}
\]
by assumption (and $\eps$ is arbitrarily small), Lemma \ref{lem:cost-expanded-relation} tells us that \eqref{eq:iteration-cost-mismatch} occurs with vanishing probability. Hence we must have $t_0 \ge \mathcal{C}_{s_0}(t_0)^{\tau-2-\eps}$ whp, and we will assume that such an inequality holds for the rest of the proof. In particular, since
$\mathcal{C}_{s_0}(t_0) \ge k/2 \ge (\log n)^{2(1+\eps/(\tau-2-\eps))/(3-\tau)}$ (again by our assumption on $k$), this implies that $t_0 \ge (\log n)^{2(\tau-2)/(3-\tau)}$ and hence the assumptions of Lemma \ref{lem:cost=seen} with $i=0$ are fulfilled. Additionally, note that we also obtain a lower bound of the form $t_0 \ge k^{\Omega(1)}$.

We now apply Lemma \ref{lem:cost=seen} together with $\mathcal{C}_{s_0}(t_0)> t^{\eps}\cdot\mathcal{C}_{s_{1}}(t_{1})$ (which we assumed at the start of the argument) to deduce that
\begin{align*}
    |S_{s_0}(t_0)| \ge \mathcal{C}_{s_0}(t_0) \cdot t_0^{-\varepsilon/2}
    \ge \mathcal{C}_{s_1}(t_1) \cdot t^{\eps/2},
\end{align*}
and by Observation \ref{obs:cost>=seen} this yields
\begin{align}\label{eq:lower-bound-seen0}
    |S_{s_0}(t_0)| \ge t^{\eps/2} \cdot (|S_{s_1}(t_1)|-1).
\end{align}
We claim that $|S_{s_1}(t_1)|\geq 2$. To see this, note that $|S_{s_1}(t_1)|=1$ if and only if $t_1=0$ and hence, due to line 4 of the algorithm $\textnormal{V-BFS}_{approx}(\mathcal{G},s_0,s_1)$, after two iterations of the while-loop we must have $t_0, t_1 \ge 1$ (except if one of the vertices $s_0, s_1$ is isolated, which we excluded), which implies that $|S_{s_0}(t_0)|,|S_{s_1}(t_1)| \ge 2$. In particular, the lower bound on $t_0 = \omega(1)$ implies that $t_1 \ge 1$ and hence $|S_{s_1}(t_1)| \ge 2$. Therefore, we deduce from \eqref{eq:lower-bound-seen0} and $t=\omega(1)$ that
\begin{align*}
    |S_{s_0}(t_0)| \ge t^{\eps/2} \cdot |S_{s_1}(t_1)|/2 \ge t^{\eps/3} \cdot |S_{s_1}(t_1)|,
\end{align*}
which implies

\begin{align}\label{eq:lower-bound-seen-difference}
    |S_{s_0}(t_0)|-|S_{s_1}(t_1)| \ge t^{\eps/4} \cdot |S_{s_1}(t_1)|.
\end{align}

Since in line 4 of the algorithm $\textnormal{V-BFS}_{approx}(\mathcal{G},s_0,s_1)$ we are choosing the side with the smallest $S_{s_i}$ set, a difference of $j$ between $S_{s_0}$ and $S_{s_1}$ can only happen if we are at an iteration where we are expanding a vertex of degree at least $j$. Hence, \eqref{eq:lower-bound-seen-difference} implies that $v^0_{t_0}$, the $t_0$-th vertex being expanded on the $s_0$-side, has degree 
\begin{align*}
    \deg(v^0_{t_0}) \ge t^{\eps/4} \cdot |S_{s_1}(t_1)|
    \ge k^{\Omega(\eps)},
\end{align*}
where the second inequality uses $t \ge t_0 \ge k^{\Omega(1)}$ (and $|S_{s_1}(t_1)|\ge 1$). Hence we have shown that
\[
\pr(\mathcal{C}_{s_0}(t_0)\ge t^{\eps}\cdot\mathcal{C}_{s_{1}}(t_{1})) \le \pr(\deg(v^0_{t_0}) \ge k^{\Omega(\eps)}).
\]
Thus, by the lower bound on $k$ in the assumptions, it remains to show that 
\begin{align}\label{eq:prob-degree-v-t0-large}
    \pr(\deg(v^0_{t_0}) \ge (\log n)^{\Omega(\eps)} ) = o(1).
\end{align}
Note that since $\mathcal{C}(t)=k \le n^{(\tau-2)/(\tau-1)+O(\eps)} \le O(n^{1-\eps})$, we can use equation \eqref{eq:bfs-new-weight-upper-bound} in Lemma \ref{lem:bfs-new-weight} to bound the density of the weight $W^*$ of $v^0_{t_0}$; in particular, we have $\pr(W^* \ge w)=O(w^{-(\tau-2)})$ for all $w\in[1,\infty)$.
By a simple union bound we then have
\begin{align*}
    \pr(\deg(v^0_{t_0}) \ge (\log n)^{\Omega(\eps)}) \le \pr(\deg(v^0_{t_0}) \ge (\log n)^{\Omega(\eps)} \mid W^* \le \log\log n) + \pr(W^*\ge \log\log n).
\end{align*}
The first term can be upper bounded by $2^{-(\log n)^{\Omega(\eps)}}=o(1)$ using Chernoff bounds, while the second term is upper bounded by $O((\log\log n)^{-(\tau-2)})=o(1)$. This shows that \eqref{eq:prob-degree-v-t0-large} holds and concludes the proof.
\end{proof}

The following lemma shows that the contributions to the cost $\mathcal{C}_{s_i}(t_i)$ are ``heavy-tailed". In other words, a few vertices in $Q_{s_i}(t_i)$ have degrees (or weights) that are almost as large as the total cost. This is a consequence of the very heavy tail of the distribution of the weights encountered during the BFS (see Lemma \ref{lem:bfs-new-weight}).

\begin{lemma}\label{lem:expansion-vertex-approx-algo}
    Let $\mathcal{G}$ be a Chung-Lu graph on $n$ vertices and let $s_0$ and $s_1$ be two vertices of $\mathcal{G}$ chosen uniformly at random.
    Let $i\in\{0, 1\}$, let $\eps>0$ be arbitrarily small, let $2(\log n)^{2(1+\eps/(\tau-2-\eps))/(3-\tau)} \le  k \le n^{\frac{\tau-2}{\tau-1}+O(\eps)}$ and consider the algorithm $\textnormal{V-BFS}_{approx}(\mathcal{G},s_0,s_1)$
   
    after expanding $t_i$ vertices on the $s_i$-side such that $k^{1+\eps} \le \mathcal{C}_{s_i}(t_i) \le k^{1+O(\eps)}$.
    Then, with high probability, the algorithm has expanded at least $\log k$ vertices of weight in $[k,2k]$ on the $s_i$-side. 
   
\end{lemma}

\begin{proof}
Denote by $T_k$ the first iteration (of the algorithm) at which the $s_i$-side cost passes the threshold $k^{1+\eps}$, i.e.\ 
$T_k\coloneqq \min\{t'\in \mathbb{N}: \mathcal{C}_{s_i}(t')\geq k^{1+\eps}\}$.
We start by showing that $T_k \ge k^{\tau-2}\log^2 k$ whp. Suppose that $T_k < k^{\tau-2}\log^2 k$. Since $\eps$ is arbitrarily small and $k\le n^{\frac{\tau-2}{\tau-1}+O(\eps)}$, the assumptions of Lemma \ref{lem:cost-expanded-relation} are fulfilled and hence we deduce that whp
\begin{align*}
    \mathcal{C}_{s_i}(T_k)
    \le \Big(T_k^{\tfrac{1}{\tau-2}} + T_k\log^2 n\Big)\cdot T_k^{o(1)}
    < \Big(k(\log k)^{2/(\tau-2)} + k^{\tau-2}\log^2 k\log^2 n\Big)\cdot k^{o(1)}.
\end{align*}
By assumption we know that $\log^2 n < k^{3-\tau}$, and hence we obtain
\begin{align*}
    \mathcal{C}_{s_i}(T_k)
    < \Big(k(\log k)^{2/(\tau-2)} + k\log^2 k\Big)\cdot k^{o(1)}
    = k^{1+o(1)},
\end{align*}
which contradicts the lower bound $\mathcal{C}_{s_i}(T_k) \ge k^{1+\eps}$. Thus, we have $T_k \ge k^{\tau-2}\log^2 k$ whp and we will assume this lower bound for the rest of the proof.

For $1 \le t' \le T_k$, denote by $B'_{t'}$ the the event that at iteration $t'$ on the $s_i$-side the algorithm expands a vertex of weight in $[k,2k]$ and write $B_{t'}$ for the corresponding indicator random variable, i.e.\ $B_{t'}\coloneqq \mathbbm{1}_{B'_{t'}}$. Note that, if $\sum_{t'=1}^{T_k}B_{t'}\geq \log k$, then obviously by iteration $T_k$ the algorithm has expanded at least $\log k$ vertices with weight in $[k,2k]$. Hence, the goal is to show that $\sum_{t'=1}^{T_k}B_{t'}\geq \log k$ whp as $n$ (and hence $k$) tends to $\infty$. 

Consider the $t'$-th vertex $v_{t'}$ that is expanded on the $s_i$-side and denote by $u_{t'}$ the parent vertex of $v_{t'}$. Later on we will need an upper bound on the weights of the vertices $u_{t'}$; to this end, we begin by splitting
\begin{align*}
\pr\Big(\sum_{t'=1}^{T_k}B_{t'}< \log k\Big)
\leq \pr\Big(\Big\{\sum_{t'=1}^{T_k}B_{t'}< \log k\Big\}\cap \bigcap_{t'\leq T_k} \{W_{u_{t'}} \le n/2k\}\Big)
+\pr\Big(\bigcup_{t'\leq T_k} \{W_{u_{t'}} > n/2k\}\Big).
\end{align*}
Using Corollary \ref{cor:weight-concentration}, we can bound
\begin{equation}\label{TTT}
    \pr\Big(\bigcup_{t'\leq T_k} \{W_{u_{t'}} > n/2k\}\Big)
    \leq \pr\Big(\bigcup_{t'\leq T_k} \{n/2k < W_{u_{t'}} \le O(\deg(u_{t'})+ \log^2 n)\} \Big)+o(1).
\end{equation}
As noted above, $u_{t'}$ was expanded at some iteration $t''\leq t'-1$ on the $s_i$-side and hence at iteration $t'$ we know everything about $u_{t'}$, i.e.\ we know its degree and its weight. Observe that, by definition of $T_k$, the degree of $u_{t'}$ is at most $k^{1+\eps}$ for all $t'\leq T_k$. Hence, since $k>(\log n)^{2/(3-\tau)}>\log^2 n$, we see that 
\[
\deg(u_{t'})+ \log^2 n \leq O(k^{1+\eps}) \le  o(n/k),
\]
where the second inequality holds because for $\eps>0$ small enough we have
\[
k^{2+\eps}  \le \big(n^{(\tau-2)/(\tau-1)+O(\eps)}\big)^{2+\eps} \le o\big(n^{(\tau-2)/(\tau-1)+1/(\tau-1)}\big) = o(n),
\]
where we used the inequality $\tau<3$. This shows that $\{n/2k < W_{u_{t'}} \le O(\deg(u_{t'}) + \log^2 n)\}$
is an empty event for $t'\leq T_k$, and hence by \eqref{TTT} we have 
\begin{align*}
    \pr\Big(\bigcup_{t'\leq T_k} \{W_{u_{t'}} > n/2k\}\Big) = o(1).
\end{align*}
Thus it only remains to show that
\[
\pr\Big(\Big\{\sum_{t'=1}^{T_k}B_{t'}<\log k \Big\}\cap \bigcap_{t'\leq T_k} \{W_{u_t'} \le n/2k\}\Big) = o(1).
\]
To this end, let us denote by $\mathcal{F}_j$ the information gathered by the algorithm until  and including the $j$-th iteration on the $s_i$-side, and observe that the event $\{W_{u_{t'}} \le n/2k\}$ is $\mathcal{F}_{t'-1}$-measurable (as it was expanded at some step $t''\leq t'-1$). Using our assumption on $\mathcal{C}_{s_i}(t_i)$ together with Lemma \ref{lem:balanced-cost}, we know that $\mathcal{C}(T_k) \le n^{\tfrac{\tau-2}{\tau-1}+O(\eps)}$. Since we also have $k\le n^{\tfrac{\tau-2}{\tau-1}+O(\eps)}$, we can use equation \eqref{eq:bfs-new-weight} in Lemma \ref{lem:bfs-new-weight} to obtain

\begin{align}\label{eq:P(B')}
    \pr(B'_{t'} \mid \mathcal{F}_{t'-1}) = \int_{k}^{2k}\Theta(w^{1-\tau})dw = \Theta(k^{-(\tau-2)}).
\end{align}
Now, since for any non-negative random variable $Y\geq 0$ and any event $\mathcal{H}$ we have, by Markov's inequality (for $h>0$)
\[
\pr(\{Y\geq h\}\cap\mathcal{H})\leq \frac{\mathbb{E}[Y\mathbbm{1}_{\mathcal{H}}]}{h},
\]
we obtain
\begin{align*}
    &\pr\Big(\Big\{\sum_{t'=1}^{T_k}B_{t'} < \log k\Big\}\cap \bigcap_{t'\leq T_k} \{W_{u_{t'}} \le n/2k\}\Big) \\
    &\qquad= \pr\Big(\Big\{\exp\big(-\textstyle\sum_{t'=1}^{T_k}B_{t'}\big) > 1/k\Big\}\cap \bigcap_{t'\leq T_k} \{W_{u_{t'}} \le n/2k\}\Big) \\
    &\qquad\le k\mathbb{E}\Big[\exp\big(-\textstyle\sum_{t'=1}^{T_k}B_{t'}\big)\mathbbm{1}_{\bigcap_{t'\leq T_k} \{W_{u_{t'}} \le n/2k\}}\Big].
\end{align*}
With the purpose of bounding the last expectation from above, we can write, using the tower property of conditional expectation
\begin{align}
\begin{split}\label{condexp}
    &\mathbb{E}\Big[\exp\big(-\textstyle\sum_{t'=1}^{T_k}B_{t'}\big)\mathbbm{1}_{\bigcap_{t'\leq T_k} \{W_{u_{t'}} \le n/2k\}}\Big]\\
    &\qquad= \mathbb{E}\Big[\exp\big(-\textstyle\sum_{t'=1}^{T_k-1}B_{t'}\big)\mathbbm{1}_{\bigcap_{t'\leq T_k} \{W_{u_{t'}} \le n/2k\}}\mathbb{E}[\exp(-B_{T_k}) \mid \mathcal{F}_{T_k-1}]\Big].
\end{split}
\end{align}
Now using \eqref{eq:P(B')} as well as the classical inequality  $1+x\le e^{x}$ (which is valid for every real $x$), we obtain (on $\{W_{u_{T_k}} \le n/2k\}$)
\[
\mathbb{E}[\exp(-B_{T_k}) \mid \mathcal{F}_{T_k-1}]
=1-(1-e^{-1})\pr(B'_{T_k} \mid \mathcal{F}_{T_k-1})
\le e^{-(1-e^{-1})\pr(B'_{T_k} \mid \mathcal{F}_{T_k-1})}
\le e^{-\Theta(k^{-(\tau-2)})}.
\]
Substituting this estimate back into \eqref{condexp} and dropping the event $\{W_{u_{T_k}} \le n/2k\}$ from the intersection we arrive at
\begin{align*}
    &\mathbb{E}\Big[\exp\big(-\textstyle\sum_{t'=1}^{T_k}B_{t'}\big)\mathbbm{1}_{\bigcap_{t'\leq T_k} \{W_{u_{t'}} \le n/2k\}}\Big]\\
    &\qquad\le \mathbb{E}\Big[\exp\big(-\textstyle\sum_{t'=1}^{T_k-1}B_{t'}\big) \mathbbm{1}_{\bigcap_{t'\leq T_k-1} \{W_{u_{t'}} \le n/2k\}} e^{-\Theta(k^{-(\tau-2)})}\Big].
\end{align*}
By iterating we obtain
\[
\mathbb{E}\Big[\exp\big(-\textstyle\sum_{t'=1}^{T_k}B_{t'}\big)\mathbbm{1}_{\bigcap_{t'\leq T_k} \{W_{u_{t'}} \le n/2k\}}\Big] \le e^{-\Theta(T_k k^{-(\tau-2)})} = e^{-\Omega(\log^2 k)},
\]
where the last line holds since we are assuming that $T_k \ge k^{\tau-2}\log^2 k$. Therefore we arrive at
\[
\pr\Big(\Big\{\sum_{t'=1}^{T_k}B_{t'} < \log k\Big\}\cap \bigcap_{t'\leq T_k} \{W_{u_{t'}} \le n/2k\}\Big)
\le k e^{-\Omega(\log^2 k)} = o(1),
\]
which concludes the proof.
\end{proof}

The next lemma pinpoints a stopping condition for the vertex-balanced algorithms. It states that, as soon as \textit{each} side has expanded a logarithmic number of vertices of weight slightly larger than $n^{(\tau-2)/(\tau-1)}$, then whp the two search trees intersect. As the proof shows, it will be enough to consider the vertex of maximal weight in the graph as the common neighbor which lies in this intersection.

\begin{lemma}\label{lem:meeting-vertex-approx-algo}
    Let $\mathcal{G}$ be a Chung-Lu graph on $n$ vertices and let $s_0$ and $s_1$ be two vertices of $\mathcal{G}$ chosen uniformly at random.
    Let $\eps>0$, and suppose the algorithm $\textnormal{V-BFS}_{approx}(\mathcal{G},s_0,s_1)$ or $\textnormal{V-BFS}_{exact}(\mathcal{G},s_0,s_1)$ after iteration $t$ has expanded $\Omega(\log n)$ vertices of weight in the interval $[n^{(\tau-2)/(\tau-1)+\eps}, 2n^{(\tau-2)/(\tau-1)+\eps}]$ on each side. Then $S_{s_0}(t) \cap S_{s_1}(t) \neq \emptyset$ with high probability.
\end{lemma}

\begin{proof}
    We observe that $S_{s_0}(t)$ intersects $S_{s_1}(t)$ in particular if an expanded vertex on the $s_0$-side has a common neighbor with some vertex expanded from the $s_1$-side. We will lower-bound this probability by considering as a potential common neighbor only a vertex $v_{max}$ of maximal weight $W_{\max}$ in the graph. By Lemma \ref{lem:maxweight}, we know that $W_{\max}=n^{\frac{1}{\tau-1} -\frac{\varepsilon}{2}}$ whp. 
    Denote by $U_{s_0}(t)$ and $U_{s_1}(t)$ the respective set of vertices on each side with  weight in $[n^{(\tau-2)/(\tau-1)+\eps}, 2n^{(\tau-2)/(\tau-1)+\eps}]$ that have been expanded up to iteration $t$. By assumption we know that $\min\{|U_{s_0}(t)|, |U_{s_1}(t)|\} \ge \Omega(\log n)$. Denote the elements of these sets by $u^1_{s_0}, u^2_{s_0}, \dots$ and $u^1_{s_1}, u^2_{s_1}, \dots$ respectively. Now for $U_{s_0}(t)$ (and analogously for $U_{s_1}(t)$), it holds for for each $i \in \{1, \dots, \min\{|U_{s_0}(t)|, |U_{s_1}(t)|\}\}$, for some small enough constant $ c >0$, that 
    \[  \Pr(u^i_{s_0}v_{max}\in\mathcal{E} \mid W_{u^i_{s_0}}, W_{{max}}) = \Theta\Big(\min\Big\{\frac{W_{u^i_{s_0}}W_{{max}}}{n}, 1\Big\}\Big) = \Theta\Big(\min\Big\{\frac{n^{\frac{(\tau-2)+1}{\tau-1} -\frac{\varepsilon}{2} + \varepsilon}}{n}, 1\Big\}\Big) \ge c. \]
    Now for any such $i$, a pair of vertices $u^i_{s_0}, u^i_{s_1}$ has $v_{max}$ as a common neighbor if and only if both $u^i_{s_0}v_{max}$ and $u^i_{s_1}v_{max}$ are present in the edge set $\mathcal{E}$. Due to the independence of these two events conditioned on $W_{u^i_{s_0}}, W_{u^i_{s_1}}, W_{{max}}$, this happens with probability
    \[\Pr(u^i_{s_0}v_{max}\in\mathcal{E} \mid W_{u^i_{s_0}}, W_{{max}})\cdot \Pr(u^i_{s_1}v_{max}\in\mathcal{E} \mid W_{u^i_{s_1}}, W_{{max}}) \ge c^2.\]
    Therefore we can upper-bound the probability that $S_{s_0}(t) \cap S_{s_1}(t) = \emptyset$ by
    \[
    \Pr(S_{s_0}(t) \cap S_{s_1}(t) = \emptyset) \le (1-c^2)^{\Omega(\log n)}=o(1),
    \]
    which concludes the proof.
\end{proof}

The preceding results suffice to prove our first main result, the runtime bound for the approximate vertex-balanced algorithm.

\begin{theorem}\label{thm:vertex-approx}
    Let $\mathcal{G}$ be a Chung-Lu graph on $n$ vertices and let $s_0$ and $s_1$ be two vertices of $\mathcal{G}$ chosen uniformly at random. Assuming that $s_0$ and $s_1$ are in the giant component of $\mathcal{G}$, $\textnormal{V-BFS}_{approx}(\mathcal{G},s_0,s_1)$ is a correct algorithm and with high probability it terminates in time at most $n^{(\tau-2)/(\tau-1)+o(1)}$.
\end{theorem}

\begin{proof}
Since $s_0$ and $s_1$ both lie in the giant component, there exists at least one path connecting $s_0$ and $s_1$ and such a path will eventually be found by the algorithm. Any path output by $\textnormal{V-BFS}_{approx}(\mathcal{G},s_0,s_1)$ will have length at most $d(s_0,s_1)+1$ by Lemma~\ref{lem:vertex-approx-correct}. Hence all that remains to prove is that a path is indeed found in the claimed runtime.

Let $\eps>0$ be arbitrarily small. We start by showing that whp there is some iteration $t$ such that $\mathcal{C}(t) \in [n^{(\tau-2)/(\tau-1)+4\eps}, n^{(\tau-2)/(\tau-1)+5\eps}]$. Defining $t$ as the first iteration for which $\mathcal{C}(t) \ge n^{(\tau-2)/(\tau-1)+4\eps}$, this boils down to showing that whp $\mathcal{C}(t) \le n^{(\tau-2)/(\tau-1)+5\eps}$. Denote by $v_t$ the $t$-th vertex that is expanded by the algorithm, which is the vertex that makes the cost cross the threshold $n^{(\tau-2)/(\tau-1)+4\eps}$. By definition of $t$, we know that $\mathcal{C}(t-1)<n^{(\tau-2)/(\tau-1)+4\eps}$ and $\deg(v_t) \ge n^{(\tau-2)/(\tau-1)+4\eps}-\mathcal{C}(t-1)$. Denoting by $\mathcal{F}$ the intersection of these two events, we need to show that
\begin{align*}
    \pr(\deg(v_t) \ge n^{(\tau-2)/(\tau-1)+5\eps}-\mathcal{C}(t-1) \mid \mathcal{F}) = o(1). 
\end{align*}
On $\mathcal{F}$, we have $n^{(\tau-2)/(\tau-1)+5\eps}-\mathcal{C}(t-1) \ge n^{(\tau-2)/(\tau-1)+5\eps}/2$, hence it suffices to show that
\begin{align*}
    \pr(\deg(v_t) \ge n^{(\tau-2)/(\tau-1)+5\eps}/2 \mid \mathcal{F}) = o(1).
\end{align*}
By Lemma \ref{lem:degree-concentration}, it is actually enough to prove that
\begin{align*}
    \pr(W_{v_t} \ge C^{-1}n^{(\tau-2)/(\tau-1)+5\eps} \mid \mathcal{F}) = o(1),
\end{align*}
for some large constant $C$. Note that by Lemma \ref{lem:bfs-new-weight}, $W_{v_t}$ follows the distribution given by \eqref{eq:bfs-new-weight-upper-bound}, with the extra conditioning on $\mathcal{F}$. This distribution can be stochastically dominated by replacing the conditioning on $\mathcal{F}$ with a conditioning on the event $\mathcal{F}' \coloneqq \{W_{v_t} \ge \Omega(n^{(\tau-2)/(\tau-1)+4\eps})\}$ (where we also used Lemma \ref{lem:degree-concentration} again). Using equation \eqref{eq:bfs-new-weight-upper-bound} and the definition of conditional probability, we easily get
\begin{align*}
    \pr(W_{v_t} \ge C^{-1}n^{(\tau-2)/(\tau-1)+5\eps} \mid \mathcal{F}') = \Theta(n^{-(\tau-2)\eps}) = o(1),
\end{align*}
which shows that indeed whp if the cost passes the threshold, then it does not exceed $n^{(\tau-2)/(\tau-1)+5\eps}$. If the cost never passes the threshold, then the proof is complete. Hence we will assume for the rest of the proof that indeed there is an iteration $t$ such that $\mathcal{C}(t) \in [n^{(\tau-2)/(\tau-1)+4\eps}, n^{(\tau-2)/(\tau-1)+5\eps}]$.

Let $t=t_0+t_1$ be this iteration. By symmetry, let us assume without loss of generality that $\mathcal{C}_{s_1}(t_1) \le \mathcal{C}_{s_0}(t_0)$. Clearly, since $\mathcal{C}(t) = \mathcal{C}_{s_0}(t_0) + \mathcal{C}_{s_1}(t_1)$, this implies that $\mathcal{C}_{s_0}(t_0) \ge \mathcal{C}(t)/2 \ge n^{(\tau-2)/(\tau-1)+3\eps}$. By Lemma \ref{lem:balanced-cost}, we know that whp $\mathcal{C}_{s_0}(t_0) \le t^{\eps}\mathcal{C}_{s_1}(t_1)$. Note that each iteration corresponds to one vertex expansion, and each vertex is expanded at most once by the algorithm, so $t\le |\mathcal{V}|=n$.

Hence we have $\mathcal{C}_{s_0}(t_0) \le n^{\eps}\mathcal{C}_{s_1}(t_1)$, or equivalently $\mathcal{C}_{s_1}(t_1) \ge n^{-\eps}\mathcal{C}_{s_0}(t_0) \ge n^{(\tau-2)/(\tau-1)+2\eps}$. Since $\mathcal{C}_{s_i}(t_i) \le \mathcal{C}(t) \le n^{(\tau-2)/(\tau-1)+5\eps}$ for all $i\in\{0,1\}$, and for small $\eps>0$ we have
\[
n^{(\tau-2)/(\tau-1)+2\eps}
\ge  \big(n^{(\tau-2)/(\tau-1)+\eps}\big)^{1+\eps},
\]
the conditions of Lemma \ref{lem:expansion-vertex-approx-algo} are fulfilled with $k \coloneqq n^{\frac{\tau-2}{\tau-1}+\eps}$ for both $i=0$ and $i=1$, and hence whp we have expanded at least $\log k = \Omega(\log n)$ vertices of weight in $[n^{\frac{\tau-2}{\tau-1}+\eps},2n^{\frac{\tau-2}{\tau-1}+\eps}]$ on both sides. By Lemma \ref{lem:meeting-vertex-approx-algo}, this implies that whp $S_{s_0}(t_0) \cap S_{s_1}(t_1) \neq \emptyset$. In other words, the condition ``$\Gamma(v) \cap S_{\overline{s}} \neq \emptyset$" on line 8 of $\textnormal{V-BFS}_{approx}(\mathcal{G},s_0,s_1)$ has been met whp, in which case the algorithm then terminates. We have thus shown that the runtime is whp at most $n^{(\tau-2)/(\tau-1)+5\eps}$, and since $\eps>0$ can be chosen arbitrarily small this concludes the proof.
\end{proof}

\subsubsection{Exact Vertex-Balanced BBFS}

The aim of this subsection is to complete the proof of our second main result, the runtime bound for the exact vertex-balanced algorithm, which is the content of Theorem~\ref{thm:vertex-exact}. Up to the first time when a path between $s_0$ and $s_1$ is discovered, the exact and approximate vertex-balanced algorithms behave precisely in the same way, so we can focus on lines~14-25 of Algorithm \ref{algo:vertex-exact} (since lines~9-13 only increase the runtime by an additive constant). We distinguish two cases, depending on the weight of the vertices in the queues $Q_{s_0}^-(t), Q_{s_1}^-(t)$ once the algorithm enters line~14.
If both of them contain a vertex of weight at least $n^{1/2+o(1)}$, then whp we will find an edge between the queues in time less than $n^{1/2+o(1)}$. This case is treated in Lemma~\ref{lem:meeting-vertex-exact-algo}, building on Lemmata \ref{lem:degree-sum<max-weight} and \ref{lem:expansion-vertex-exact-algo}. Otherwise, at least one of the queues does not have any vertex of weight larger than $n^{1/2+o(1)}$. In that case, we show in Proposition \ref{prop:small-weight-implies-small-layer} that the algorithm's choice of the smallest remaining queue (on line~14) coincides with the queue containing no weight larger than $n^{1/2+o(1)}$, and Lemma \ref{lem:degree-sum<max-weight} then shows that this queue can be emptied in time $n^{1/2+o(1)}$. We start by showing that, in a queue containing vertices of weight at most $n^{1/2+o(1)}$, the sum of the degrees of the vertices in this queue is also upper bounded by $n^{1/2+o(1)}$. This is again a consequence of the heavy-tail distribution of the weights discovered during the BFS.

\begin{lemma}\label{lem:degree-sum<max-weight}
     Let $\mathcal{G}$ be a Chung-Lu graph on $n$ vertices. Consider a set $Q$ of vertices whose weights follow the conditional distribution given by equation \eqref{eq:bfs-new-weight}. Let $\eps>0$ be arbitrarily small and suppose that, for all $v\in Q$ and their corresponding parent vertex $u\in\mathcal{V}$, we have $W_v, W_u \le O(n^{1/2+\eps})$. Then with high probability
    \begin{align*}
        \sum_{v\in Q} \deg(v) \le n^{1/2+O(\eps)}.
    \end{align*}   

\end{lemma}

\begin{proof}
By Lemma \ref{lem:degree-concentration}, we know that whp
\begin{align}\label{eq:sum-degree<=sum-weight}
    \sum_{v\in Q} \deg(v) 
    \le \sum_{v\in Q} O(W_v + \log^2 n) 
    \le O\Big(\sum_{v\in Q} W_v\Big) + O(|Q|\log^2 n)
    \le O\Big(\log^2 n\sum_{v\in Q} W_v\Big),
\end{align}
where the last inequality holds because $|Q| = \sum_{v\in Q} 1 \le \sum_{v\in Q} W_v$. Observe that, since for each $v\in Q$ with parent vertex $u$ it holds that $W_u \le O(n^{1/2+\eps}) = o(n)$, we are never in the last case of equation \eqref{eq:bfs-new-weight}. Hence, for each such $v\in Q$ we have
\begin{align*}
    f_{W_v}(w) = \Theta(w^{-\tau}\cdot\min\{w, \tfrac{n}{W_u}\}).
\end{align*}
Since by assumption $W_u, W_v \le O(n^{1/2+\eps})$, we also have
\[
\frac{n}{W_u} \ge \Omega(n^{1/2-\eps}) > n^{(1/2+\eps)(1-4\eps)+\eps^2} > W_v^{1-4\eps},
\]
and hence the sequence $(W_v)_{v\in Q}$ can be stochastically sandwiched by two sequences of i.i.d.\ random variables following power-laws with respective exponents $\tau-1$ and $\tau-1+4\eps$. Since $\tau<3$ and $\eps$ is small, we have $\tau-2+4\eps<1$. Hence, using the same reasoning that was used to derive the inequality \eqref{eq:sum-weight<=max-weight} in the proof of Lemma \ref{lem:cost-expanded-relation}, we see that whp
\begin{align*}
    \sum_{v\in Q} W_v \le |Q|^{O(\eps)} \cdot \max_{v\in Q} W_v \le n^{1/2+O(\eps)},
\end{align*}
where the second inequality comes from the assumption $\max_{v\in Q} W_v \le O(n^{1/2+\eps})$ and the trivial inequality $|Q|\le |\mathcal{V}| = n$. Plugging the above equation in \eqref{eq:sum-degree<=sum-weight} concludes the proof.
\end{proof}

The next lemma guarantees that, if the maximal weight in the set $Q_{s_i}^- (t)$ on a side $s_i$ exceeds $n^{1/2+\varepsilon}$, then at least order $\log n$ many vertices of weight $\Theta(n^{1/2}$) have been added to the queue beforehand. The key insight needed for the proof is that vertices of weight of order $n^{1/2}$ are polynomially more abundant than vertices of weight of order $n^{1/2+\delta}$, for any $\delta>0$.

\begin{lemma}\label{lem:expansion-vertex-exact-algo}
   Let $\mathcal{G}$ be a Chung-Lu graph on $n$ vertices and let $s_0$ and $s_1$ be two vertices of $\mathcal{G}$ chosen uniformly at random. Consider the algorithm $\textnormal{V-BFS}_{exact}(\mathcal{G},s_0,s_1)$ and suppose that the condition ``$\Gamma(v) \cap S_{\overline{s}} \neq \emptyset$" on line 8 is satisfied for the first time at iteration $t$. Assume that, for some arbitrarily small constant $\eps >0$, we have $\max_{v\in Q_{s_i}^-(t)} W_v \ge n^{1/2+\eps}$ for all $i\in\{0,1\}$. Then, for $i \in \{0,1\}$, the set $Q_{s_i}(t) \cup Q_{s_i}^-(t)$ with high probability contains $\Omega(\log n)$ vertices of weight in $[n^{1/2}, 2n^{1/2}]$. Furthermore, for any $\delta>0$, with high probability the following holds: Over the course of the run of the algorithm, at least $\Omega(\log n)$ such vertices were added to $Q_{s_i}(t) \cup Q_{s_i}^-(t)$ before any vertex of weight $\Omega(n^{1/2+\delta})$ was added.
\end{lemma}

\begin{proof}
By symmetry and a union bound, it suffices to prove the statement for $i=0$. So consider the sequence $(W_v)_{v\in (Q_{s_0}(t) \cup Q_{s_0}^-(t))\setminus\{s_0\}}$, where the order is the one in which the corresponding vertices were added to the queue.
We start by noting that, until the condition ``$\Gamma(v) \cap S_{\overline{s}} \neq \emptyset$" is satisfied, i.e.\ until iteration $t$, the algorithms $\textnormal{V-BFS}_{exact}(\mathcal{G},s_0,s_1)$ and $\textnormal{V-BFS}_{approx}(\mathcal{G},s_0,s_1)$ behave exactly in the same way, and in particular accumulate the same cost. Hence, using Theorem \ref{thm:vertex-approx}, we know that $\mathcal{C}(t) \le n^{(\tau-2)/(\tau-1)+o(1)} \le O(n^{1-\eps})$. In particular, $\mathcal{C}_{s_0}(t) = O(n^{1-\eps})$, so by Lemma \ref{lem:bfs-new-weight} we know the density of the weights $W_v$ for $v\in (Q_{s_0}(t) \cup Q_{s_0}^-(t))\setminus\{s_0\}$ satisfy equation \eqref{eq:bfs-new-weight-upper-bound}, i.e.\ they are stochastically dominated by a sequence of i.i.d.\ random variables following a power-laws with exponent $\tau-1$. In particular, the probability of adding a weight in $[n^{1/2}, 2n^{1/2}]$ is $n^{\Omega(\eps)}$ times more likely than adding a weight of order $\Omega(n^{1/2+\eps})$. The statement now follows using Chernoff bounds.
\end{proof}

The next lemma ensures that the algorithm terminates with cost at most $n^{1/2+O(\delta)}$ once the conditions guaranteed by the previous lemma are met. This will settle the first case in the proof of Theorem~\ref{thm:vertex-exact}.

\begin{lemma}\label{lem:meeting-vertex-exact-algo}
    Let $\mathcal{G}$ be a Chung-Lu graph on $n$ vertices and let $s_0$ and $s_1$ be two vertices of $\mathcal{G}$ chosen uniformly at random, and assume that $s_0$ and $s_1$ are in the giant component of $\mathcal{G}$. Consider the algorithm $\textnormal{V-BFS}_{exact}(\mathcal{G},s_0,s_1)$ and let $t$ be the first iteration for which the condition ``$\Gamma(v) \cap S_{\overline{s}} \neq \emptyset$" on line 8 is fulfilled. Assume further that each of the sets $Q_{s_i}(t) \cup Q_{s_i}^-(t)$, $i\in\{0,1\}$, contains $\Omega(\log n)$ vertices of weight in $[n^{1/2}, 2n^{1/2}]$ that will be or have been expanded before any vertex of weight $\Omega(n^{1/2+\delta})$, for some arbitrarily small $\delta>0$. Then with high probability the algorithm terminates with total cost at most $n^{1/2+O(\delta)}$.  
\end{lemma}
\begin{proof}
By Theorem \ref{thm:vertex-approx}, we know that the condition ``$\Gamma(v) \cap S_{\overline{s}} \neq \emptyset$" on line 8 is met within cost at most $n^{(\tau-2)/(\tau-1)+o(1)}$ whp, so we can safely ignore this initial cost and assume that we are in the inner while-loop of Algorithm \ref{algo:vertex-exact} (note that $\frac{\tau-2}{\tau-1}<\frac{1}{2}$ since $\tau<3$).
The proof then consists of two parts. First, we show that whp the condition ``$\Gamma(v') \cap Q^-_{\overline{p}} \neq \emptyset$" on line 17 has been fulfilled under the given assumptions. Then we show that this generates a total cost bounded above by $n^{1/2+O(\delta)}$.

For the first part, note that the condition ``$\Gamma(v') \cap Q^-_{\overline{p}} \neq \emptyset$" is satisfied in particular when a vertex in $Q_{s_0}(t) \cup Q_{s_0}^-(t)$ is connected to a vertex in $Q_{s_1}(t) \cup Q_{s_1}^-(t)$. After iteration $t$, Algorithm \ref{algo:vertex-exact} enters its inner while-loop and expands all vertices in $Q^-_p(t)$ for some $p\in\{s_0,s_1\}$. By symmetry, we may assume that $p=s_0$. By assumption, there are $\Omega(\log n)$ disjoint pairs $(u_0,u_1)\in (Q_{s_0}(t) \cup Q_{s_0}^-(t)) \times (Q_{s_1}(t) \cup Q_{s_1}^-(t))$ satisfying $W_{u_0}, W_{u_1} \in [n^{1/2}, 2n^{1/2}]$.
Each such vertex pair is connected independently with probability
\[
\Pr(u_0u_1\in\mathcal{E} \mid W_{u_0}, W_{u_1}) =  \Theta\Big(\min\Big\{\frac{W_{u_0}W_{u_1}}{n}, 1\Big\}\Big) = \Theta\Big(\min\Big\{\frac{n^{\frac{1}{2}}\cdot n^{\frac{1}{2}}}{n}, 1\Big\}\Big) \ge c,
\]
where the last inequality holds for some constant $c>0$. Therefore, we can upper-bound the probability that the algorithm does not terminate before expanding these $\Omega(\log n)$ vertices on the $s_0$-side by the probability that none of these vertex pairs is connected, that is, by $(1-c)^{\Omega (\log n)} = o(1)$. 

For the second part, note that by assumption together with the above reasoning, we know that the algorithm has terminated before expanding any vertex of weight $\Omega(n^{1/2+\delta})$. Moreover, by Theorem \ref{thm:vertex-approx}, the cost up to iteration satisfies whp $\mathcal{C}(t) \le n^{(\tau-2)/(\tau-1)+\delta} = o(n^{1/2})$, and in particular $\mathcal{C}_{p}(t) \le n^{(\tau-2)/(\tau-1)+\delta}$. Therefore, by Lemma \ref{lem:bfs-new-weight}, all vertices expanded by the algorithm after iteration $t$ follow the distribution given by equation \eqref{eq:bfs-new-weight}. Hence, since $\delta$ is small, we can use Lemma \ref{lem:degree-sum<max-weight} to bound the additional cost (after iteration $t$) by $n^{1/2+O(\delta)}$, which concludes the proof.
\end{proof}

We turn to the second case, assuming therefore that $\max_{v\in Q_{s_i}^-(t)} W_v \le n^{1/2+o(1)}$ for some $i\in\{0,1\}$. This second scenario implies, as we demonstrate below, that the above bound on the maximal weight also holds for the queue which will be emptied by the algorithm, i.e.\ for side $p \coloneqq \argmin_{s_0,s_1}\{|Q_{s_0}^-(t)|, |Q_{s_1}^-(t)|\}$. To establish this, we use the tight relation that ties the largest weight in the queue to the number of vertices in that queue (see Lemma \ref{lem:max-weight-shifted-distribution}).

\begin{proposition}\label{prop:small-weight-implies-small-layer}
     Let $\mathcal{G}$ be a Chung-Lu graph on $n$ vertices and let $s_0$ and $s_1$ be two vertices of $\mathcal{G}$ chosen uniformly at random. Consider the algorithm $\textnormal{V-BFS}_{exact}(\mathcal{G},s_0,s_1)$ and the first iteration $t$ at which the condition \textquotedblleft$\Gamma(v) \cap S_{\overline{s}} \neq \emptyset$\textquotedblright on line 8 is satisfied. Let $p \coloneqq \argmin_{s_0,s_1}\{|Q_{s_0}^-(t)|, |Q_{s_1}^-(t)|\}$ as set on line 14 of the algorithm. Then, for any $\eps>0$, we have
     \begin{equation*}
         \mathbb{P}\Big(\big\{\exists i\in \{0,1\}: \max_{v\in Q_{s_i}^-(t)} W_v \le n^{1/2+\eps}\big\} \cap \big\{\max_{v\in Q_{p}^-(t)} W_v > n^{1/2+2\eps}\big\}\Big)=o(1).
     \end{equation*}
\end{proposition}

\begin{proof}
Let us start by noticing that if $\max_{v\in Q_{s_i}^-(t)} W_v \le n^{1/2+2\eps}$ for all $i\in\{0,1\}$, then we are trivially done since $p$ is set either to $s_0$ or to $s_1$. Hence we focus on showing that 
\begin{equation}\label{newgoal}
\mathbb{P}\Big(\big\{\max_{v\in Q_{p}^-(t)} W_v > n^{1/2+\varepsilon}\big\} \cap \mathcal{F}\Big)=o(1),
\end{equation}
where we set
\begin{align*}
    \mathcal{F} \coloneqq \big\{\max_{v\in  Q_{s_0}^-(t)} W_v \le n^{1/2+\eps}  \text{ and }  \max_{v\in  Q_{s_1}^-(t)} W_v > n^{1/2+2\eps}\big\}.
\end{align*}
Note that, by symmetry and a union bound, this can be done without loss of generality (we can just swap the roles of $s_0$ and $s_1$). 

First of all, note that the algorithms $\textnormal{V-BFS}_{approx}(\mathcal{G},s_0,s_1)$ and $\textnormal{V-BFS}_{exact}(\mathcal{G},s_0,s_1)$ behave exactly in the same way until the condition \textquotedblleft$\Gamma(v) \cap S_{\overline{s}} \neq \emptyset$\textquotedblright on line 8 is satisfied for the first time. By Theorem \ref{thm:vertex-approx}, this induces a cost of at most $n^{\frac{\tau-2}{\tau-1}+\eps}$ whp, so the conditions of Lemma \ref{lem:bfs-new-weight} are satisfied, which implies by Lemma \ref{lem:max-weight-shifted-distribution} that \eqref{eq:max-weight-shifted-upper-bound} holds whp (as their cardinality goes to infinity) for both $Q_{s_0}^-(t)$ and $Q_{s_1}^-(t)$. Combined with the fact that $\max_{v\in Q_{s_0}^-(t)}W_v \le n^{1/2+\eps}$ on the event $\mathcal{F}$, this also implies that the additional assumption of Lemma \ref{lem:max-weight-shifted-distribution} are satisfied by $Q_{s_0}^-(t)$.

Secondly, note that on the event $\mathcal{F}$ we have $|Q_{s_1}^-(t)|=\omega(1)$ whp. Indeed, using (the conditional version of) Markov's inequality and equation \eqref{eq:bfs-new-weight-upper-bound} from Lemma \ref{lem:bfs-new-weight}, we readily see that
\begin{align*}
   \mathbb{P}( \{|Q_{s_1}^-(t)|=O(1)\} \cap \mathcal{F})
   &\le \mathbb{E}\Big[\mathbbm{1}_{\{|Q_{s_1}^-(t)|=O(1)\}}\mathbb{P}( \max_{v\in Q_{s_1}^-(t)}W_v > n^{1/2+2\varepsilon} \mid |Q_{s_1}^-(t)|)\Big]\\
   &\leq O\Big(\int_{n^{1/2+2\varepsilon}}^{\infty}w^{-\tau+1}dw\Big)\mathbb{E}\Big[\mathbbm{1}_{\{|Q_{s_1}^-(t)|=O(1)\}}\cdot|Q_{s_1}^-(t)|\Big] = o(1).
\end{align*}
Hence we can bound the probability in \eqref{newgoal} by
\begin{align*}
    \mathbb{P}\Big(\big\{\max_{v\in Q_{p}^-(t)} W_v > n^{1/2+\varepsilon}\big\} \cap \mathcal{F}\Big) \le
    \mathbb{P}\Big(\big\{\max_{v\in Q_{p}^-(t)} W_v > n^{1/2+\varepsilon}\big\} \cap \{|Q_{s_1}^-(t)|=\omega(1)\} \cap \mathcal{F}\Big) + o(1).
\end{align*}
Observe that, if $|Q_{s_0}^-(t)| < |Q_{s_1}^-(t)|$, then $p=s_0$ on line 14 of the $\textnormal{V-BFS}_{exact}(\mathcal{G},s_0,s_1)$ algorithm and hence on the event $\mathcal{F}$ we have
\[\max_{v\in Q_{p}^-(t)} W_v =\max_{v\in Q_{s_0}^-(t)} W_v \leq n^{1/2+\eps}.\]
Therefore we obtain
\begin{equation}\label{eq:smaller-layer}
    \mathbb{P}\Big(\big\{\max_{v\in Q_{p}^-(t)} W_v > n^{1/2+\varepsilon}\big\} \cap \mathcal{F}\Big) \le 
    \mathbb{P}\big(\{|Q_{s_0}^-(t)| \geq |Q_{s_1}^-(t)|\} \cap \{|Q_{s_1}^-(t)|=\omega(1)\} \cap \mathcal{F}\big)+o(1).
\end{equation}
Define the events
\begin{align*}
     \mathcal{E}_0&\coloneqq \Big\{\max_{v\in Q_{s_0}^-(t)}W_v \ge |Q_{s_0}^-(t)|^{1/(\tau-2)}/\log|Q_{s_0}^-(t)|\Big\}, \\
     \mathcal{E}_1&\coloneqq \Big\{\max_{v\in Q_{s_1}^-(t)}W_v \le |Q_{s_1}^-(t)|^{1/(\tau-2)}\log|Q_{s_1}^-(t)|\Big\}.
\end{align*}
By equation \eqref{eq:max-weight-shifted-upper-bound} in Lemma \ref{lem:max-weight-shifted-distribution}, we have
\begin{align*}
    \pr(\mathcal{E}_1^c \cap \{|Q_{s_1}^-(t)|=\omega(1)\})=o(1),
\end{align*}
and hence by \eqref{eq:smaller-layer} we have
\begin{align}
\begin{split}\label{eq:control-s1-max-weight}
    &\mathbb{P}\Big(\big\{\max_{v\in Q_{p}^-(t)} W_v > n^{1/2+\varepsilon}\big\} \cap \mathcal{F}\Big) \\
    &\qquad\le \mathbb{P}\big(\{|Q_{s_0}^-(t)| \geq |Q_{s_1}^-(t)|\} \cap \{|Q_{s_1}^-(t)|=\omega(1)\} \cap \mathcal{F} \cap \mathcal{E}_1\big)+o(1).
\end{split}
\end{align}
Recall that $Q_{s_0}^-(t)$ also satisfies the additional assumptions of Lemma \ref{lem:max-weight-shifted-distribution}. Since $Q_{s_0}^-(t) = \omega(1)$ on the event $\{|Q_{s_0}^-(t)| \geq |Q_{s_1}^-(t)|\} \cap \{|Q_{s_1}^-(t)|=\omega(1)\}$, we obtain, using equation \eqref{eq:max-weight-shifted-lower-bound} in Lemma \ref{lem:max-weight-shifted-distribution},
\begin{align*}
    \mathbb{P}\big(\{|Q_{s_0}^-(t)| \geq |Q_{s_1}^-(t)|\} \cap \{|Q_{s_1}^-(t)|=\omega(1)\} \cap \mathcal{F} \cap \mathcal{E}_0^c\big) = o(1),
\end{align*}
and hence going back to \eqref{eq:control-s1-max-weight} we arrive at
\begin{align}\label{eq:control-s0-max-weight}
    \mathbb{P}\Big(\big\{\max_{v\in Q_{p}^-(t)} W_v > n^{1/2+\varepsilon}\big\} \cap \mathcal{F}\Big)
    \le \mathbb{P}\big(\{|Q_{s_0}^-(t)| \geq |Q_{s_1}^-(t)|\} \cap \mathcal{F} \cap \mathcal{E}_0 \cap \mathcal{E}_1\big)+o(1).
\end{align}
On the event $\mathcal{F} \cap \mathcal{E}_0$, we have
\begin{align*}
    n^{1/2+\eps} \ge \max_{v\in Q_{s_0}^-(t)}W_v \ge |Q_{s_0}^-(t)|^{1/(\tau-2)}/\log|Q_{s_0}^-(t)|.
\end{align*}
Hence, on the event $\{|Q_{s_0}^-(t)| \geq |Q_{s_1}^-(t)|\} \cap \mathcal{F} \cap \mathcal{E}_0$, we have
\begin{align*}
    n^{1/2+\eps} \ge |Q_{s_1}^-(t)|^{1/(\tau-2)}/\log|Q_{s_0}^-(t)|.
\end{align*}
So we conclude that, on the event $\{|Q_{s_0}^-(t)| \geq |Q_{s_1}^-(t)|\} \cap \mathcal{F} \cap \mathcal{E}_0 \cap \mathcal{E}_1$, 
\begin{align*}
    n^{1/2+\eps} \ge \frac{\max_{v\in Q_{s_1}^-(t)}W_v}{\log|Q_{s_0}^-(t)|\log|Q_{s_1}^-(t)|}
    \ge \frac{\max_{v\in Q_{s_1}^-(t)}W_v}{\log^2|Q_{s_0}^-(t)|}
    \ge \frac{n^{1/2+2\eps}}{\log^2|Q_{s_0}^-(t)|} ,
\end{align*}
which implies that $\log^2|Q_{s_0}^-(t)| \ge n^{\eps}$, which is impossible because of the trivial upper bound $|Q_{s_0}^-(t)| \le |\mathcal{V}| = n$. Therefore, we deduce that the probability in \eqref{eq:control-s0-max-weight} satisfies
\begin{align*}
     \mathbb{P}\big(\{|Q_{s_0}^-(t)| \geq |Q_{s_1}^-(t)|\} \cap \mathcal{F} \cap \mathcal{E}_0 \cap \mathcal{E}_1\big) = 0,
\end{align*}
which concludes the proof.
\end{proof}

Collecting the previous results allows us to prove the main theorem of this subsection, which is our second main result.

\begin{theorem}\label{thm:vertex-exact}
    Let $\mathcal{G}$ be a Chung-Lu graph on $n$ vertices and let $s_0$ and $s_1$ be two vertices of $\mathcal{G}$ chosen uniformly at random. Assuming that $s_0$ and $s_1$ are in the giant component of $\mathcal{G}$, $\textnormal{V-BFS}_{exact}(\mathcal{G},s_0,s_1)$ is a correct algorithm and with high probability it terminates in time at most $n^{1/2+o(1)}$.
\end{theorem}

\begin{proof}
Both $s_0$ and $s_1$ lie in the giant component by assumption. Hence they are connected by at least one path and it will eventually be found by the algorithm. By Lemma~\ref{lem:vertex-exact-correct}, any such path will have length exactly $d(s_0,s_1)$. Hence the remainder of the proof consists of showing that such a path will be found in time at most $n^{1/2+o(1)}$. Note that until some path is found, the algorithms $\textnormal{V-BFS}_{exact}(\mathcal{G},s_0,s_1)$ and $\textnormal{V-BFS}_{approx}(\mathcal{G},s_0,s_1)$ behave exactly the same and in particular accumulate the same runtime. By Theorem~\ref{thm:vertex-approx}, which bounds the total runtime of $\textnormal{V-BFS}_{approx}(\mathcal{G},s_0,s_1)$ by $n^{(\tau-2)/(\tau-1)+o(1)}$, we may therefore assume that a (not necessarily shortest) path has been found and that up to and including line 14 of the algorithm, also $\textnormal{V-BFS}_{exact}(\mathcal{G},s_0,s_1)$ has accumulated cost at most $n^{(\tau-2)/(\tau-1)+o(1)}=o(n^{1/2})$. We may therefore focus our analysis on the runtime of lines 14-25 of its pseudocode. We proceed by making a case distinction based on the vertex weights that are contained in the queues when the algorithm enters line 14. If for some $\eps >0$ we have $\max_{v\in Q_{s_i}^-(t)} W_v \ge n^{1/2+\eps}$ for all $i\in\{0,1\}$, then Lemma~\ref{lem:expansion-vertex-exact-algo} becomes applicable, guaranteeing that for any $\delta>0$ and for $i \in \{0,1\}$, at least $\Omega(\log n)$ vertices of weight in $[n^{1/2}, 2n^{1/2}]$ are expanded over the course of the algorithm before any vertex of weight $\Omega(n^{1/2+\delta})$ is expanded. Therefore, the requirements of Lemma~\ref{lem:meeting-vertex-exact-algo} are met and we can conclude that the algorithm terminates with total cost at most $n^{1/2 + O(\delta)}$. Since $\delta$ can be chosen arbitrarily small, this settles the first case. In the other case, Proposition~\ref{prop:small-weight-implies-small-layer} can be applied. Hence whp the layer $Q^-_p(t)$ that the algorithm empties satisfies $\max_{v\in Q^-_p(t)} W_v \le n^{1/2+O(\eps)}$. Moreover, by Lemma \ref{lem:bfs-new-weight} all the vertices in this layer have a weight that follows the distribution given by equation \eqref{eq:bfs-new-weight}. Hence we can use Lemma \ref{lem:degree-sum<max-weight} to conclude that $\sum_{v\in Q^-_p(t)} \deg(v) \le n^{1/2+O(\eps)}$ whp. This sum of degrees is an upper bound for the runtime of lines 14-25, and since $\eps>0$ was arbitrary this concludes the proof.
\end{proof}

\subsubsection{Edge-Balanced BBFS}
The main goal of this subsection is to prove the runtime bound of $n^{1/2 + o(1)}$ for the edge-balanced algorithm $\textnormal{E-BFS}_{approx}(\mathcal{G},s_0,s_1)$, stated in Theorem~\ref{thm:edge-approx}. We begin with a proposition that says it takes at most $n^{1/2+o(1)}$ rounds to reach vertices of weight close to $n^{1/2}$ (or larger) on both sides.

\begin{proposition}
\label{prop:expansion-edge-approx-algo}
    Let $\mathcal{G}$ be a Chung-Lu graph on $n$ vertices and let $s_0$ and $s_1$ be two vertices of $\mathcal{G}$ that an adversary can choose by looking at the weight sequence $(W_v)_{v\in \mathcal{V}}$, and consider the $\textnormal{E-BFS}_{approx}(\mathcal{G},s_0,s_1)$ algorithm. Assuming that $s_0$ and $s_1$ are in the giant component of $\mathcal{G}$, then for arbitrarily small $\eps>0$ with high probability there are rounds $k_0, k_1 \le n^{1/2+o(1)}$ such that $W_{v_{s_0}(k_0)}, W_{v_{s_1}(k_1)} \ge n^{1/2-\eps}$.
    
\end{proposition}

\begin{proof}
Let $\eps>0$ be arbitrarily small. Note that after $k$ rounds we have expanded $\frac{k}{2}$ edges on each side if $k$ is even, or $\frac{k+1}{2}$ edges on the $s_0$-side and $\frac{k-1}{2}$ edges on the $s_1$-side if $k$ is odd.
Therefore, by a union bound it is enough to focus on one side, say the $s_0$-side, and show that whp the considered unidirectional BFS satisfies the desired property: we start expanding a vertex of weight at least $n^{1/2-\eps}$ after $k_0$ rounds (or expanded edges), where $k_0 \le n^{1/2+o(1)}$. Note that when we focus on the unidirectional BFS of the $s_0$-side, then the edge-balanced algorithm $\textnormal{E-BFS}_{approx}(\mathcal{G},s_0,s_1)$ and the vertex-balanced algorithm $\textnormal{V-BFS}_{approx}(\mathcal{G},s_0,s_1)$ are equivalent and the round number after expanding $t_0$ vertices corresponds to the cost $\mathcal{C}_{s_0}(t_0)$.

The adversary can choose $s_0$ after looking at the weight sequence $(W_v)_{v\in \mathcal{V}}$. If they choose $s_0$ with $W_{s_0} \ge n^{1/2-\eps}$ then we are immediately done (namely we have $k_0=1$), so without loss of generality we can assume that $W_{s_0} \le n^{1/2-\eps}$. Suppose that $W_{s_0} \le n^{(\tau-2)/(\tau-1)}$. Then by Lemma \ref{lem:expansion-vertex-approx-algo} with the choice $k_{\ref{lem:expansion-vertex-approx-algo}} \coloneqq n^{(\tau-2)/(\tau-1)}$ (note that the number of expanded edges on the $s_0$-side will reach $k_{\ref{lem:expansion-vertex-approx-algo}}$ eventually since we assume that $s_0$ is in the giant component), we know that whp after at most $(n^{(\tau-2)/(\tau-1)})^{1+o(1)} = o(n^{1/2})$ rounds we are expanding a vertex of weight at least $n^{(\tau-2)/(\tau-1)}$. Therefore we can also assume that the BFS is at a vertex $v_0$ with $W_{v_0} \in [n^{(\tau-2)/(\tau-1)}, n^{1/2-\eps}]$, and moreover that $v_0$ is the vertex of largest weight that has been expanded so far.

For the rest of the proof, we only focus on rounds before $n^{1/2+o(1)}$, i.e.\ where the cost of the algorithm is bounded by $n^{1/2+o(1)}$. Hence, the condition $w \cdot C(t) = O(n^{1-\eps})$ of Lemma \ref{lem:bfs-new-weight} is satisfied for all $w=O(n^{1/2-\eps})$. Therefore, since $v_0$ is the vertex of largest weight expanded so far, for every vertex $v\in\Gamma(v_0)\setminus\{s_0\}$ with parent vertex $u$ of weight $W_u \le W_{v_0}$ and all $w=O(n^{1/2-\eps})$ we have $f_{W_v}(w \mid W_{u}, \textnormal{BFS}^v(\mathcal{G},s))) = \Theta(w^{-\tau}\cdot\min\{w, \tfrac{n}{W_{u}}\})$ by Lemma \ref{lem:bfs-new-weight}. In particular, if we choose $v\in\Gamma(v_0)\setminus\{s_0\}$ uniformly at random, we have
\[
\pr(W_v \ge n^{1/2-\eps})
\ge \pr(W_v \in [n^{1/2-\eps}, 2n^{1/2-\eps}])
= \int_{n^{1/2-\eps}}^{2n^{1/2-\eps}} \Theta(w^{-\tau}\cdot\min\{w, \tfrac{n}{W_{u}}\}) dw.
\]
Since $W_u \le W_{v_0} \le n^{1/2-\eps}$ and the integration domain is $[n^{1/2-\eps}, 2n^{1/2-\eps}]$, we have $w <\tfrac{n}{W_{u}}$ in the integral above, and hence
\begin{align}\label{eq:prob-weight-sqrt-n}
    \pr(W_v \ge n^{1/2-\eps})
    \ge \int_{n^{1/2-\eps}}^{2n^{1/2-\eps}} \Omega(w^{1-\tau})dw
    = \Omega(n^{-(\tau-2)/2}).
\end{align}
Since $W_{v_0} \ge n^{(\tau-2)/(\tau-1)}$, by Lemma \ref{lem:degree-concentration} we know that $|\Gamma(v_0)\setminus\{s_0\}| = \Omega(n^{(\tau-2)/(\tau-1)})$ whp. These $\Omega(n^{(\tau-2)/(\tau-1)})$ vertices are discovered in a random order by the BFS started at $s_0$, so by \eqref{eq:prob-weight-sqrt-n}, since $2<\tau<3$, whp among the first $n^{(\tau-2)/(\tau-1)} = o(n^{1/2})$ vertices of $\Gamma(v_0)\setminus\{s_0\}$ discovered by that BFS there is at least on vertex of weight in $[n^{1/2-\eps},\infty)$. Consequently, if we denote by $k_0'$ the first round in which a vertex of weight in $[n^{1/2-\eps},\infty)$ is discovered by the BFS, then we know that $k_0'=o(n^{1/2})$ whp.

The number of rounds $k_0$ it takes until the algorithm is actually expanding this vertex of weight in $[n^{1/2-\eps},\infty)$ is then the sum of the degrees of the previously expanded vertices. The weight of $s_0$ is at most $n^{1/2-\eps}$, and hence its degree is $O(n^{1/2-\eps})$ whp by Lemma \ref{lem:degree-concentration}, so we can safely ignore this vertex. As we have seen, all the assumptions of Lemma \ref{lem:bfs-new-weight} are fulfilled, and hence the distribution of the other vertices expanded by the BFS up to round $k_0$ is given by equation \eqref{eq:bfs-new-weight}. Moreover, by definition of $k_0$ all the weights involved are smaller than $n^{1/2-\eps}$, and therefore we can apply Lemma \ref{lem:degree-sum<max-weight} to conclude that the sum of the degrees of these vertices is at most $n^{1/2+o(1)}$ whp, which concludes the proof.

\end{proof}

The second and final auxiliary fact needed to prove the main result of this subsection is Lemma \ref{lem:meeting-edge-approx-algo} below. This lemma guarantees that from the moment at which the algorithm has started to expand a vertex of weight close to $n^{1/2}$ or larger on each side (and potentially finished to do so), it will terminate in at most $n^{g(\tau)+o(1)}$ rounds. In this runtime bound, $g(\tau)$ is an explicit function which depends only on the power-law parameter $\tau$. Furthermore, as shown in Remark \ref{rem:g(tau)<1/2}, this function satisfies $g(\tau)<1/2$ for all $\tau\in (2,3)$, hence this second phase until termination takes less than $n^{1/2}$ rounds. 

\begin{lemma}\label{lem:meeting-edge-approx-algo}
    Let $\mathcal{G}$ be a Chung-Lu graph on $n$ vertices and let $s_0$ and $s_1$ be two vertices of $\mathcal{G}$ that an adversary can choose by looking at the weight sequence $(W_v)_{v\in \mathcal{V}}$. Let $\eps>0$ be arbitrarily small, and suppose that at some point the algorithm $\textnormal{E-BFS}_{approx}(\mathcal{G},s_0,s_1)$ has started to expand two vertices $v_0$ and $v_1$ (on the $s_0$-side resp.\ the $s_1$-side) of weights $W_{v_0},W_{v_1} \ge n^{1/2-\eps}$ (and potentially finished to expand one or both of them). Then with high probability the algorithm terminates in the next $n^{\frac{\tau^2-4\tau+5}{2(\tau-1)}+o(1)} \eqqcolon n^{g(\tau)+o(1)}$ rounds.
\end{lemma}

\begin{remark}\label{rem:g(tau)<1/2}
    The exponent $g(\tau)$ is strictly smaller than $1/2$ on the range $\tau\in(2,3)$. Indeed, the function $g(\tau) = \frac{\tau^2-4\tau+5}{2(\tau-1)}$ is strictly convex on the interval $[2,3]$ (since its second derivative is $2/(\tau-1)^3 > 0$), and $g(2) = g(3) = 1/2$.
\end{remark}

\begin{proof}
The idea is that, since $v_0$ and $v_1$ have many common neighbors, at least one such neighbor will be discovered after $n^{g(\tau)}\log n$ edges have been expanded on each side. For that we start by lower bounding the number $|\Gamma(v_0) \cap \Gamma(v_1)|$ of neighbors that $v_0$ and $v_1$ share. Let $\overline{w} \coloneqq \tfrac{n}{\max\{W_{v_0}, W_{v_1}\}}$, and note that $\overline{w}\le n^{1/2+\eps}$ since $W_{v_0},W_{v_1} \ge n^{1/2-\eps}$, and whp $\overline{w} \ge n^{(\tau-2)/(\tau-1)}/\log n$ by Lemma \ref{lem:maxweight}, which we will assume for the rest of the proof. Conditioned on $W_{v_0}, W_{v_1}$, for all vertices $u\in\mathcal{V}$ the events $\{W_u \ge \overline{w}\}\cap\{u \in \Gamma(v_0) \cap \Gamma(v_1)\}$ are independent. Moreover,
\begin{align*}
\pr(W_u \ge \overline{w}, u \in \Gamma(v_0) \cap \Gamma(v_1))
&= \pr(W_u \ge \overline{w})\pr(u \in \Gamma(v_0) \cap \Gamma(v_1) \mid W_u \ge \overline{w}) \\
&\ge \Theta\Big(\overline{w}^{-(\tau-1)} \min\Big\{\frac{\overline{w}W_{v_0}}{n}, 1\Big\}\cdot\min\Big\{\frac{\overline{w}W_{v_1}}{n}, 1\Big\}\big) \\
&= \Theta\Big(\frac{\overline{w}^{3-\tau}W_{v_0}W_{v_1}}{n^2}\Big),
\end{align*}
where the last equality comes from the definition of $\overline{w} = \tfrac{n}{\max\{W_{v_0}, W_{v_1}\}}$. Therefore,
\begin{align*}
    \E[|\Gamma(v_0) \cap \Gamma(v_1)|]
    \ge \E[|\{u\in\Gamma(v_0) \cap \Gamma(v_1) \mid W_u \ge \overline{w}\}|]
    \ge \Theta\Big(\frac{\overline{w}^{3-\tau}W_{v_0}W_{v_1}}{n}\Big).
\end{align*}
In particular, note that since $W_{v_0}, W_{v_1} \ge n^{1/2-\eps}$, we have $\E[|\Gamma(v_0) \cap \Gamma(v_1)|] \ge \Omega(\overline{w}^{3-\tau}/n^{2\eps})$, which is $\omega(1)$ since $\overline{w} \ge n^{(\tau-2)/(\tau-1)}/\log n$ and $\eps$ is small. Thus we can use Chernoff bounds to deduce that whp
\begin{align}\label{eq:common-neighborhood-lower-bound}
    |\Gamma(v_0) \cap \Gamma(v_1)| \ge \frac{\overline{w}^{3-\tau}W_{v_0}W_{v_1}}{n\log n}.
\end{align}
For the rest of the proof, we condition on event \eqref{eq:common-neighborhood-lower-bound} to hold.

For $i\in\{0,1\}$ and $u\in\Gamma(v_i)$, let us denote by $F_i(u)$ the event that $u$ is among the first $n^{g(\tau)}\log n$ neighbors of $v_i$ that are explored by the algorithm $\textnormal{E-BFS}_{approx}(\mathcal{G},s_0,s_1)$. We now need to show that whp there is a vertex $u\in\Gamma(v_0) \cap \Gamma(v_1)$ for which $F_0(u) \cap F_1(u)$ occurs. We define the random variable
\[
X \coloneqq \sum_{u\in\Gamma(v_0) \cap \Gamma(v_1)} X_u,
\]
where $X_u$ is the indicator random variable that $F_0(u) \cap F_1(u)$ occurs. By Lemma \ref{lem:degree-concentration} we know that whp $\deg(v_i) = \Theta(W_{v_i})$ for $i\in\{0,1\}$, so let us condition on that event for the rest of the proof. The neighbors of $v_i$ are explored in a uniformly random order by $\textnormal{E-BFS}_{approx}(\mathcal{G},s_0,s_1)$, hence for a given $u\in\Gamma(v_i)$ we have $\pr(F_i(u)) = \frac{n^{g(\tau)}\log n}{\deg(v_i)} = \Theta(\frac{n^{g(\tau)}\log n}{W_{v_i}})$. Moreover, the order in which the neighbors of $v_0$ are explored is independent of the order in which the neighbors of $v_0$ are explored, hence for every $u\in\Gamma(v_0) \cap \Gamma(v_1)$ we have 
\[
\E[X_u] = \pr(F_0(u) \cap F_1(u))
= \pr(F_0(u))\pr(F_1(u))
= \Theta\Big(\frac{n^{2g(\tau)}\log^2 n}{W_{v_0}W_{v_1}}\Big).
\]
Using equation \eqref{eq:common-neighborhood-lower-bound}, this implies
\begin{align*}
    \E[X] 
    \ge \frac{\overline{w}^{3-\tau}W_{v_0}W_{v_1}}{n\log n} \cdot \Theta\Big(\frac{n^{2g(\tau)}\log^2 n}{W_{v_0}W_{v_1}}\Big) 
    = \Theta\Big(\frac{\overline{w}^{3-\tau}\log n}{n^{1-2g(\tau)}}\Big).
\end{align*}
Using $\overline{w} \ge n^{(\tau-2)/(\tau-1)}/\log n$ we get
\begin{align}\label{eq:lower-bound-early-common-neighbors}
    \E[X] 
    \ge \Omega\big(n^{(3-\tau)(\tau-2)/(\tau-1)+2g(\tau)-1} (\log n)^{\tau-2}\big)
    =\Omega((\log n)^{\tau-2}).
\end{align}

Note that $(X_u)_{u\in\Gamma(v_0) \cap \Gamma(v_1)}$ is a collection of Bernoulli random variables. We claim that these random variables are pairwise negatively correlated, i.e.\ that $\E[X_uX_v] \le \E[X_u]\E[X_v]$ for all distinct $u,v \in\Gamma(v_0) \cap \Gamma(v_1)$. Since these are indicator random variables, this statement is equivalent to $\pr(X_u= 1 \mid X_v = 1) \le \pr(X_u=1)$. Indeed, knowing that $X_v = 1$ just tells us that $v$ is among the first $n^{g(\tau)}\log n$ neighbors of $v_i$ to be explored by that side of the BFS, $i=0,1$. This only makes it less likely that $u$ is also among the first $n^{g(\tau)}\log n$ neighbors of $v_i$ to be explored by that side of the BFS. Since $X$ is a sum of negatively correlated Bernoulli random variables, we know by Lemma~\ref{lem:chebyshev-negatively-correlated} that $\E[X] = \omega(1) \Rightarrow \pr(X=0)=o(1)$. Hence, equation \eqref{eq:lower-bound-early-common-neighbors} allows us to conclude the proof.

\end{proof}

We are now in a position to prove the main theorem of this subsection.

\begin{theorem}
\label{thm:edge-approx}
    Let $\mathcal{G}$ be a Chung-Lu graph on $n$ vertices and let $s_0$ and $s_1$ be two vertices of $\mathcal{G}$ that an adversary can choose by looking at the weight sequence $(W_v)_{v\in \mathcal{G}}$. Assuming that $s_0$ and $s_1$ are in the giant component of $\mathcal{G}$, $\textnormal{E-BFS}_{approx}(\mathcal{G},s_0,s_1)$ is a correct algorithm and with high probability it terminates in time at most $n^{1/2+o(1)}$.
\end{theorem}
\begin{proof}
The correctness follows directly from Lemma \ref{lem:edge-approx-correct}, so we only need to prove the runtime bound. Let $\eps>0$ be arbitrarily small. By Proposition \ref{prop:expansion-edge-approx-algo}, whp after at most $n^{1/2+o(1)}$ rounds we have started to expand a vertex of weight at least $n^{1/2-\eps}$ on each side. From this point on, whp it takes at most $n^{g(\tau)+o(1)}$ additional rounds for the algorithm to terminate by Lemma \ref{lem:meeting-edge-approx-algo}, where $g(\tau) \coloneqq \frac{\tau^2-4\tau+5}{2(\tau-1)}$. Finally by Remark \ref{rem:g(tau)<1/2} we know that $g(\tau)<1/2$, hence the total runtime is indeed at most $n^{1/2+o(1)}$.
\end{proof}

\subsubsection{Vertices outside the giant component}

Finally, we analyze the runtimes of our algorithms in  the various cases where at least one of the starting vertices lies outside the giant component. We will show that in each such situation, our algorithms will terminate in time $n^{o(1)}$ whp.

We first treat the case where both starting vertices lie outside the giant component. The runtime bound in this case is a straightforward consequence of the at most polylogarithmic sizes of all non-giant components.

\begin{lemma}\label{lem:runtime-outside-giant}
Let $\mathcal{G}$ be a Chung-Lu graph or a GIRG on $n$ vertices and $s_0, s_1$ two vertices, both of them outside the (unique) giant component of $\mathcal{G}$. Then any run of $\textnormal{V-BFS}_{approx}(\mathcal{G},s_0,s_1)$, $\textnormal{V-BFS}_{exact}(\mathcal{G},s_0,s_1)$, or $\textnormal{E-BFS}_{approx}(\mathcal{G},s_0,s_1)$ terminates in time $n^{o(1)}$ with high probability.
\end{lemma}
\begin{proof}
We may pessimistically assume that $s_0$ and $s_1$ are in different components and that the searches started from both points exhaust their respective component. Since all non-giant components contain at most polylogarithmically many vertices whp by Theorem \ref{thm:component-sizes}, and hence at most polylogarithmically many edges, whp this takes time $O(\textnormal{polylog(n)}) = n^{o(1)}$. 
\end{proof}

When one starting vertex is located within the giant component, bounding component sizes does not suffice anymore, since a priori our algorithms could explore a large chunk of the giant component before exhausting the non-giant component in which the other starting vertex lies. As the next two lemmata show, the first for the vertex-balanced algorithms, the second for the edge-balanced algorithm, their respective balancing mechanisms prevent precisely that, leading to upper runtime bounds of $n^{o(1)}$.

\begin{lemma}\label{lem:runtime-outside-giant-vertex}
Let $\mathcal{G}$ be a Chung-Lu graph on $n$ vertices and $s_0, s_1$ two vertices of $\mathcal{G}$ chosen uniformly at random. If exactly one of them lies outside the (unique) giant component of $\mathcal{G}$, then the algorithms $\textnormal{V-BFS}_{approx}(\mathcal{G},s_0,s_1)$ and $\textnormal{V-BFS}_{exact}(\mathcal{G},s_0,s_1)$ terminate in time $n^{o(1)}$ with high probability.
\end{lemma}

\begin{proof}
By Theorem \ref{thm:component-sizes} we can assume without loss of generality that there is a unique giant component, and that all other components have at most polylogarithmic size. Moreover, by symmetry and a union bound, we can also assume that $s_0$ lies in the giant component and $s_1$ lies in another component. Note that the two algorithms will terminate (latest) once the smaller component is exhausted (line 3 in Algorithms \ref{algo:vertex-approx} and \ref{algo:vertex-exact}). What remains to show is therefore, in intuitive terms, that the algorithms remain ``balanced", i.e.\ that we do not expand a polynomial number of edges in the giant component before fully exhausting the smaller component. Note that $\textnormal{V-BFS}_{approx}(\mathcal{G},s_0,s_1)$ and $\textnormal{V-BFS}_{exact}(\mathcal{G},s_0,s_1)$ behave differently only from the first iteration in which a path has been found. This does not occur by assumption, since $s_0$ and $s_1$ lie in different components. Hence we can jointly analyze their runtime in this case, assuming without loss of generality that we are considering $\textnormal{V-BFS}_{approx}(\mathcal{G},s_0,s_1)$.

Let $\eps>0$ and consider a run of this algorithm, and the smallest iteration $\Tilde{t}$ such that $\mathcal{C}(\Tilde{t}) \ge 2(\log n)^{2(1+\varepsilon/(\tau-2-\varepsilon))/(3-\tau)}$, pessimistically assuming that by then the algorithm has not terminated yet. Then Lemma~\ref{lem:balanced-cost} applies and continues to apply as long as the total cost is bounded from above by $n^{(\tau-2)/(\tau-1)}$. 
Now observe that it always holds that $t\le n$ and, by definition of $\mathcal{C}_{s_1}(t_1)$ and by the handshake lemma, $\mathcal{C}_{s_1}(t_1)$ can never exceed twice the number of edges in the non-giant component that contains $s_1$, which is at most polylogarithmic. In particular, when $\mathcal{C}_{s_1}(t_1)$ reaches this number, we know that the algorithm has terminated. Hence, using Lemma \ref{lem:balanced-cost}, we can bound the total cost for $t \ge \Tilde{t}$ by 
\begin{align*}
    \mathcal{C}(t) =\mathcal{C}_{s_0}(t_0)+\mathcal{C}_{s_1}(t_1)\le (t^\varepsilon+1) \mathcal{C}_{s_1}(t)= (t^{\varepsilon}+1)\cdot O(\textnormal{polylog(n)}) \le n^{2\varepsilon}.
\end{align*}
Since $\varepsilon>0$ was arbitrary, this concludes the proof.
\end{proof}

\begin{lemma}\label{lem:runtime-outside-giant-edge}
Let $\mathcal{G}$ be a Chung-Lu graph or a GIRG on $n$ vertices and $s_0, s_1$ two vertices, at least one of them outside the (unique) giant component of $\mathcal{G}$. Then $\textnormal{E-BFS}_{approx}(\mathcal{G},s_0,s_1)$ terminates in time $n^{o(1)}$ with high probability.
\end{lemma}

\begin{proof}
By Theorem \ref{thm:component-sizes} we can assume without loss of generality that there is a unique giant component, and that all other components have at most polylogarithmic size, and in particular contain at most $O(\textnormal{polylog(n)})$ edges. Moreover, by symmetry and a union bound, we can also assume that $s_0$ lies in the giant component and $s_1$ lies in another component. Note that the algorithm will terminate (latest) once the smaller component is exhausted (line 4 in Algorithm \ref{algo:edge-approx}). In each iteration of the while-loop of the algorithm, exactly one edge is explored (line 11) and after every iteration of the while-loop the algorithm switches between the search-sides (line 23). Therefore after $2 \cdot O(\textnormal{polylog(n)}) \le n^{o(1)}$ iterations, the non-giant component is fully explored, the condition in line 4 is not fulfilled anymore and the algorithm terminates. 
\end{proof}

Combining Lemmata \ref{lem:runtime-outside-giant}-\ref{lem:runtime-outside-giant-edge}, we directly get the following proposition.

\begin{proposition}\label{prop:runtime-outside-giant}
Let $\mathcal{G}$ be a Chung-Lu graph on $n$ vertices and $s_0, s_1$ two vertices chosen uniformly at random. If at least one of them lies outside the (unique) giant component of $\mathcal{G}$, then any run of $\textnormal{V-BFS}_{approx}(\mathcal{G},s_0,s_1)$, $\textnormal{V-BFS}_{exact}(\mathcal{G},s_0,s_1)$, or $\textnormal{E-BFS}_{approx}(\mathcal{G},s_0,s_1)$ terminates in time $n^{o(1)}$ with high probability. Moreover, if the vertices are chosen adversarially, then the runtime bound continues to hold for $\textnormal{E-BFS}_{approx}(\mathcal{G},s_0,s_1)$.
\end{proposition}

\subsubsection{Applicability of runtime bounds to GIRGs}
\label{subsubsec:applicability}

While our runtime bounds rely in part on intermediate statements, in particular the conditional weight density of vertices found by the BFS from each side given by Lemma~\ref{lem:bfs-new-weight}, which we prove only for Chung-Lu graphs, we remark that they continue to apply for GIRGs. The intuitive reasons for this are the following. First note that for both GIRGs and Chung-Lu graphs, the sum of the degrees of vertices expanded so far is dominated by the vertex of largest degree and this degree is of the same order in GIRGs and Chung-Lu graphs (with the same power-law exponent). This implies in particular that up to constant factors, the number of neighbors at graph distance $d$ in a GIRG is lower-bounded by the number vertices at graph distance $d$ in a Chung-Lu graph (with the same power-law exponent). On the other hand, the number of neighbors at graph distance $d$ in a Chung-Lu graph is also an upper bound since the underlying geometry of GIRGs only makes it more likely to encounter the same vertex more than once during the BFS.
 
Now for the vertex-balanced algorithm, as soon as the algorithm has explored $n^{\frac{\tau-2}{\tau-1}+o(1)}$ edges, it will have expanded many vertices of weight $n^{\frac{\tau-2}{\tau-1}}$ and due to the balancing, this holds on both sides. But in the course of this expansion, the algorithm will have discovered the maximum-degree vertex from both sides whp, and this also holds for GIRGs. This leads to the termination of the algorithm in the vertex-balanced approximate case. 
 
For the edge-balanced algorithm, note again that both for GIRGs and Chung-Lu graphs, the searches on each side (due to the balancing) will quickly discover and start to expand a vertex of weight $n^{\frac{1}{2}+o(1)}$. Once the algorithm has started to expand a vertex of weight close to $n^{\frac{1}{2}+o(1)}$ on both sides, since the neighborhoods of such vertices have a large intersection in GIRGs as well, we will find a common neighbor quickly, in particular in time less than $n^{\frac{1}{2}}$ and the algorithm terminates. This gives the runtime bound for the edge-balanced approximate algorithm. 
 
Finally, for the vertex-balanced exact algorithm, either both queues contain a vertex of weight of order larger than order $n^{\frac{1}{2}+o(1)}$ and then both queues also contain many vertices of weight order $n^{\frac{1}{2}+o(1)}$ and these occur polynomially more often and will therefore be expanded first and lead to the termination of the algorithm. If not, then the queue with smaller maximum is also smaller in size and will be emptied completely, leading to the termination of the algorithm. 

\section{Experimental results}\label{sec:simulations}

We performed some simulations in order to compare the performance of our proposed algorithms with the layer-balanced algorithm analyzed in \cite{borassi2019kadabra}, both on generated networks (Chung-Lu graphs and GIRGs) and on real-world networks. The edge-balanced algorithm is designed to be robust against an adversarial choice of the source and target nodes, however it is not clear what is the best strategy for this adversary, and it is forbiddingly expensive to test all possibilities. This makes it difficult to properly test the edge-balanced algorithm. Hence we focus on the approximate and exact vertex-balanced algorithms (Algorithms \ref{algo:vertex-approx}, abbreviated as VBA, and Algorithm \ref{algo:vertex-exact}, abbreviated as VBE, respectively) in this section.

\begin{figure}
\centering
\begin{subfigure}{0.32\textwidth}
    \includegraphics[width=\textwidth]{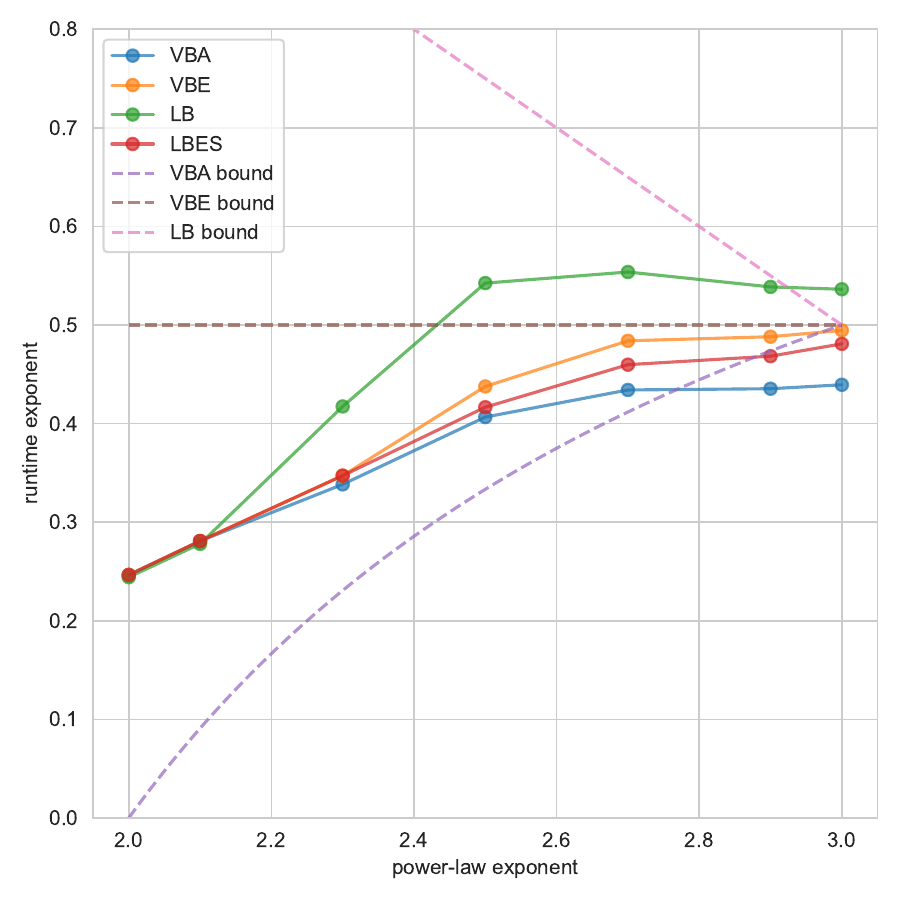}
    \caption{Chung-Lu graphs}
    \label{fig:tau-runtime-exponent-chung-lu}
\end{subfigure}
\begin{subfigure}{0.32\textwidth}
    \includegraphics[width=\textwidth]{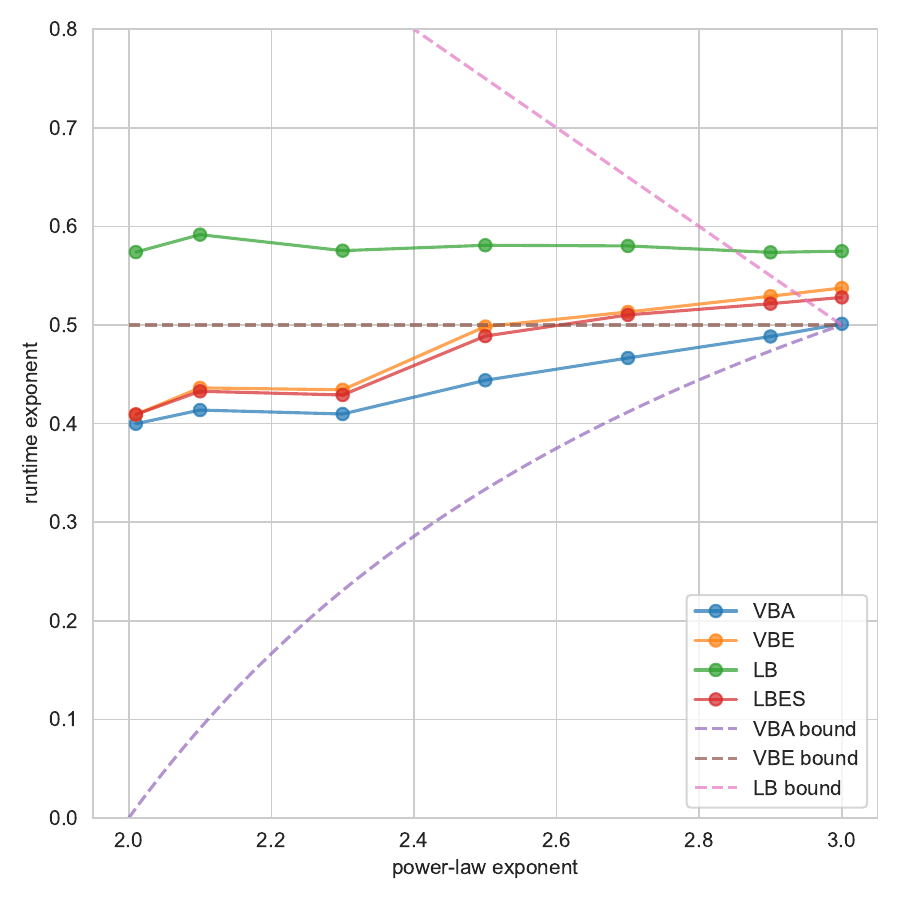}
    \caption{GIRGs with $\alpha=1.5$}
    \label{fig:tau-runtime-exponent-girg1.5}
\end{subfigure}
\begin{subfigure}{0.32\textwidth}
    \includegraphics[width=\textwidth]{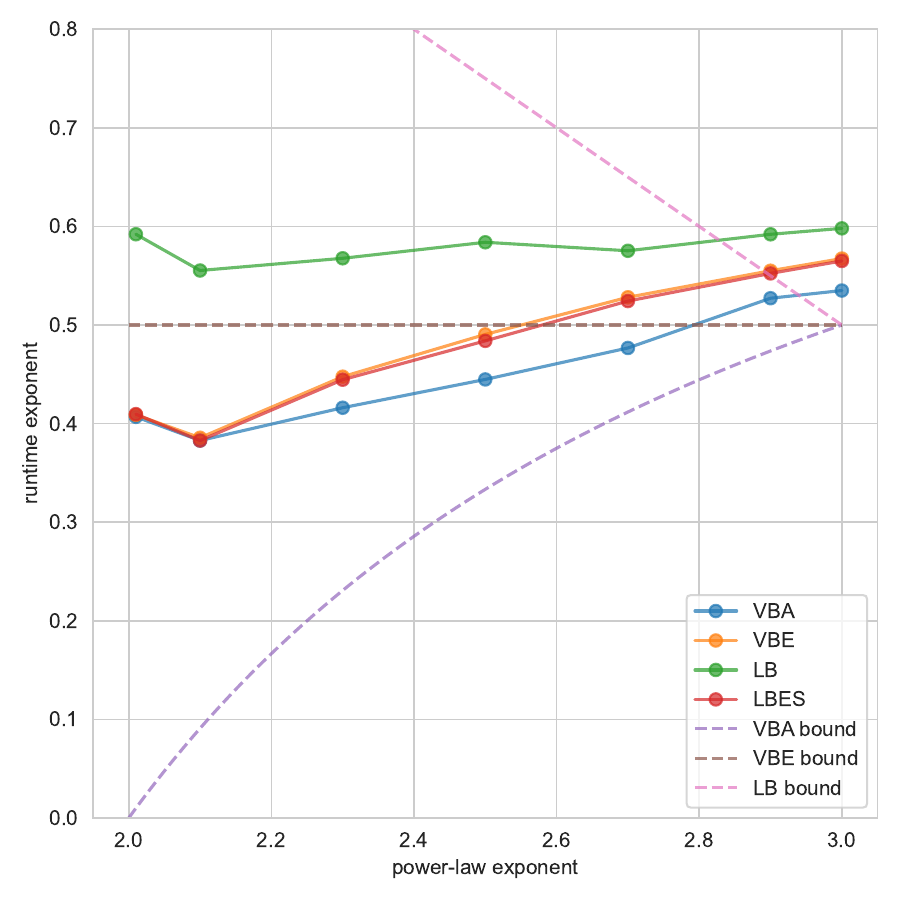}
    \caption{GIRGs with $\alpha=5$}
    \label{fig:tau-runtime-exponent-girg5}
\end{subfigure}

\caption{Plot of the runtime exponent $\rho$ of the cost for the four algorithms VBA, VBE, LB and LBES on three graphs models. For each data point, three different connected graphs with $n\approx80'000$ nodes and $m\approx 1'200'000$ edges were generated, and on each of these graphs the algorithms were run for 100 random pairs of nodes $(s,t)$. The runtime $\mathcal{C}$ is then taken as the median cost over these 300 algorithm runs, and $\rho$ is computed by solving the equation $\mathcal{C}=m^{\rho}$. The theoretical guarantees on the exponent $\rho$ (up to a $o(1)$ term) are also plotted for comparison (as VBA bound, VBE bound, and LB bound). The generated GIRGs have underlying dimension $d=2$. We performed the same experiment for higher dimension $d$ and obtained very similar plots.}
\label{fig:tau-runtime-exponent-plots}
\end{figure}

Figure \ref{fig:tau-runtime-exponent-plots} shows how the exponent of the runtime varies as a function of the power-law exponent $\tau\in(2,3)$ for various algorithms, namely the VBA and VBE algorithms, as well as two versions of the layer-balanced (exact) algorithms. One version is the original layer-balanced bidirectional BFS analyzed in \cite{borassi2019kadabra}, which we refer to as LB, and we also plot the performance of the early-stopping version implemented in \cite{blasius2024external} (abbreviated LBES), where the last layer is not fully expanded, but instead the algorithm expands the layer vertex-by-vertex and terminates as soon as a path is found (which makes it similar to our vertex-balanced algorithms). Unsurprisingly, our VBA algorithm is always the fastest. However, both our VBE algorithm and the LBES algorithm perform very similarly. The advantage of our two vertex-balanced algorithms becomes clear when we compare them with the original LB algorithm, and this gap is wide for Chung-Lu graphs with $\tau\in(2.5,3)$ or for GIRGs with $\tau\in(2,2.5)$. All four algorithms are slightly slower on GIRGs compared to Chung-Lu graphs, and this may be due to the fact that the clustering of vertices in GIRGs makes it more likely that vertices are discovered multiple times, which could slow down the search (by a subpolynomial factor).
Figure \ref{fig:tau-runtime-exponent-plots} also displays the theoretical upper bound guarantees on the runtime exponents, which are $\tfrac{\tau-2}{\tau-1}$ for the VBA algorithm, $\tfrac{1}{2}$ for the VBE algorithm, and $\tfrac{4-\tau}{2}$ for the layer-balanced algorithm(s). Note that the first bound $\tfrac{\tau-2}{\tau-1}$, although an upper bound for the exponent for VBA in the limit $n\to \infty$, lies empirically \emph{below} the corresponding runtime for VBA for large ranges of $\tau$. We believe that this is due to the $n^{o(1)}$ terms, which are still not negligible. Those terms are more relevant when the exponent of the VBA runtime is close to zero, which our theoretical results predict when $\tau$ is close to 2. Hence, our estimate for the exponent is likely unstable for VBA when $\tau$ is close to 2. Nevertheless, our experimental results are in line with all theoretical findings (namely Theorems \ref{thm:vertex-approx-intro}-\ref{thm:vertex-exact-intro} as well as the results proven in \cite{borassi2019kadabra}), although for the VBA algorithm (whose runtime is upper bounded by $n^{(\tau-2)/(\tau-1)+o(1)}$ whp) the $n^{o(1)}$ term is clearly visible, especially in GIRGs. 

\begin{figure}
    \centering
    \subfloat{\includegraphics[width=0.33\textwidth]{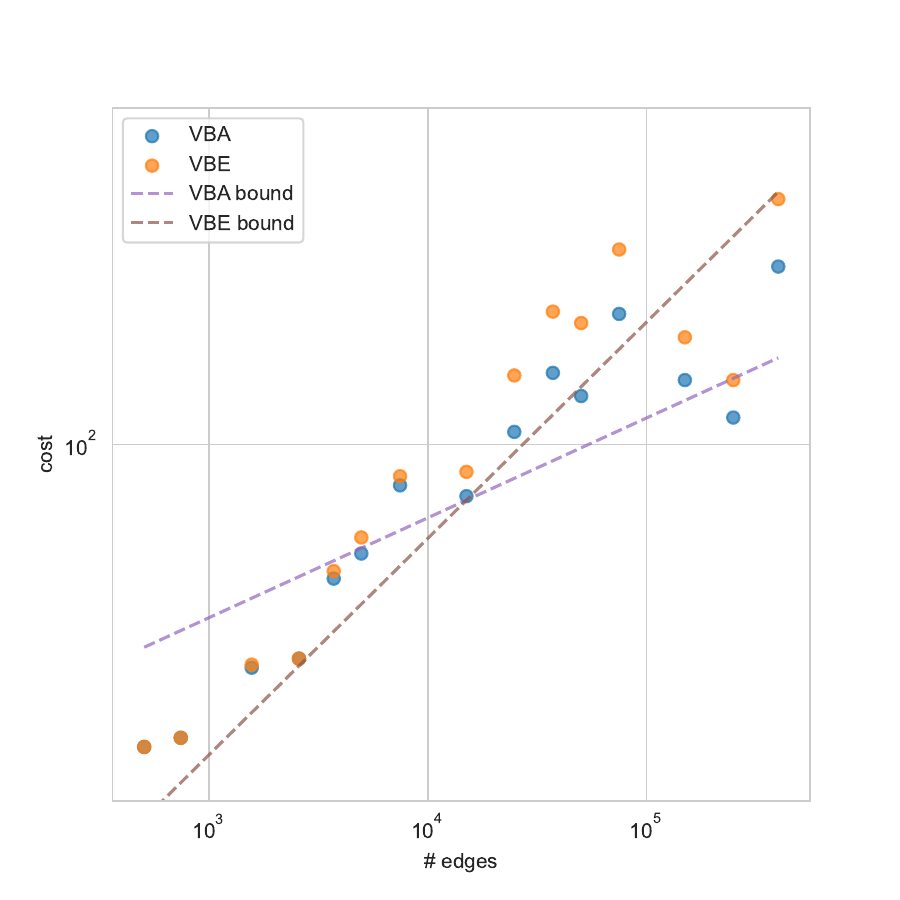}}
    \subfloat{\includegraphics[width=0.33\textwidth]{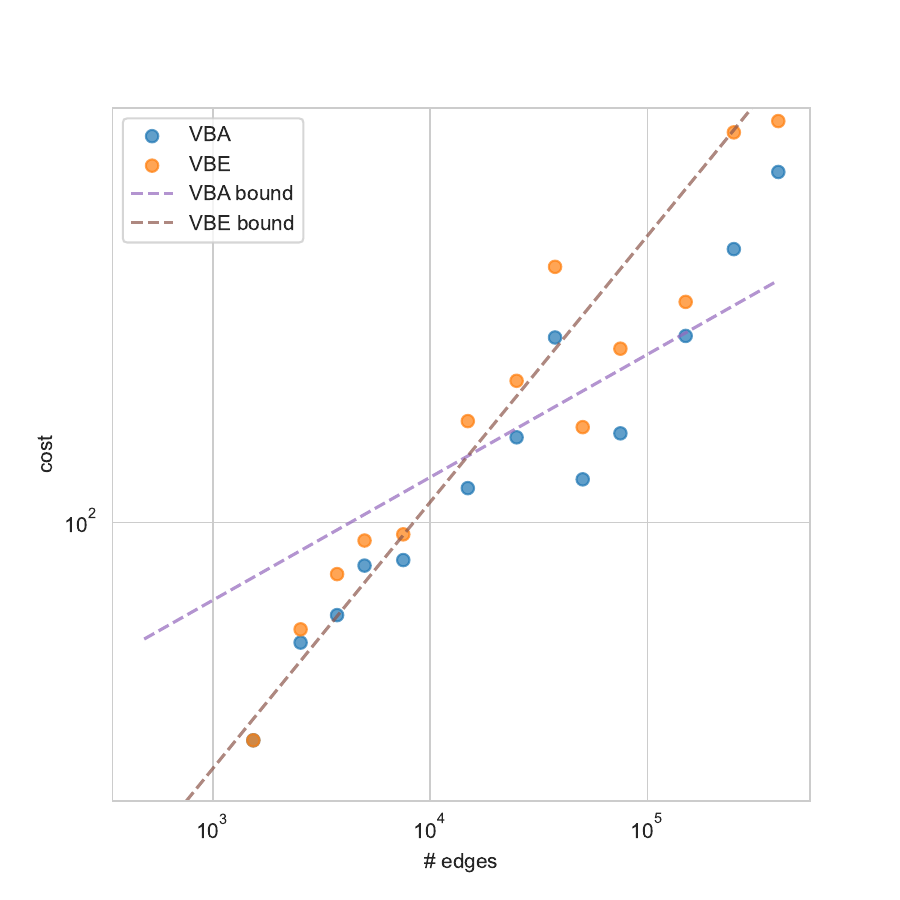}}
    \subfloat{\includegraphics[width=0.33\textwidth]{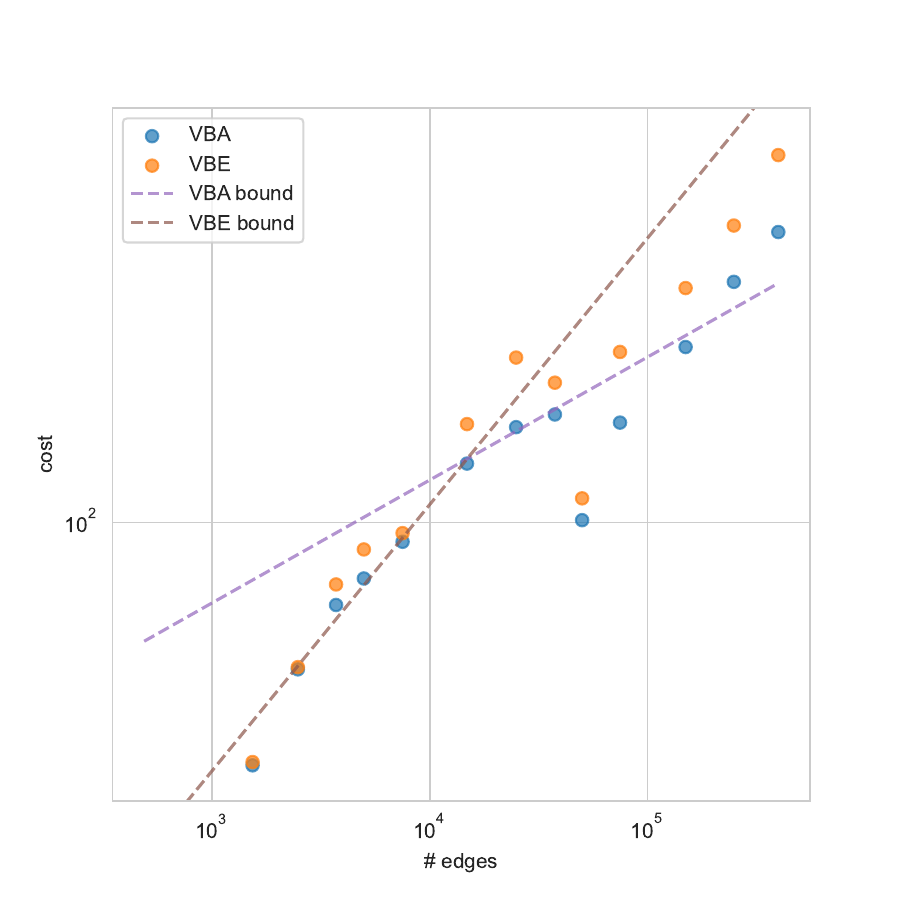}}
    \\
    \subfloat{\includegraphics[width=0.33\textwidth]{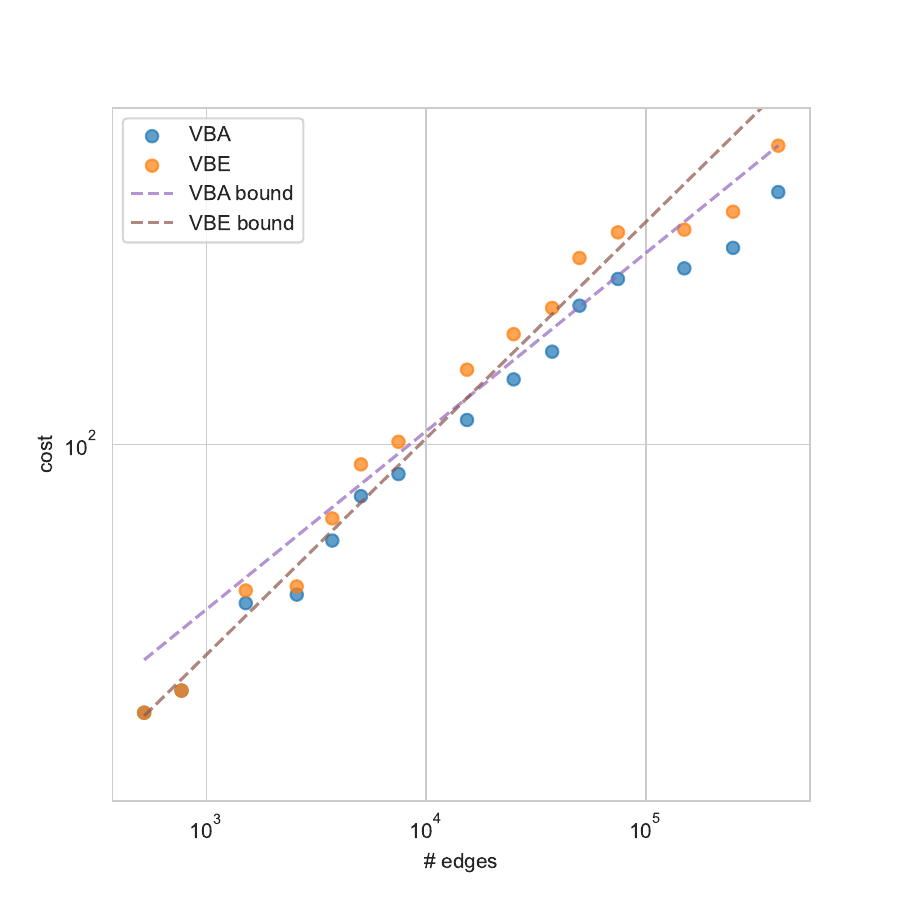}}
    \subfloat{\includegraphics[width=0.33\textwidth]{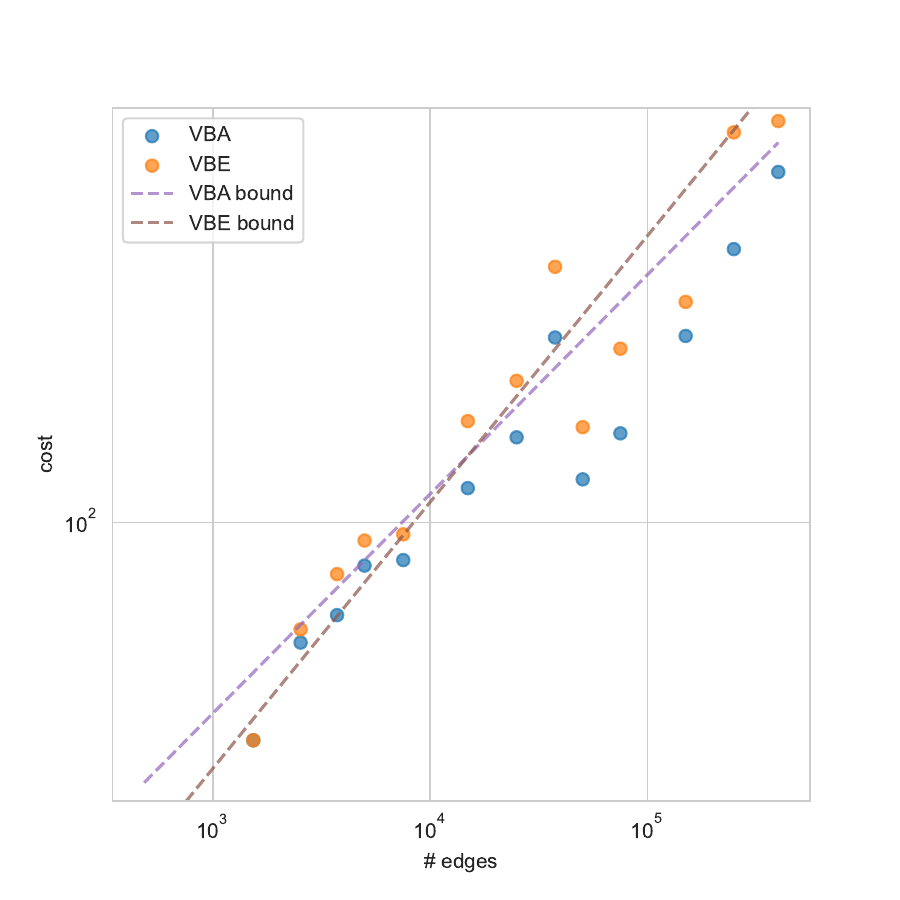}}
    \subfloat{\includegraphics[width=0.33\textwidth]{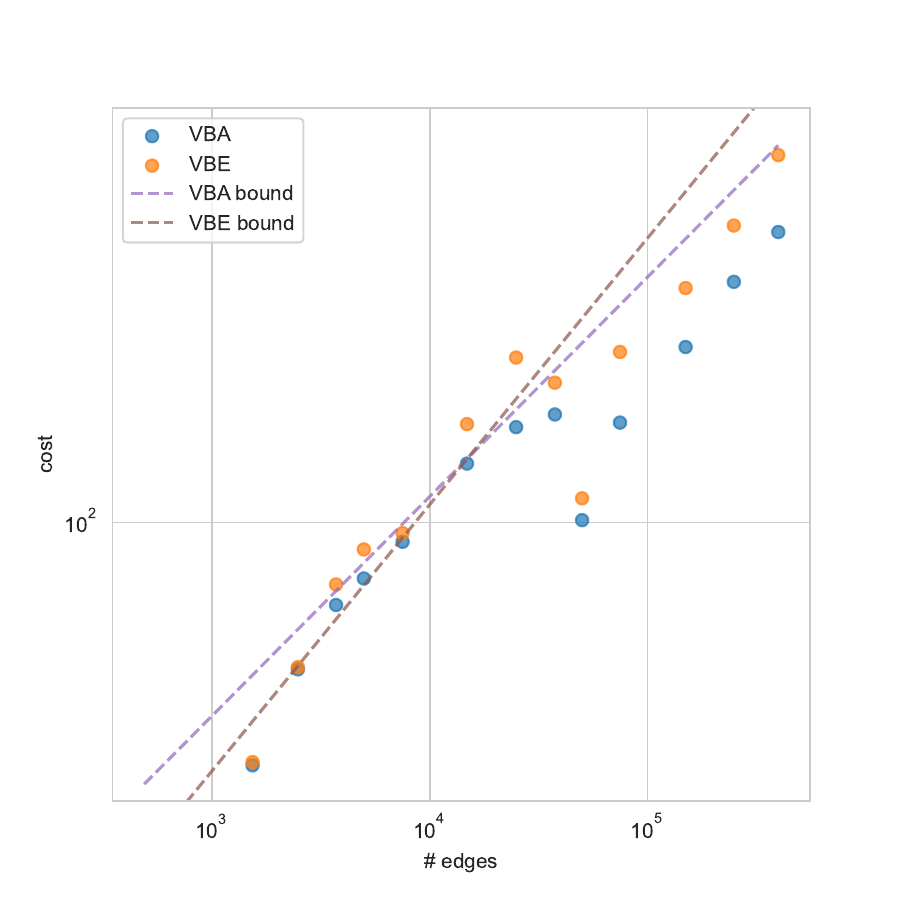}}
    \caption{Log-log plots of the runtime (or cost) of the VBA and VBE algorithms compared to the number of edges in the graph. The two leftmost subfigures show Chung-Lu graphs, the two middle subfigures GIRGs with $\alpha=1.5$, and the  two rightmost subfigures GIRGs with $\alpha=5$. In the top row all graphs are generated with power-law parameter $\tau=2.3$ and in the bottom row with $\tau=2.7$. All graphs have average degree 10, and all GIRGs have dimension $d=2$. The dashed lines (VBA bound and VBE bound) indicate the theoretical bounds from Theorem \ref{thm:vertex-approx-intro}-\ref{thm:vertex-exact-intro}, and are shifted up for ease of comparison. Each datapoint is the median value obtained from first sampling three graph instances with the same parameters and aggregating the runs over 100 random pairs of nodes $(s,t)$ on each graph instance.
    }
    \label{fig:log-edge-log-cost-plots}
\end{figure}

In the log-log plots of Figure~\ref{fig:log-edge-log-cost-plots}, we focus on the vertex-balanced algorithms VBA and VBE and display how the runtimes scale with increasing graph size (measured in the number of edges) for Chung-Lu graphs and GIRGs. This is then compared to our theoretical upper bounds (Theorems \ref{thm:vertex-approx-intro}-\ref{thm:vertex-exact-intro}), which are displayed as (shifted) dashed lines whose slopes yield the bound on the runtime exponent $\rho$. For the VBE algorithm, this bound on the runtime exponent is $\rho = \frac{1}{2}$, while for the VBA algorithm it is $\rho = \frac{\tau-2}{\tau-1}$, which for the generated graphs with $\tau=2.3$ is $\rho \approx 0.23$ and for $\tau=2.7$ is $\rho \approx 0.41$. Across parameter configurations (and graph models), the experimental runtimes are consistent with our theoretical predictions. Also here, we see that the VBA tends to be faster than the VBE, though the difference is not as stark as in Figure~\ref{fig:tau-runtime-exponent-plots}. Remarkably, for smaller graphs (up to ca.\ 300 edges), the runtimes of the VBA and the VBE coincide almost perfectly.

We also analyze the performance of the two vertex-balanced algorithms, as well as the two versions of the layer-balanced algorithm on generated networks juxtaposed to real-world networks (we examine a subset of 2740 networks from the Network Repository \cite{rossi2015network}). To this avail, we used the approach proposed in \cite{blasius2024external} by Bläsius and Fischbeck. The key idea is to define a (degree) \emph{heterogeneity} measure and a \emph{locality} measure. This allows to quantify similarity between real-world networks and generated networks, with the hope that the BFS algorithms perform similarly on real-world networks and generated networks of comparable heterogeneity and locality. We refer the reader to \cite{blasius2024external} for the definitions of these heterogeneity and locality measures, as well as the network data selection and the parameter choices for the generated networks, see also Figure \ref{fig:generated-graph-params} for an illustration of the heterogeneity and locality measures of our generated networks.

\begin{figure}
    \centering
    \includegraphics[width=0.7\textwidth]{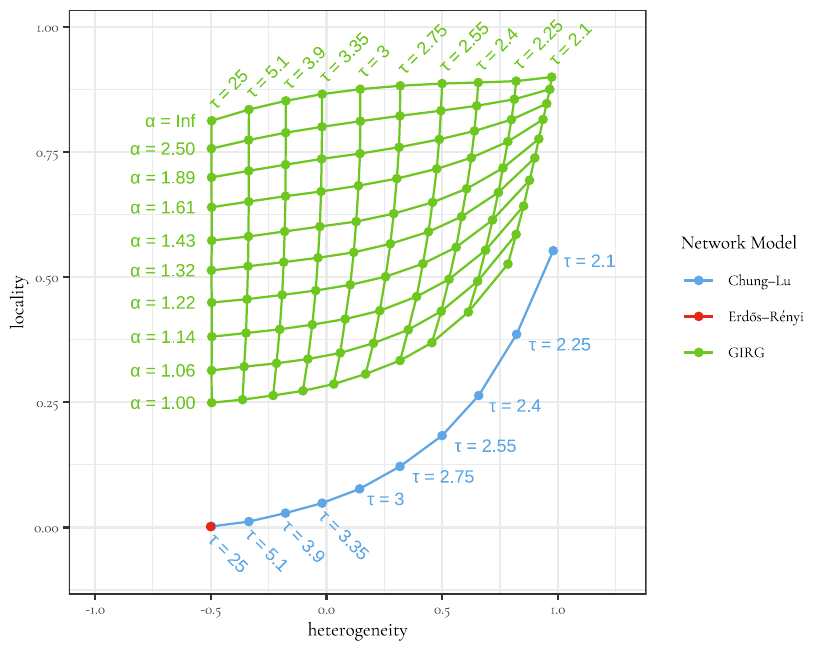}
    \caption{Heterogeneity and locality of the generated networks from the different models. Each point is the average of five samples with the given parameter configuration. This is a slightly adapted version of Figure 3 in \cite{blasius2024external}.}
    \label{fig:generated-graph-params}
\end{figure}

\begin{figure}
\centering
\begin{subfigure}{0.7\textwidth}
    \includegraphics[width=\textwidth]{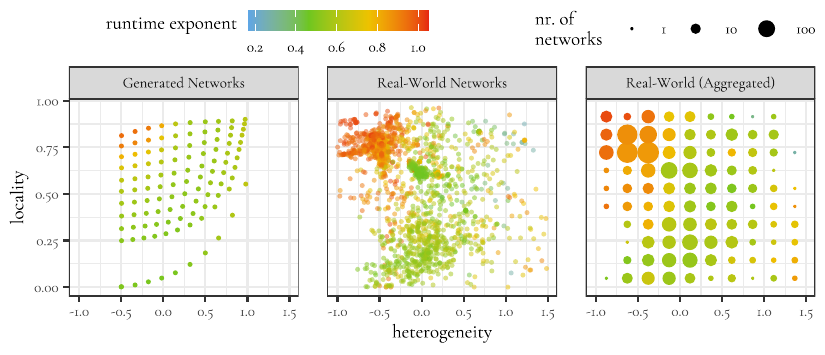}
    \caption{Vertex-balanced approximate algorithm (VBA).}
    \label{fig:real-generated-comparison-vba}
\end{subfigure}
\begin{subfigure}{0.7\textwidth}
    \includegraphics[width=\textwidth]{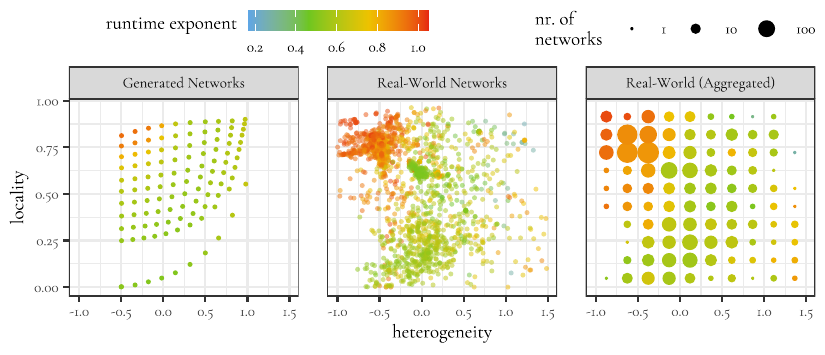}
    \caption{Vertex-balanced exact algorithm (VBE).}
    \label{fig:real-generated-comparison-vbe}
\end{subfigure}
\begin{subfigure}{0.7\textwidth}
    \includegraphics[width=\textwidth]{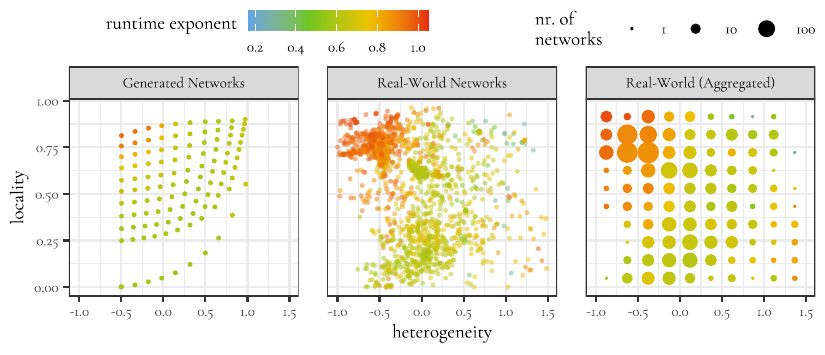}
    \caption{Layer-balanced algorithm without early stopping (LB).}
    \label{fig:real-generated-comparison-lbes}
\end{subfigure}
\begin{subfigure}{0.7\textwidth}
    \includegraphics[width=\textwidth]{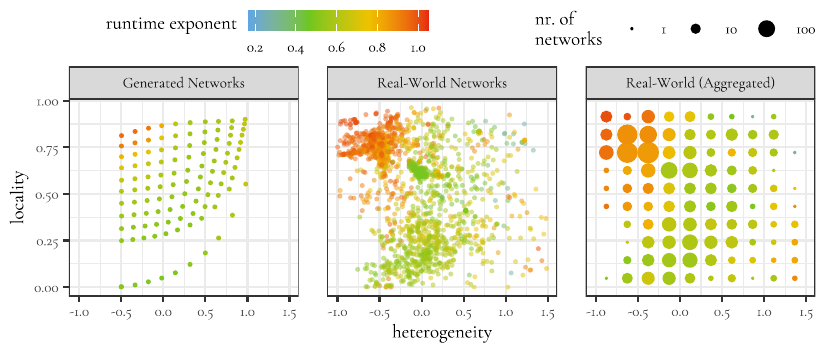}
    \caption{Layer-balanced algorithm with early stopping (LBES).}
    \label{fig:real-generated-comparison-lb}
\end{subfigure}
\caption{The runtime exponent $\rho$ of the cost $\mathcal{C}=m^{\rho}$ of the algorithms VBA (\ref{fig:real-generated-comparison-vba}), VBE (\ref{fig:real-generated-comparison-vbe}), LB (\ref{fig:real-generated-comparison-lb}) and LBES (\ref{fig:real-generated-comparison-lbes}), averaged over 100 $st$-pairs, and plotted against the networks' heterogeneity and locality. }
\label{fig:real-generated-comparison}
\end{figure}

Figure \ref{fig:real-generated-comparison} shows the results of this performance comparison between generated networks and real-world networks. In all four plots, we see qualitatively that generated networks are good proxies for estimating the algorithms' runtime on real-world networks (of similar heterogeneity and locality), with all algorithms being slowest on networks of high locality and low heterogeneity. Note that for the range of power-law exponents $\tau$ that we consider in our proofs, the runtime in the generated networks, displayed in the leftmost plot for each algorithm in Figure \ref{fig:real-generated-comparison}, does not exhibit any runtime dependence on the locality. The orange and red datapoints (corresponding to runtime exponents $\rho$ larger than 0.7) belong to network instances with $\tau>3$ (cf.\ Figure \ref{fig:generated-graph-params}). As already observed in Figure \ref{fig:tau-runtime-exponent-plots}, the VBA, VBE, and LBES algorithms have very similar running times on generated networks, and this fairly surprising insight also transfers to real-world networks. As expected, the LB algorithm is slower than the other three algorithms both on generated and real-world networks. Note that whenever the goal is to simulate the runtimes on surrogate graphs tailored to specific real networks, it is crucial to fit the model parameters to the real network. For the case of GIRGs, we refer the reader to \cite{dayan2024expressivity}.

Finally, our theoretical results make the surprising prediction that the size of both bidirectional search trees is dominated by the largest expanded degree of a single vertex. In Figure~\ref{fig:costfrac} we thus plot the ratio of the degree of the highest-degree vertex which is expanded throughout the execution to the total cost (i.e.\ the sum of the degrees of all expanded vertices) of the four algorithms VBA, VBE, LB and LBES. While the ratio is highest for the vertex-balanced algorithms, also for the layer-balanced algorithms we see that both for Chung-Lu graphs and GIRGs, and for the entire range of $\tau\in(2,3)$, the expanded vertex of maximal degree has a median contribution on the order of a constant fraction of the cost (sometimes accounting for more than half of the total cost!). We remark here that we conducted the same experiments with the ratio of the degree of the final expanded vertex to the total cost. This yielded very similar plots, with values of the same order of magnitude - indicating that the final expanded vertex  has a degree of similar order as the expanded vertex of maximum degree or even coincides with this vertex.

\begin{figure}

    \centering
   
    \begin{subfigure}{0.32\textwidth}
    \includegraphics[width=\textwidth]{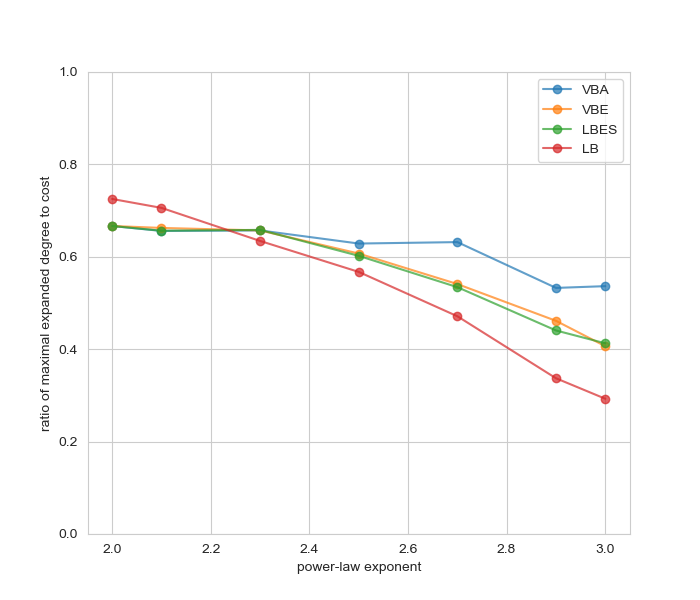}
    \caption{Chung-Lu graphs}
    \label{fig:degreecostratio_cl}
    \end{subfigure}
    \begin{subfigure}{0.32\textwidth}
    \includegraphics[width=\textwidth]{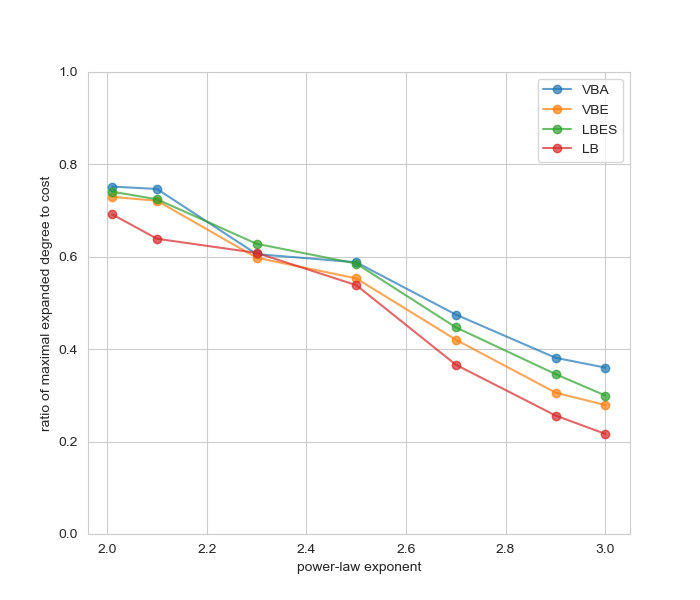}
    \caption{GIRGs with $\alpha=1.5$, $d=2$}
    \label{fig:degreecostratio_girgs_1.5}
    \end{subfigure}
    \begin{subfigure}{0.32\textwidth}
    \includegraphics[width=\textwidth]{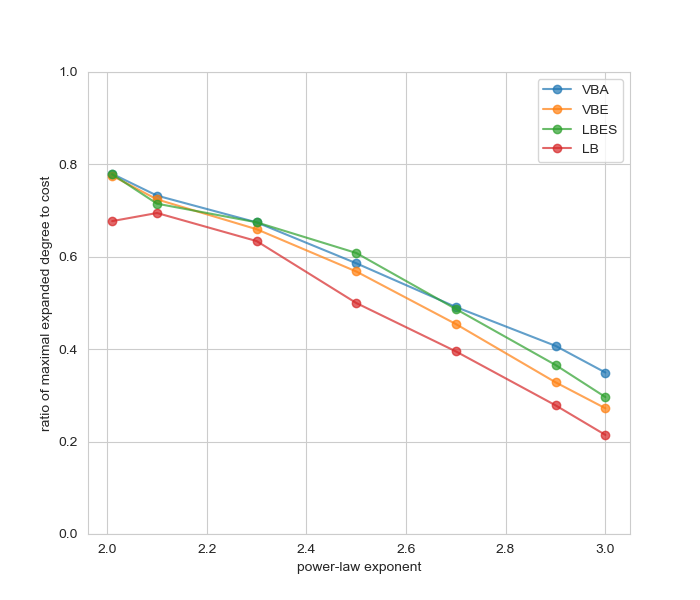}
    \caption{GIRGs with $\alpha=5$, $d=2$}
    \label{fig:degreecostratio_girgs_5.0}
    \end{subfigure}
    
    \caption{Ratio of  maximum degree among the expanded vertices vs.\ cost (i.e.\ sum of degrees of all expanded vertices). For each data point, three different connected graphs with $n\approx80'000$ nodes and $m\approx 1'200'000$ edges were generated, and on each of these graphs the algorithms were run for 100 random pairs of nodes $(s,t)$. The plots show the median value over these 300 runs.}
    \label{fig:costfrac}   
\end{figure}

\paragraph{Acknowledgments}
We thank Familie Türtscher for their hospitality during the research retreat in Buchboden in July 2023 where this project was started.

\bibliographystyle{abbrv}	
\bibliography{references}

\end{document}